\documentclass{statsoc}

\usepackage{graphicx}
\usepackage{enumerate}
\usepackage{amsmath,amssymb}               
\usepackage[a4paper]{geometry}
\usepackage{graphics}
\usepackage{soul,color}
\usepackage{natbib}
\usepackage{rotating}
\usepackage{cancel}
\usepackage{xcolor}
\usepackage{subfigure}
\usepackage{multirow}
\usepackage[normalem]{ulem}
\newtheorem{proposition}{Proposition}

\newtheorem{lemma}{Lemma} 
\usepackage{DNstyle}

\newcommand\indep{\protect\mathpalette{\protect\independenT}{\perp}}
\def\independenT#1#2{\mathrel{\rlap{$#1#2$}\mkern2mu{#1#2}}}

\title[Matching methods for truncation by death problems]{Matching methods for truncation by death problems}
\author[Tamir Zehavi and Daniel Nevo]{Tamir Zehavi and Daniel Nevo}
\address{Tel Aviv University,
Tel Aviv, Israel.}
\email{danielnevo@gmail.com}

\begin{document}

%%%%%%%%%%%%%%%%%%%%%%%%%%%%%%%%%%%%%%%%%%%%%%%%%%%%%%%%%%%%%%%%%
%%% Abstract
\begin{abstract}
Even in a carefully designed randomized trial, outcomes for some study participants can be missing, or more precisely, ill-defined, because participants had died prior to date of outcome collection. This problem, known as truncation by death, means that the treated and untreated are no longer balanced with respect to covariates determining survival. Therefore, researchers often utilize principal stratification and focus on the Survivor Average Causal Effect (SACE). The SACE is the average causal effect among the subpopulation that will survive regardless of treatment status. In this paper, we present matching-based methods for SACE identification and estimation. We provide an identification result for the SACE that motivates the use of matching to restore the balance among the survivors. We discuss various practical issues, including the choice of distance measures, possibility of matching with replacement, and post-matching crude and model-based SACE estimators.  Simulation studies and data analysis demonstrate the flexibility of our approach. Because the cross-world assumptions needed for SACE identification can be too strong, we also present sensitivity analysis techniques and illustrate their use in real data analysis. Finally, we show how our approach can also be utilized to estimate conditional separable effects, a recently-proposed alternative for the SACE.
\end{abstract}
				

%%%%%%%%%%%%%%%%%%%%%%%%%%%%%%%%%%%%%%%%%%%%%%%%%%%%%%%%%%%%%%%%%
\section{Introduction}
\label{Sec:Intro}

Randomized controlled trials (RCTs) are considered by many the gold standard for  estimating causal effects of treatments or interventions on an outcome of interest. However, even in a carefully designed trial, outcomes for some study participants can be missing, or more precisely,  ill-defined, because participants had died prior to date of outcome collection. This problem is known as ``truncation by death'' \citep{zhang2003estimation,hayden2005estimator,lee2009training, ding2011identifiability,ding2017principal}.

Examples of truncation by death in practice are ubiquitous:  in clinical studies, when the outcome is Quality Of Life (QOL) index one year after initiating treatment for a terminal disease, but some study participants die earlier \citep{ding2017principal,stensrud2022conditional}; in labor economics, when studying the effect of a  training program  on wage, but some remain unemployed (in this example, unemployment is the analogue of death) \citep{lee2009training,zhang2009likelihood}; in education studies, different school programs are evaluated with respect to students' achievements, but some students may drop out  \citep{garcia2010impact}. 

In the presence of  truncation by death,  a naive comparison of the outcomes between the treated and untreated groups among the survivors does not correspond to a causal effect. Intuitively, those survived with the treatment might have died had they did not receive the treatment, hence the two treatment groups  are not comparable, and are unbalanced with respect to covariates determining survival. %If those covariates are also affecting the non-survival outcome, selection bias is expected.
An alternative composite outcome approach combines the two outcome types into one, often by  setting the outcome to zero for those not surviving. This approach estimates a well-defined causal effect, but it entangles treatment effect  on survival with  treatment effect  on the outcome. 

To overcome the lack of causal interpretation for the comparison among the survivors, or the partial information gained by the composite outcome  approach, researchers have utilized principal stratification \citep{robins1986new,frangakis2002principal}  to focus on a subpopulation among which a causal effect is well-defined. The Survivor Average Causal Effect (SACE) is the effect within the subpopulation who would survive regardless of treatment assignment; this subpopulation is often termed the \textit{always-survivors}.    

Identification and estimation of the SACE possess a challenge even in a RCT, as the standard assumptions of randomization and stable unit treatment value assumption (SUTVA) do not guarantee identification from the observed data distribution.   As a result, there is a built-in tension between the strength of the assumptions researchers are willing to make and the degree to which causal effects can be identified from the observed data.  Ordered by the increasing strength of the assumptions underpinning them, three main approaches have been considered for the SACE: constructing bounds \citep{zhang2003estimation,lee2009training,nevo2021causal}, conducting sensitivity analyses \citep{hayden2005estimator,chiba2011simple,ding2011identifiability, chiba2012estimation, ding2017principal,nevo2021causal} and leveraging additional assumptions for full identification \citep{hayden2005estimator,zhang2009likelihood, chiba2011identification, ding2011identifiability,ding2017principal,feller2017principal,wang2017identification}.

The practical use of SACE has been criticized because the underlying subpopulation for which the SACE is relevant can be an irregular subset of the population and  because identification of the SACE relies on unfalsifiable assumptions \citep{robins1986new,stensrud2022conditional}.  The approach of \textit{conditional separable effects} (CSEs) \citep{stensrud2022separable}, which is applicable to a setting more general than truncation by death, can be used as an alternative to the SACE. In the truncation by death setting, CSEs consider a mechanistic separation of the treatment to the component affecting survival and the component affecting the outcome. Interestingly, these effects can be identified even in a trial that only includes the original treatment, without its mechanistic separation.

Nevertheless, in practice, while researchers across a variety of disciplines increasingly recognize the invalidity of the aforementioned naive approaches \citep{mcconnell2008truncation,colantuoni2018statistical},  the novel methods and theory developed by causal inference researchers for truncation by death are not yet fully implemented in practice. Therefore,  methods that are intuitive, simple to use, and with clear practical guidelines are of need.% \citep{ICH}. %The recent ICH E9 addendum on estimands and sensitivity analysis in clinical trials  \citep{ICH} highlighted the importance of considering intercurrent events, including truncation by death, when designing and analyzing RCTs. These considerations include, among other issues, the target estimand for defining treatment effect,  and a laid out sensitivity analysis approach for the different assumptions.  

The goal of this paper is to introduce matching procedures to analyses targeting the SACE or the CSEs.  Traditionally, matching has been used to adjust for pre-treatment differences in observed covariates between the treated and untreated groups in observational studies \citep{stuart2010matching,rosenbaum2011observational,dehejia1999causal,dehejia2002propensity,ho2007matching}.  Here, we propose to use matching to restore the balance originally achieved by randomization, and broken by the differential survival.% We show that careful consideration of the  variables matched upon leads to a flexible framework for SACE.

With this in mind, the main contributions of the paper are as follows. First, we review and clarify SACE identification assumptions previously considered in the literature and show that weaker assumptions can be used to obtain identification of the SACE. Second, we explain how the obtained observed data functional motivates the use of matching methods. Third, we discuss how matching-based estimation can be used in practice to target the SACE or the CSEs, and discuss for the former how to tailor sensitivity analysis techniques for the presented assumptions when matching is used. Fourth, we demonstrate key ideas for SACE estimation and inference via simulations studies and data analysis.

The rest of the paper is organised as follows. Section \ref{Sec:Prelim} presents notations and the SACE causal estimand. Section \ref{Sec:SACEident} considers assumptions and the SACE identification result motivating the use of matching as described in detail in Section \ref{Sec:Matching}. Section \ref{Sec:Sens} discusses how to tailor sensitivity analysis approaches for the assumptions, when matching is used, and in Section \ref{Sec:CSEs}, we show how matching methods can be adapted to the recently-proposed CSEs.  In Section \ref{Sec:Sims}, we present a simulation study, and in Section \ref{Sec:Data}, we  illustrate the utility of matching for SACE estimation in a real data analysis. Concluding remarks are offered in Section \ref{Sec:Discuss}. Code and data analysis are available from  \texttt{https://github.com/TamirZe/Matching-methods-for-truncation-by-death-problems}.

%%%%%%%%%%%%%%%%%%%%%%%%%%%%%%%%%%%%%%%%%%%%%%%%%%%%%%%%%%%%%%%%%%%%%%%%%%%
\section{Preliminaries}
\label{Sec:Prelim}

Using the potential outcomes framework, for each individual $i$ in the population, we let $S_i(a)$ and $Y_i(a)$ be the survival status and the outcome, respectively, under treatment value $a=0,1$, at a pre-specified time after treatment assignment (e.g., one year). The outcome $Y_i(a)$ takes values in $\mathbb{R}$ if $S_i(a)=1$. Truncation by death is introduced by setting $Y_i(a)=^*$ whenever $S_i(a)=0$.  Throughout the paper, we assume SUTVA, namely that there is no interference between units and no multiple versions of the treatment leading to different outcomes. 

With these notations, one can define the following four principal strata \citep{frangakis2002principal}: \textit{always-survivors} ($as$) $\{i: S_i(0)=1, S_i(1)=1\}$,   \textit{protected} ($pro$) $\{i: S_i(0)=0, S_i(1)=1\}$,  \textit{harmed} ($har$) $\{i: S_i(0)=1, S_i(1)=0\}$, and  \textit{never-survivors} ($ns$) $\{i: S_i(0)=0, S_i(1)=0\}$. Let also $G_i$ represent the stratum of unit $i$ and $\pi_g=\Pr(G=g), g \in \{as,pro,har,ns\}$, be the stratum proportions. Because the  causal contrast $Y_i(1)-Y_i(0)$  is well-defined only among the always-survivors, researchers often set the SACE,
$$
\text{SACE} = E[Y(1)-Y(0)|S(0)=S(1)=1],
$$
as the target causal parameter \citep{rubin2006causal}.
 
For a sample of size $n$ from the population, the observed data for each individual $i=1,...,n$ is $(\bX_{0i},A_i,\bX_{1i,}S_i,Y_i)$, where $\bX_{0i}$ is a vector of covariates that are either pre-treatment or unaffected by the treatment; $A_i$ is the treatment indicator ($A_i=1$ for a  treated individual, and $A_i=0$ for an untreated individual);  $\bX_{1i}$ is a vector of post-treatment covariates, that are affected by $A_i$; $S_i$ is the survival status; $Y_i$ is the outcome of interest. From SUTVA it follows that $S_i=A_iS_i(1) + (1-A_i)S_i(0)$, and $Y_i=A_iY_i(1) + (1-A_i)Y_i(0)$. In addition to SUTVA, we assume that treatment was randomly assigned. 
\begin{assumption}
	Randomization $A\indep \{S(a),Y(a),\bX_{0}\}$ for $a=0,1$.
\end{assumption}

SUTVA and randomization suffice to identify the causal effect of $A$ on survival $S$, but not the SACE or the CSEs (defined in Section \ref{Sec:CSEs}). The covariates $\bX_{0i}$ and $\bX_{1i}$ play a key role in the identification of the SACE or the CSEs. Assumptions on these sets of covariates are necessary, because even though treatment was  randomized, truncation by death  forces researchers to consider shared causes of the variables $S$ and $Y$, and their relation to the treatment $A$.

\begin{table}
\caption{\label{Tab:ObsDataStrat} Relationship between the observed data and the strata. Each cell gives the possible stratum membership of a unit observed with $A=a$ and $S=s$, with (right) and without (left) the monotonicity assumption.}
\centering
\fbox{%
\begin{tabular}{*{3}{c}}
& $S=0$ & $S=1$ \\
\hline
$A=0$ & \textit{ns}, \textit{pro} & \textit{as}, \textit{har}\\
$A=1$ & \textit{ns}, \textit{har} & \textit{as}, \textit{pro} \\
\end{tabular}} 
\hspace{1.5cm}
 \fbox{
 \begin{tabular}{*{3}{c}} 
 & $S=0$ & $S=1$\\
 \hline
 $A=0$ & \textit{ns}, \textit{pro} & \textit{as}\\
 $A=1$ & \textit{ns} & \textit{as}, \textit{pro} \\
 \end{tabular}}
 \end{table}

%%%%%%%%%%%%%%%%%%%%%%%%%%%%%%%%%%%%%%%%%%%%%%%%%%%%%%%%%%%%
\section{SACE identification}
\label{Sec:SACEident}

Because the SACE  is not point-identified  from the observed data under SUTVA and randomization, additional assumptions have been considered.  Possibly the most common assumption, which seems plausible in certain studies, is monotonicity \citep{zhang2003estimation,chiba2011identification,ding2017principal,feller2017principal,wang2017identification}.
\begin{assumption}
	Monotonicity. $S_i(0) \le S_i(1)$ for all $i$.
\end{assumption}
Monotonicity means that treatment could not cause death, that is, the harmed stratum is empty. Under monotonicity, SUTVA, and randomization, the proportion of always-survivors is identifiable from the data by $\Pr[S(0)=1, S(1)=1]=\Pr(S=1|A=0)$. This proportion is crucial in practice, because it speaks to the proportion of the population for which the SACE is relevant. Table \ref{Tab:ObsDataStrat} presents the possible stratum membership for each unit according to its observed $A,S$ values, with and without the monotonicity assumption.

While the assertion that treatment cannot hurt survival may be reasonable in some studies, the monotonicity assumption is often criticized for being too strong \citep{ding2011identifiability,yang2018using}. To describe a weaker assumption, consider the \textit{principal score} $\pi_g(\bx_0)=\Pr(G=g|\bX_0=\bx_0)$, where $g\in \{as,ns,pro,har\}$ \citep{ding2017principal,feller2017principal}. For example, the probability of being an always-survivor given the covariates $\bX_0$ is $\pi_{as}(\bx_0)=\Pr[S(1)=S(0)=1|\bX_0=\bx_0]$. We now ready to introduce the constant principal score ratio (CPSR) assumption.
\begin{assumption} 
Constant principal score ratio (CPSR). For all possible  $\bx_0$ values,
 $\xi = \frac{\pi_{har}(\bx_0)}{\pi_{as}(\bx_0)}$ for some possibly-unknown $\xi$.
\label{Ass:CPSR} 
\end{assumption}
The CPSR assumption asserts that the ratio  between the proportions of harmed and always-survivors is constant at all levels of $\bX_0$. It is weaker than the monotonicity assumption, under which $\xi=0$. An assumption similar to CPSR was previously considered by \cite{ding2017principal} to construct a sensitivity analysis for monotonicity.

To illustrate the interpretation of the two assumptions, consider the following data generating mechanism (DGM). Given a vector of covariates $\bZ_i$ (that may include an intercept), assume that the potential survival status under treatment level $a=0$ follows the logistic regression
\begin{equation}
\label{Eq:SeqRegModelS0}
\Pr[S_i(0)=1|\bZ_i] = \frac{\exp(\bgamma_{S(0)}^T \bZ_i)} 
{1 + \exp(\bgamma_{S(0)}^T \bZ_i)}.
\end{equation}
Then, among those who would not have survived under $a=0$, the potential survival status under treatment level $a=1$ follows the logistic regression
\begin{equation}
\label{Eq:SeqRegModelS1_given_S0=0}
\Pr[S_i(1)=1|S_i(0)=0,\bZ_i] = \frac{\exp(\bgamma_{S(1)}^T \bZ_i)}{1 + \exp(\bgamma_{S(1)}^T \bZ_i)},
\end{equation}
and among those that would have survived under $a=0$, $\Pr[S_i(1)=1|S_i(0)=1, \bZ_i]=\Pr[S_i(1)=1|S_i(0)=1]$ equals to one under monotonicity and to $\frac{1}{1 + \xi}$ under CPSR. We henceforth call this DGM the ``sequential model''. The monotonicity and CPSR assumptions assert that any heterogeneity in survival under treatment, among those who would have survived under no treatment, is not captured by the covariates associated with survival status when untreated (i.e., $\bZ_i$). Monotonicity further constraints that all of this subset will survive, while CPSR ``allows'' for a proportion $\frac{\xi}{1 + \xi}$ of those to die. This would be the case, for example, in a trial studying treatment for lung cancer, with death under treatment being due to an adverse event, as a result of treatment damaging a different organ. If this adverse effect is only related to a genetic marker unrelated to lung cancer, it may be that this marker is independent of $\bZ_i$ and hence $\Pr[S_i(1)=1|S_i(0)=1, \bZ_i]=\Pr[S_i(1)=1|S_i(0)=1]$ is independent of $\bZ_i$.  An alternative similar DGM is a multinomial regression for the strata with three categories: protected, never-survivors, and always-survivors or harmed (Section B of the Web Appendix). Then, monotonicity or CPSR constraint the proportion of harmed within the always-survivors or harmed group. We focus mainly on the sequential model because we find its interpretation clearer with respect to the causal assumptions.

We turn to assumptions on the relationship between the potential survival and non-survival outcomes.  The first assumption is strongly related to the assumptions known as Principal Ignorability and General Principal Ignorability \citep{jo2009use,feller2017principal,ding2017principal}.  These assumptions, and their variants, state conditional independence  between the potential non-survival outcomes $Y(a)$ and the principal stratum membership, conditionally on covariates.  In this section, we consider an assumption we term partial principal  ignorability (PPI) stating that for study participants who would have survived under treatment, conditionally on $\bX_0$, the outcome under treatment  is not informative about the survival status had those participants were untreated.
\begin{assumption}
Partial Principal Ignorability (PPI). $Y(1)  \indep  S(0) | S(1)=1, \bX_0$.
\end{assumption}
The PPI assumption was considered by \cite{chiba2011identification}, who also showed that the SACE is identifiable under PPI and monotonicity. 
A slightly stronger assumption is Strong Partial Principal Ignorability (SPPI).
\begin{assumption}
Strong Partial Principal Ignorability (SPPI). $Y(a) \indep S(1-a)|S(a)=1,\bX_0$ for $a=0,1$.
\end{assumption}
\cite{hayden2005estimator} combined SPPI with a strong survival status independence assumption $S(0)\indep S(1)|\bX_0$ to obtain identification for the SACE.

To illustrate the meaning of PPI, SPPI and the difference between the two assumptions, consider for example the case $Y_i$ is QOL, and assume that whenever $S_i(a)=1$, the potential QOL is created by
\begin{equation*}
%\label{Eq:ExampleYpoModel}
    Y_i(a) = \nu + \tau a + \bbeta_{ax}^T\bX_{0i} + \beta_{au}U_i + \epsilon_{i}(a), \quad a=0,1,
\end{equation*}
where $U_i$ and $\epsilon_i(a)$ are unobserved random variables, independent of all variables introduced so far. If 
$\{S_i(0),S_i(1)\} \indep U_i |\bX_{0i}$, 
then both PPI and SPPI will hold. That will be the case if $U_i$ is, e.g., the baseline functionality of a muscle or group of muscles affecting the success in one of the tasks used to measure QOL, but is not informative about survival (at least not conditionally on $\bX_{0i}$, if we relax $U_i \indep 
\bX_{0i}$). If  $\{S_i(0),S_i(1)\} \ \cancel{\indep} \ U_i |\ \bX_{0i}$ and $U_i$  affects survival and QOL under either treatment status, then PPI and SPPI do not necessarily hold. If, however, for the entire population, the treatment improves the muscle functionality such that the task can be completed successfully, then $\beta_{1u}=0$, $Y_i(1)\indep U_i$ and thus, conditionally on $\{S_i(1)=1,\bX_{0i}\}$,  $Y_i(1)$ holds no information on $S_i(0)$, which is exactly PPI. However, $Y_i(0)$ does contain information about $S_i(1)$, even conditionally on $\{S_i(0)=1,\bX_{0i}\}$. As an additional illustration, we describe in Section A of the Web Appendix a structural equation model under which both PPI and SPPI do not hold due to the post-treatment variables $\bX_1$ being common causes of $S$ and $Y$.

Considering each of the two sets of assumptions, it is clear that CPSR is weaker than monotonicity and PPI is weaker than SPPI. \cite{chiba2011identification} showed that the SACE is identifiable under PPI and monotonicity. Here we extend this result by providing two new insights. First, the same identification result as in \cite{chiba2011identification} holds when PPI is replaced with the stronger SPPI and monotonicity with the weaker CPSR. 
\begin{proposition}
\label{Prop:SACEident}
Under SUTVA, randomization, SPPI and CPSR (for any $\xi$), the SACE is identified from the observed data by 	
\begin{equation}
\label{Eq:SACEident}
\text{SACE}=E_{\bX_0|A=0,S=1}\big[E(Y|A=1, S=1, \bX_0) - E(Y|A=0, S=1, \bX_0)\big].
\end{equation}
\end{proposition}
The proof is given in Section A of the Web Appendix. In this paper, the SACE is studied under randomization. Nevertheless, in Section A of the Web Appendix we extend  Proposition \ref{Prop:SACEident} for scenarios randomization did not take place and conditional exchangeability is a more plausible assumption.
Second, at face value, that the identification formula is the same under PPI and monotonicity, and under SPPI and CPSR, hints a trade-off exists with respect to the different  cross-world assumptions.  However, this trade-off does not exist, at least not for SACE point-identification.   
\begin{proposition}
\label{Prop:AssumptionsImply}
PPI and monotonicity together imply SPPI.
\end{proposition}
The proof is given in Section A of the Web Appendix. By Proposition \ref{Prop:AssumptionsImply}, the combination of SPPI and CPSR is weaker than the combination of PPI and monotonicity. Proposition \ref{Prop:SACEident} means that the SACE is identifiable without assuming a value for $\xi$. 
Nevertheless, assuming a specific value for $\xi$ has two important implications. First, the principal scores $\pi_{g}(\bx_0)$ are only identifiable if $\xi$ is known, which means it cannot be used for estimation of the SACE unless a specific value of $\xi$ is assumed (see Section \ref{Sec:Matching}). Second,  under PPI and CPSR, the SACE is only identifiable if a specific $\xi$ value is assumed (as well as a value for an additional sensitivity parameter, Table A1). We use this to construct a sensitivity analysis for monotonicity under PPI and CPSR as a function of $\xi$ (Section \ref{Sec:Sens}).  Finally, while $\xi$ is not identifiable from the data, in Section A of the Web Appendix we show that whenever $p_0+p_1> 1$, $\xi$ is bounded within the range
$\big[\max(0, \: \frac{p_0 - p_1}{p_1}), \
\: \frac{(1 - p_1)}{p_0 - (1 - p_1)}\big]$, where $p_a=\Pr(S=1|A=a)$. When $p_0+p_1 \le 1$, only the lower bound is applicable.

Estimation of the SACE can be done by standardization \citep{chiba2012estimation}, or weighting using the principal scores \citep{ding2017principal,feller2017principal}. In the next section, we discuss the adaptations needed for matching, a widely-used and well-known approach by practitioners, to estimate the SACE. 

\section{Matching methods for SACE estimation}
\label{Sec:Matching}

Proposition \ref{Prop:SACEident} implies that a consistent SACE estimator can be  obtained by comparing the treated and untreated survivors at each level of $\bX_0$, and then average the estimated differences according to the distribution of $\bX_0$ in the $\{A=0,S=1\}$ group. In Section A of the Web Appendix, we show that under monotonicity or CPSR, ignoring finite-sample variability, the conditional distribution of $\bX_0|A=0,S=1$ is identical to the conditional distribution of $\bX_0|G=as$. Therefore, matching for SACE estimation seeks to create a matched sample with the following properties: 
\begin{enumerate}[Property (a)]
    \item \textit{Survivors only.} The matched sample includes only survivors $S=1$.
    \item \textit{Closeness.} Within each matched pair or group of the matched sample, the values of $\bX_0$ should be close to each other.
    \item \textit{Distribution preservation.} The matched sample should have the distribution of $\bX_0$ as in the $\{A=0,S=1\}$ group. 
\end{enumerate}
The closeness property underpins having the differences between the treated and untreated at each matched pair/group  being sensible estimators of  $E(Y|A=1, S=1, \bX_0) - E(Y|A=0, S=1, \bX_0)$. The distribution preservation property ensures the outer expectation  in  \eqref{Prop:SACEident}  could be correctly estimated nonparametrically.  

In observational studies, we seek to create balance between the treated and the untreated with respect to pre-treatment confounders associated with the treatment and the outcome.  For SACE estimation, we also seek to create balance between the treated and the untreated, but with respect to covariates that are associated with the non-mortality outcome and, due to selection bias caused by survival, are also associated with the treatment.
 
With the three properties in mind, we  propose the following general matching procedure and analysis for SACE estimation, which parallels the use of matching in observational studies.
\begin{enumerate}[Step (i)]
	\itemsep0em 
	\item Identify covariates $\bX_0$ for which SPPI (or PPI) is a reasonable assumption.
	\item Choose a distance measure $\mathcal{D}(\bx_{0i},\bx_{0j})$ determining proximity of two units $i$,$j$ with covariate vectors $\bX_{0i}=\bx_{0i},\bX_{0j}=\bx_{0j}$. 
	\item Create matched pairs or groups by matching  each unit from $\{A=0, S=1\}$ to a unit(s) from  the  $\{A=1,S=1\}$ group while minimizing the total distances within pairs/groups. 
	\item Analyze the matched sample.
\end{enumerate}
Step (i) pertains to SACE identification, discussed in the previous section. Regarding the distance measure (Step (ii)), for the standard use of matching for observational data, three common choices are exact distance (zero if covariate vectors are identical, infinity otherwise), Mahalanobis distance, 
or the difference in the (logit of) the propensity score.  The exact or Mahalanobis distance can also be used for achieving balance on $\bX_0$. However, the number of covariates needed for SPPI (or PPI) to hold can be non-small, making the exact and Mahalanobis distances less attractive as distance measures, especially when there are both discrete and continuous covariates on which one would like to match. 

In observational studies, matching on the propensity score is justified by its balance properties \citep{rosenbaum1983central}.  Recently, analogous results for balance on covariates with respect to different principal strata were derived for functions of the principal scores $\pi_g(\bx_0)$ \citep{ding2017principal,feller2017principal}. 
Let $\widetilde{\pi}^{1}_{as}(\bx_0) = \frac{\pi_{as}(\bx_0)}{\pi_{as}(\bx_0) + \pi_{pro}(\bx_0)}$, and $\widetilde{\pi}^{0}_{as}(\bx_0) = \frac{\pi_{as}(\bx_0)}{\pi_{as}(\bx_0) + \pi_{har}(\bx_0)}$.  
Matching on $\widetilde{\pi}^{1}_{as}(\bx_0)$ seems attractive because given $\widetilde{\pi}^{1}_{as}(\bx_0)$, always-survivors and protected in $\{A=1,S=1\}$ are balanced with respect to $\bX_0$ \citep{ding2017principal}.
Thus, due to randomization, conditionally on $\widetilde{\pi}^{1}_{as}(\bx_0)$, the untreated always-survivors and units in $\{A=1,S=1\}$ are balanced with respect to $\bX_0$. The principal score function $\pi_{g}(\bx_0)$ can be estimated, e.g., using an EM algorithm \citep{ding2017principal}. The details of the EM algorithms for the sequential model and for the multinomial regression model are given in Section B of the Web Appendix. 

Taking the estimated principal score as the sole distance measure, i.e., $\mathcal{D}(\bx_{0i},\bx_{0j})=|\widehat{\widetilde{\pi}}^{1}_{as}(\bx_{0i})-\widehat{\widetilde{\pi}}^{1}_{as}(\bx_{0j})|$, may suffer from disadvantages analogue to those of the propensity score in observational studies, namely, two units with very different covariate vectors can have similar $\widehat{\widetilde{\pi}}^{1}_{as}$. Furthermore, a potential disadvantage shared by both matching and weighting using $\widehat{\widetilde{\pi}}^{1}_{as}(\bx_0)$ is that the analysis relies on correct model specification for the principal scores. This challenge is amplified for principal score models, as they model a latent stratum which is not directly observable for all study participants. A third challenge is that $\widetilde{\pi}^{1}_{as}(\bx_0)$ is only identifiable  if $\xi$ is known. Assuming $\xi=0$ (i.e., monotonicity) or any specific value of $\xi$ enables identification and estimation of  $\widetilde{\pi}^{1}_{as}(\bx_0)$, but the price is that assumptions stronger than what is needed for SACE estimation are imposed (Table A1 in the Web Appendix). Assuming the wrong value for $\xi$ will then result in bias, which is not expected for methods not using the principal scores. We illustrate this point in Section \ref{Sec:Sims}.

A possible compromise is to combine the Mahalanobis distance with a \textit{caliper} \citep{rubin2000combining,stuart2010matching} on the principal score, namely
\begin{equation}
\label{Eq:CaliperDef}
 \mathcal{D}(\bx_{0i},\bx_{0j})=
 \begin{cases}
 (\widetilde{\bx}_{0i}-\widetilde{\bx}_{0j})^T\Sigma^{-1}_{\widetilde{\bX}_0}(\widetilde{\bx}_{0i}-\widetilde{\bx}_{0j}) & |\widehat{\widetilde{\pi}}^{1}_{as}(\bx_{0i})-\widehat{\widetilde{\pi}}^{1}_{as}(\bx_{0j})| \le c\\
    \hfil  \infty& |\widehat{\widetilde{\pi}}^{1}_{as}(\bx_{0i})-\widehat{\widetilde{\pi}}^{1}_{as}(\bx_{0j})| > c
  \end{cases}
\end{equation}
where $\widetilde{\bx}_{0i},\widetilde{\bx}_{0j}$ and $\widetilde{\bX}_0$ are subsets of key covariates from $\bx_{0i},\bx_{0j}$ and $\bX_0$, respectively,  $\Sigma_{\widetilde{\bX}_0}$ is the empirical covariance matrix of $\widetilde{\bX}_0$ and $c$ is a user-specified threshold. 

Turning to Step (iii) above, if all $\{A=0,S=1\}$ are matched, then if the closeness property was achieved, then distribution preservation holds in the matched sample.
However, the number of always-survivors in the $\{A=0, S=1\}$ and $\{A=1,S=1\}$ groups might differ (e.g., if randomization probabilities are unequal). This may harm both closeness if inadequate matches are made, and distribution preservation if not all $\{A=0,S=1\}$ group members are matched. A remedy to this problem is found by considering matching with replacement, so potential always-survivors from $\{A=1,S=1\}$ can be chosen as matches for more than one individual from $\{A=0,S=1\}$. %Finally, the matching framework offers a number of potential methods and variants, among which assigning multiple treated to each untreated unit (1:$k$ matching), or implementing full matching \citep{hansen2004full}. 

\subsection{Analysis of the matched data}
\label{Sec:Analysis}

Upon achieving reasonable balance, researchers can turn their attention to estimation of and inference about the SACE (Step (iv)).  One key use of matching has been as a way to achieve balance so randomization-based inference and estimation can be carried out \citep{rosenbaum2011observational}. Consider the sharp null hypothesis among the always-survivors
\begin{equation}
\label{Eq:SharpNull}
H_0: Y_i(0) = Y_i(1), \quad \forall i \in \{j: S_j(0)=S_j(1)=1\}.
\end{equation}
 The proposed matching framework enables researchers to answer a basic question: is there evidence for a causal effect of the treatment on the non-mortality outcome for at least one person? Failure to reject the null \eqref{Eq:SharpNull} means the data do not support a positive answer to this question.

In a simple 1:1 matched set without replacement, the null  hypothesis \eqref{Eq:SharpNull} can be tested using permutation tests, such as Wilcoxon signed-rank test. 
For 1:1 matching with replacement, the null of no causal effect can be non-parametrically tested by first grouping each treated survivor with all of their untreated matches, and then carrying out the aligned-rank test \citep{hodges1962rank,heller2009matching}. 

We turn to estimation and inference based on the sampling distribution.  For standard observational studies, the average treatment effect is often estimated by the crude difference between the treated and the untreated among the matched sample, and (paired or unpaired) Student's $t$-test is used for inference. However, because the analogue of closeness for matching in observational studies typically cannot be achieved with non-small and/or continuous confounders, bias is expected. \cite{abadie2006large} studied the magnitude of this bias, gave conditions under which the crude difference estimator is consistent, and derived the asymptotic distribution of this estimator.

It is generally recommended to consider covariate-adjusted treatment effects beyond the crude differences in the matched sample \citep{stuart2010matching}. This may both reduce bias from inexact matching, and improve efficiency of estimators. In our case, one may consider the regression model %including the treatment and the covariates that were used for matching. In our case, this translates to the model
\begin{equation}
\label{Eq:LinOutModel}
    E(Y|A,S=1,\bX_0) = \beta_0 + \beta_1A + \bbeta_2^T\bX_0,
\end{equation}
and non-linear terms, interactions and additional variables associated with $Y$ can also be included. Easy to see from \eqref{Eq:SACEident} that under model \eqref{Eq:LinOutModel}, the SACE equals to $\beta_1$. More generally, given a  model for $E_{\btheta}(Y|A,S=1,\bX_0)$, with parameters $\btheta$, we can write 
\begin{equation}
\label{Eq:gtheta}
g(\bX_0,\btheta) = E_{\btheta}(Y|A=1,S=1,\bX_0) - E_{\btheta}(Y|A=0,S=1,\bX_0)
\end{equation}
for some $g$ induced by the model. Note that this model can be either non-parametric, semi-parametric or fully parametric. The analysis of the matched dataset starts with obtaining estimates $\hat{\btheta}$, possibly using weighed estimation when matching is either  1:$k$ or with replacement. Then, the SACE can be estimated by $\widehat{SACE}_{\btheta}=\sum_{i=1}^{n}r_ig(\bx_{0i},\hat{\btheta})$, where $r_i=(1-A_i)S_i/\sum_{j=1}^{n}(1-A_j)S_j$. Consistency of $\widehat{SACE}_{\btheta}$ depends on correctly specifying the model $g(\bX_0,\btheta)$. Asymptotic normality and convergence rates of the estimators depend on the specific model and estimator for $g(\bX_0,\btheta)$. When the model is misspecified, it has been argued that the matching step may reduce the bias due to misspecification \citep{ho2007matching}. 

A general variance estimator for $\widehat{SACE}_{\btheta}$ is unavailable, as variance estimation for post-matching estimators is an active field of research.  Even for the simple linear regression case, only recently \cite{abadie2022robust} have shown that the variance of the estimated coefficients in post-matching regression can be consistently  estimated by clustered standard errors (SEs). Their results are valid only for 1:$k$ matching without replacement. Furthermore,  \cite{abadie2008failure} showed that when matching is implemented with replacement, the standard bootstrap cannot be used.
\cite{abadie2006large} 
derived large sample properties of matching estimators (when matching is with replacement), and proposed consistent variance estimators. They showed that matching estimators are not $\sqrt{n}$-consistent in general, and described the required conditions for the matching estimators to be $\sqrt{n}$-consistent.

As an alternative to the above-described estimators, we also consider a bias-corrected (BC) approach \citep{abadie2011bias}. The BC estimator combines any consistent regression model fitted in the original dataset, with a  mean difference estimator. The idea is that the BC estimator adjusts for the bias   resulting from inexact matching by adding a bias term to the difference of the means. \cite{abadie2011bias} showed that the BC approach yields $\sqrt{n}$-consistent and asymptotically normal estimators for the average treatment effect and the average treatment effect in the treated, even when the outcome model is fitted using nonparametric series regression, as long as the number of parameters increases slowly enough. Their results can be naturally extended to our case, by recognizing that our proposed approach essentially estimates the average effect on the untreated among the survivors. This is reflected by the outer expectation in \eqref{Eq:SACEident} being over $\bX_0|A=0,S=1$. 

The BC estimator for the SACE under 1:1 matching is constructed as follows. Prior to creating the matched sample, a regression estimator $\hat{\mu}_1(\bx_0)$ for $E[Y|A=1,S=1,\bX_0=\bx_0]$ is calculated. Let $j=1,...,M$ index the matched pair, and let $Y_{ja}$ and $\bx_{0,ja}$ be the observed outcome and covariates for the unit with $A=a$ in matched pair $j$. Then, following \cite{abadie2011bias}, let $\tilde{Y}_{j1} = Y_{j1} + [\hat{\mu}_1(\bx_{0,j0}) - \hat{\mu}_1(\bx_{0,j1})]$ be the imputed outcome for the treated unit in matched pair $j$. The BC SACE estimator is
\begin{equation}
\label{Eq:BCdef}
\widehat{SACE}_{BC} = \frac{1}{M} \sum_{j=1}^{M} (\tilde{Y}_{j1} - Y_{j0}).
\end{equation}
As with $\widehat{SACE}_{\btheta}$,  $\widehat{SACE}_{BC}$  will be $\sqrt{n}$-consistent only if the model is correctly specified.
% %%%%%%%%%%%%%%%%%%%%%%%%%%%%%%%%%%%%%%%%%%%%%%%%%%%%%%%

\section{Sensitivity analyses} 
\label{Sec:Sens}

Causal inference identification relies on assumptions that are untestable from the data. Therefore,  methods are often coupled with sensitivity analyses that relax one or more of the identifying assumptions to obtain an identification result as a function of a meaningful, but unknown, sensitivity parameter(s). Then, researchers can vary the value of this parameter to receive a curve of possible estimates as a function of this parameter. Sensitivity analyses of this type are often used for the SACE \citep{hayden2005estimator,chiba2012estimation,ding2017principal}. Additionally, if these parameters can be bounded, bounds for the causal effect of interest can be estimated.

Here, we describe how the sensitivity analyses proposed by \cite{ding2017principal} for their weighting-based methods can be adapted for matching-based estimators. Like \cite{ding2017principal}, our sensitivity parameters are based on the ratio of principal stratum proportions ($\xi$) and the ratio between the mean potential outcomes in different strata (defined below as $\alpha_0$ and $\alpha_1$). Because \cite{ding2017principal} developed a sensitivity analysis for a general principal stratum setup, our choice of sensitivity parameters slightly differs. For example, \cite{ding2017principal} uses the ratio $\frac{\pi_{har}(\bx_0)}{\pi_{pro}(\bx_0)}$ for $\xi$, while we took $\xi$ be the ratio between the harmed and always-survivors (at each level of $\bX_0)$. 

Table A1 in the Web Appendix reviews identifiability of the SACE and of the principal scores, and summarizes the sensitivity parameters needed for identification under each combination of assumptions. Here we focus on sensitivity analysis for PPI/SPPI under monotonicity and for monotonicity under PPI and CPSR. 

The first step of both sensitivity analyses carries out matching on $\bX_0$ among the survivors. The considerations in how to choose the distance measure for this step are the same as described in Section \ref{Sec:Matching}.

%%%%%%%%%%%%%%%%%%%%%%%%%%%%%%%%%%%%%%%%%%%%%%%%%%%%%%%%%%%
%%%%%%%%%%%%%%%%%%%%%%%%%%%%%%%%%%%%%%%%%%%%%%%%%%%%%%%%%%%

%%%%%%%%%%%%%%%%%%%%%%%%%%%%%%%%%%%%%%%%%%%%%%%%%%%%%%%%%%%%
\subsection{Sensitivity analysis for PPI/SPPI}
\label{SubSec:SensPPI}

Define  $\mu_{a,g}(\bx_0)=E[Y(a)|G=g, \bX_0=\bx_0]$ 
to be the mean potential outcome under $A=a$ for units within the principal strata $g$ with covariate values $\bx_0$. Consider the sensitivity parameter $\alpha_1=\frac{\mu_{1,pro}(\bx_0)}{\mu_{1,as}(\bx_0)}$.
In words, $\alpha_1$ is the ratio between the mean outcomes under treatment of the protected and of the always-survivors at any value of $\bX_0$.	Because $\alpha_1=1$ under PPI, values further away from one reflect a more substantial deviation from PPI. Whether values larger or smaller than one (or both) are taken for $\alpha_1$ could be based on subject-matter expertise. For example, when higher outcome values reflect better health, the always-survivors might be expected to be healthier than the protected (at each level of $\bX_0$) and hence the sensitivity analysis would focus on $\alpha_1$ values smaller than one. In Section A of the Web Appendix we show that $\alpha_1$ can be bounded by $\frac{\min_{\bx_0} E(Y|A=1,S=1,\bx_0)}{\max_{\bx_0} E(Y|A=1,S=1,\bx_0)}\le \alpha_1 \le \frac{\max_{\bx_0} E(Y|A=1,S=1,\bx_0)}{\min_{\bx_0} E(Y|A=1,S=1,\bx_0)}$. To reduce variance, model-based bounds can be used by fitting a regression model for $E(Y|A=1,S=1,\bx_0)$, and then replacing the minimum and maximum expectations with the minimal and maximal model prediction in the data.
	
The following proposition provides the basis for the proposed sensitivity analysis.
\begin{proposition} \label{Prop:SAppi}
Under SUTVA, randomization, and monotonicity, the SACE is identified from the data as a function of $\alpha_1$ by
\begin{equation*}
%\label{Eq:PropSensPPI}
E_{\bX_0|A=0,S=1}\left\{\frac{E(Y|A = 1, S=1, \bX_0)}{\widetilde{\pi}^1_{as}(\bX_0) + \alpha_1[1 - \widetilde{\pi}^1_{as}(\bX_0)]} - E(Y|A=0, S = 1, \bX_0)\right\}.
\end{equation*}
\end{proposition}
The proof is given in Section A of the Web Appendix.
%Importantly, under monotonicity, $\widetilde{\pi}^1_{as}(\bX_0)$ is identifiable from the observed data.
Under PPI ($\alpha_1 = 1$) we obtain the same identification formula as in Proposition \ref{Prop:SACEident}. Following Proposition \ref{Prop:SAppi},  matching-based SACE estimation for each  $\alpha_1$ value is as follows.
\begin{enumerate}
\itemsep 0.2cm
\item Estimate the principal scores and subsequently
$\widetilde{\pi}^1_{as}(\bx_0)$ as described in Section \ref{Sec:Matching}.
\item Implement a matching procedure with the chosen distance measure, possibly, but not necessarily, using the estimated $\widehat{\widetilde{\pi}}^1_{as}(\bx_0)$.
\item Estimate $E(Y|A=a,S=1,\bX_0=\bx_0)$ for $a=0,1$. %For example, using two separate linear regression estimators.
\item Replace \eqref{Eq:gtheta} with
\begin{equation*}
%\label{Eq:gthetaSensPPI}
g^{\alpha_1}(\bX_0,\btheta) = \frac{E_{\btheta}(Y|A = 1, S=1, \bX_0)}{\widehat{\widetilde{\pi}}^1_{as}(\bx_0) + \alpha_1[1 - \widehat{\widetilde{\pi}}^1_{as}(\bX_0)]} - E_{\btheta}(Y|A=0,S=1,\bX_0),
\end{equation*}
and estimate $\widehat{SACE}^{\alpha_1}_{\btheta}=\sum_{i=1}^{n}r_ig^{\alpha_1}(\bx_{0i},\hat{\btheta}).$
\end{enumerate}
We illustrate the use of this approach in Section \ref{Sec:Data}. Although our proposed sensitivity parameter is identical to the one proposed by \cite{ding2017principal}, and although it requires estimation of  $\widetilde{\pi}^1_{as}(\bx_0)$, it does offer the flexibility of not using it in the matching process, and thus it is expected to reduce, at least to some extent, the dependence of the final results on correct specification of the principal score model.

%%%%%%%%%%%%%%%%%%%%%%%%%%%%%%%%%%%%%%%%%%%%%%%%%%%%%%%%%%%
%%%%%%%%%%%%%%%%%%%%%%%%%%%%%%%%%%%%%%%%%%%%%%%%%%%%%%%%%%%

\subsection{Sensitivity analysis for monotonicity} \label{SubSec:Sensmono}

To develop a sensitivity analysis for monotonicity under PPI, without imposing the stronger SPPI,  we present here an identification formula under PPI and CPSR, depending on two sensitivity parameters. The first is the previously-defined ratio $\xi$.  Larger values of $\xi$ reflect larger divergence from monotonicity. The second is $\alpha_0 = \frac{\mu_{0,har}(\bx_0)}{\mu_{0,as}(\bx_0)}$. 
With a similar logic to the interpretation of $\alpha_1$, the sensitivity parameter $\alpha_0$ represents the relative frailty of those at the harmed stratum compared to the always-survivors, as reflected by the ratio between mean outcome when untreated. Similarly to $\alpha_1$, model-based bounds for $\alpha_0$ can be obtained; see Section A of the Web Appendix.

Unlike the results of \cite{ding2017principal}, our proposed sensitivity analysis does not formally require estimation of the principal scores. However, if one wishes to use $\widehat{\widetilde{\pi}}^1_{as}(\bx_0)$ in the matching process, the principal scores need to be estimated. The following proposition provides the basis for the proposed sensitivity analysis.
\begin{proposition} 
\label{Prop:SAmono}
Under SUTVA, randomization, PPI, and CPSR, the SACE is identified from the data as a function of $\xi$ and $\alpha_0$ by 
\begin{equation*}
%\label{Eq:PropSensMono}
 E_{\bX_0|A=0,S=1}\left[E(Y|A = 1, S=1, \bX_0) - \frac{(1 + \xi)}{(1 + \xi\alpha_0)} E(Y|A=0, S = 1, \bX_0)\right]. 
 \end{equation*}
 \end{proposition}
The proof is given in Section A of the Web Appendix.
%Propositions \ref{Prop:SACEident} and \ref{Prop:SACEidentSPPI} can be seen as direct corollaries of Proposition \ref{Prop:SAmono}, because under monotonicity ($\xi = 0$), and/or under SPPI ($\alpha_0=1$), the identification result \eqref{Eq:PropSensMono} coincides with the one presented in Propositions \ref{Prop:SACEident} and \ref{Prop:SACEidentSPPI}.

To utilize Proposition \ref{Prop:SAmono} for a sensitivity analysis, we can repeat the following analysis for different combinations of $(\xi, \alpha_0)$. 
\begin{enumerate}
\itemsep 0.2cm
\item Implement a matching procedure as previously described, with or without estimating the principal scores.
\item Estimate $E(Y|A=a, S=1, \bX_0=\bx_0)$ for $a=0,1$. 
\item Replace \eqref{Eq:gtheta} with
\begin{equation*}
%\label{Eq:gthetaSensMono}
g^{\alpha_0,\xi}(\bX_0,\btheta) = E_{\btheta}(Y|A = 1, S=1, \bX_0) - \frac{(1 + \xi)}{(1 + \xi\alpha_0)}E_{\btheta}(Y|A=0, S=1,\bX_0),
\end{equation*}
and estimate
$$
\widehat{SACE}^{\alpha_0,\xi}_{\btheta}=\sum_{i=1}^{n}r_ig^{\alpha_0,\xi}(\bx_{0i},\hat{\btheta}).
$$
\end{enumerate}
If one chooses to use the principal scores in the matching procedure (e.g., with a caliper on $\widehat{\widetilde{\pi}}^1_{as}(\bx_0)$), the EM algorithm for estimating the principal scores should be revised, and implemented separately for each value of $\xi$. The details are given in Section B of the Web Appendix.

\section{Matching for CSEs and for SACE without PPI}
\label{Sec:CSEs}

CSEs have been recently proposed as an alternative to SACE. For complete motivation, definitions and theory, we refer the reader to \cite{stensrud2022conditional}. Here we provide a summary of key points, before showing how the matching framework can be used for estimating CSEs.

Two key limitations of the SACE are often raised. First, the always-survivors stratum can be a non-trivial subset of the population, and it cannot be known, both for the observed data and for future individuals, who is an always-survivor and who is not. Second, identification of the SACE relies on cross-world assumptions, namely assumptions on the unidentifiable joint distribution of potential outcomes under different intervention values. Such assumptions cannot be tested from the data nor they can be guaranteed to hold by experimental design. Here, both monotonicity and PPI are cross-world assumptions. 

Motivated by these limitations, \cite{stensrud2022conditional} proposed the CSEs approach. A key prerequisite for this approach is having an alternative to the single treatment $A$, represented by two treatments $A_S,A_Y$, with potential outcomes $S(A_S=a_S,A_Y=a_Y)$ and $Y(A_S=a_S,A_Y=a_Y)$, such that joint interventions on $A_S$ and $A_Y$ lead to the same results as interventions on $A$. The latter is formalized by the \textit{modified treatment assumption}, which implies that setting $A=a$ or $A_S=a,A_Y=a$ lead to the same potential outcomes of $S$ and $Y$.

A second critical assumption is $A_Y$ \textit{partial isolation}, that states there are no causal paths between $A_Y$ and $S$. Under this assumption, $$
S(A_S=a_S,A_Y=0)=S(A_S=a_S,A_Y=1):=S(A_S=a_S).
$$ 
Under $A_Y$ partial isolation, the CSEs for $a_S=0,1$ are
\begin{equation}
\label{Eq:CSEdef}
CSE(a_S)=E[Y(a_S, 1) - Y(a_S, 0) | S(a_S)=1]. 
\end{equation} 
In a study of training program effect on employment ($S$) and earnings ($Y$), $A_Y$ can include, for example, modules targeting directly salary negotiation skills and hence are not expected to affect $S$, while $A_S$ represents the rest of the program contents (that affect $S$ and may or may not affect $Y$). A detailed example involving cancer treatment ($A$) and QOL ($Y$) is given by \cite{stensrud2022conditional}.

\cite{stensrud2022conditional} make the point that unlike the members of the principal stratum $\{S(0),S(1)\}$, the members of the conditioning set $S(a_S)=1$ can be observed from the data (under certain assumptions). Furthermore, there are scenarios under which the SACE is not identifiable and the CSEs are. Note that the interpretation of \eqref{Eq:CSEdef} as a direct effect of $A$ on $Y$ is non-trivial, because,  under the modified treatment assumption and $A_Y$ partial isolation, there can still be causal paths between $A_S$ and $Y$ which do not involve $S$. Nevertheless, under the stronger full isolation assumption there are no such paths, and the CSEs retain an interpretation as a direct effect on $Y$. Further discussion of these estimands and their utility can be found in \cite{stensrud2022conditional}. 

The following Proposition provides an identification formula for $CSE(a_S)$, which resembles our identification formula for the SACE, and even coincides with it under certain conditions. It additionally relies on a positivity assumption and dismissible component conditions, all described in Section 7 of \cite{stensrud2022conditional}.
\begin{proposition}
\label{Prop:CSEident}
Under randomization, $A_Y$ partial isolation, the modified treatment assumption, positivity, and the dismissible component conditions
$$
CSE(a_S)=E_{\bX_0,\bX_1|S=1, A=a_S}\bigg[E(Y| A=1, S=1,\bX_0, \bX_1) - E(Y| A=0, S=1,\bX_0, \bX_1)\bigg]
$$
\end{proposition}
The proof follows from Theorem 1 of \cite{stensrud2022conditional}, and a few more lines given in Section A of the Web Appendix. Proposition \ref{Prop:CSEident} resembles the identification result for the SACE presented in \eqref{Prop:SACEident}. Note that for estimation of $CSE(a_S=0)$, the process is identical to SACE estimation by matching, with the only change is being that balance should be achieved with respect to both $\bX_0$ and $\bX_1$. For the matching-based approach for estimation of $CSE(a_S=1)$, the distribution preservation property is revised to have the distribution of $(\bX_0,\bX_1)$ in the matched set to be the same as in $\{A=1,S=1\}$.  

%%%%%%%%%%%%%%%%%%%%%%%%%%%%%%%%%%%%%%%%%%%%%%%%%%%%%%%%%%%
%%%%%%%%%%%%%%%%%%%%%%%%%%%%%%%%%%%%%%%%%%%%%%%%%%%%%%%%%%%

\section{Simulation studies}
\label{Sec:Sims}

We conducted simulation studies to assess the performance of the matching-based approach and to compare the proposed estimators to naive approaches and to the weighting-based method \citep{ding2017principal}.  The number of simulation repetitions was 1,000 for each simulation scenario.  The sample size of each simulated dataset was 2,000. The \textbf{R} package \texttt{Matching} \citep{sekhon2011Matching} was used for the matching process. Technical details, simulation parameters, and additional results are given in Section C of the Web Appendix.

\subsection{Data generating mechanism}
\label{SubSec:DGM}

For each unit $i$, a covariate vector of length $k$ was simulated from the  multivariate normal distribution, $\bX_{0i} \sim N_k({\boldsymbol{0.5}}_k,  \bI_k)$,
where ${\boldsymbol{0.5}}_k$ is a vector of length $k$ with entries 0.5, and $\bI_k$ is the identity matrix of dimension $k$. The stratum $G_i$ was generated according to the sequential logistic regression model \eqref{Eq:SeqRegModelS0}--\eqref{Eq:SeqRegModelS1_given_S0=0} with $\bZ_i=(1, \bX_{0i})$. We set $\Pr[S_i(1)=1|S_i(0)=1,\bx_{0i}] = \frac{1}{1+\xi}$ for a given $\xi$ under CPSR and set $\xi=0$ under monotonicity. In an additional simulation study, we considered the multinomial regression model for $G_i$ (Section B and Tables C4--C5 of the Web Appendix).

For the always-survivors, the potential outcomes $\{Y(0),Y(1)\}$ were generated according to linear regression models with means $E[Y_i(a)|S_i(a)=1, \bX_{0i}=\bx_{0i}]=\beta_{0,a} + \bbeta^T_{a}\bx_{0i}$, $a=0,1$, and  additive correlated normal errors $\{\epsilon_i(0),\epsilon_i(1)\}$ with zero mean, unit variance, and correlation $0.4$. For the protected or harmed, only $Y(1)$ or $Y(0)$, respectively, were generated, and for the never-survivors, both potential outcomes were truncated by death. For the linear regression models, we considered scenarios including all $A$-$\bX_0$ interactions and scenarios without including any interaction. The treatment assignment $A_i$ was randomized with probability $0.5$. Finally, the observed survival status and outcomes were determined by  $S_i = A_iS_i(1) + (1 - A_i)S_i(0)$, and $Y_i = A_i Y_i(1) + (1 - A_i)Y_i(0)$.

%%%%%%%%%%%%%%%%%%%%%%%%%%%%%%%%%%%%%%%%%%%%%%%%%%%%%%%%%%%
%%%%%%%%%%%%%%%%%%%%%%%%%%%%%%%%%%%%%%%%%%%%%%%%%%%%%%%%%%%

\subsubsection{Scenarios}
\label{SubSubSec:Scenarios}

To have a fair assessment of the finite-sample performance of the different methods, we considered a variety of scenarios under different number of covariates, stratum models and outcome models. In Scenario A, the always-survivors stratum comprised 50\% of the population. In Scenario B, the always-survivors stratum  comprised 75\% of the population.
In each of these scenarios, we considered the option of relatively high versus relatively low proportion of protected. These scenarios were created by choosing different values for the sequential logistic regression model coefficients $\bgamma_{S(a)}, a=0,1$ (Tables C2 and C3). The coefficient values  were chosen to obtain the desired stratum proportions.
We also considered different values for $\xi \in \{0,0.05,0.1,0.2\}$.

In reality, the functional forms of the principal score and the outcome models are unknown to the researchers. Therefore, we conducted simulations under misspecification of the principal score model and/or the outcome model.
When misspecified, the true outcome model included squared and exponential terms for two of the covariates, and the true principal score model included such terms and an interaction term of two covariates.

%%%%%%%%%%%%%%%%%%%%%%%%%%%%%%%%%%%%%%%%%%%%%%%%%%%%%%%%%%%
%%%%%%%%%%%%%%%%%%%%%%%%%%%%%%%%%%%%%%%%%%%%%%%%%%%%%%%%%%%

\subsection{Analyses}
\label{SubSec:SimsAnalyses}

For each of the simulated datasets, we calculated the two naive estimators -- the mean difference in the survivors, and the mean difference in the composite outcomes. 
For the matching approach, we followed the procedure described in Section \ref{Sec:Matching}, by matching every untreated survivor to a treated survivor. For the matching-based estimators,
we compared matching on $\widehat{\widetilde{\pi}}^1_{as}(\bx_0)$; matching on the Mahalanobis distance;  and matching on Mahalanobis distance with a caliper on $\widehat{\widetilde{\pi}}^1_{as}(\bx_0)$, taking $c$ (Equation \eqref{Eq:CaliperDef}) to be 0.25 standard deviations (SDs) of the estimated $\widehat{\widetilde{\pi}}^1_{as}(\bx_0)$. An EM algorithm was used to estimate
$\widehat{\widetilde{\pi}}^1_{as}(\bx_0)$ under the sequential logistic regression model.
For each of the above options, we considered matching with and without replacement of the treated survivors. Of note is that in practice, as we also illustrate in Section \ref{Sec:Data},
one chooses the distance measure that achieves the best balance, so our assessment of the matching framework approach is likely to be pessimistic. 

In the matched datasets, we considered the following estimators:  crude difference estimators, linear regression estimators without and with all $A-\bX_0$ interactions, and the BC matching estimator $\widehat{SACE}_{BC}$ \eqref{Eq:BCdef} (with linear regression for $\hat{\mu}_1(\bx_0)$). The latter is implemented only for matching with replacement \citep{sekhon2011Matching}. Ordinary least squares (OLS) was used when matching was without replacement and weighted least squares (WLS) when matching was with replacement.
% Matching was carried out using Mahalanobis distance with caliper on $\widehat{\widetilde{\pi}}^1_{as}(\bx_0)$ absolute difference (0.25 estimated SD), with the crude estimator which followed matching on $\widehat{\widetilde{\pi}}^1_{as}(\bx_0)$ being the only exception.
 
For matching without replacement, we used the simple SE estimator for the crude differences and clustered SEs for the OLS \citep{abadie2022robust}.
For matching with replacement, we used SE estimators accounting for weights for the crude and BC estimators  \citep{abadie2006large,abadie2011bias,sekhon2011Matching} and clustered weighted SEs for the WLS estimator.

%proposed by \cite{abadie2006large} for the crude difference (except for matching on $\widehat{\widetilde{\pi}}^1_{as}(\bx_0)$ solely), clustered SEs for the WLS estimator, and the SE estimator proposed by \cite{abadie2011bias} for the BC estimator. 

We also considered the model-assisted weighting-based estimator of \cite{ding2017principal} (DL), which uses the estimated principal scores. This estimator was found by the authors to be more efficient than the simple weighting estimator.

As indicated in Section \ref{SubSubSec:Scenarios}, we considered scenarios with misspecification of the functional form of the principal score model and/or the outcome model. Additionally, for analyses using the principal scores, we repeated the analyses under different (possibly incorrect) specified values for $\xi$, denoted by $\xi_{assm} \in \{0,0.05,0.1,0.2\}$.

%%%%%%%%%%%%%%%%%%%%%%%%%%%%%%%%%%%%%%%%%%%%%%%%%%%%%%%%%%%%%

%%%%%%%%%%%%%%%%%%%%%%%%%%%%%%%%%%%%%%%%%%%%%%%%%%%%%%%%%%%%%

\subsection{Results}

We first focus on the case the true outcome model included all $A-\bX_0$ interactions. Figure \ref{Fig:biasS1withInterDGMseq} presents the bias of selected matching estimators and of the DL estimator, under monotonicity.
Table \ref{Tab:SimRes} presents more estimators and more detailed results on their performance for Scenarios A and B with $k=5$ covariates and low $\pi_{pro}$. %Additional results are presented in Section C.2 of the Web Appendix.

When both the principal score and the outcome models were correctly specified, all matching-based estimators and the DL estimator had low or negligible bias in comparison to the naive estimators (Figure \ref{Fig:biasS1withInterDGMseq} and Table \ref{Tab:SimRes}). The SEs of the regression-based post-matching estimators were generally well estimated (Table \ref{Tab:SimRes}).

%Generally, the absolute bias of the crude estimator, and DL estimator (under wrong $\xi$ values), increased as the number of covariates $k$ increased.
%This bias was apparent only in part of the scenarios, and was mitigated when matching on $\widehat{\widetilde{\pi}}^1_{as}(\bx_0)$ solely.

Interestingly, under the sequential regression model, misspecification of the principal score model, but not the outcome model, resulted in minimal bias, for both matching and weighting using $\widehat{\widetilde{\pi}}^1_{as}(\bx_0)$. Under a misspecified multinomial regression model for the principal strata, more substantial bias was observed for the DL estimator than for matching using Mahalanobis distance with a caliper on $\widehat{\widetilde{\pi}}^1_{as}(\bx_0)$ (Figure C6 and Table C13).
When only the outcome model was misspecified,
methods that use $\widehat{\widetilde{\pi}}^1_{as}(\bx_0)$ were unbiased, while the model-based estimators after  matching on Mahalanobis distance (with or without a caliper) were biased (Figure \ref{Fig:biasS1withInterDGMseq}). 
Under misspecification of both the principal score and the outcome models, all methods showed bias under certain scenarios, with larger bias observed as the number of coefficients grew. A relatively robust estimator was the OLS estimator which followed matching without replacements on Mahalanobis distance with a caliper (Figure \ref{Fig:biasS1withInterDGMseq} and Table \ref{Tab:SimRes}). 

The empirical SDs of all estimators have increased as $k$ increased, and when the functional form of the outcome model was misspecified. Under correctly-specified outcome model, the SD of the post-matching model-based estimators were comparable to the SD of  the DL estimator.
Under a misspecified outcome model, the SD of the DL estimator was larger than those of Mahalanobis post-matching estimators.

We turn to discuss simulations under $\xi>0$, namely when monotonicity does not hold. When the chosen value of $\xi_{assm}$ was correct ($\xi_{assm}=\xi$), the performance of nearly all estimators was not affected by the true $\xi$ value (Figure C1).

The post-matching model-based estimators were robust to wrongfully assumed $\xi$ values, even when the principal score model was misspecified, as long as the functional form of the outcome model was correctly specified. In this case, the DL estimator was more sensitive to wrong $\xi$ values, with some bias observed for larger $\xi$ values (Figure C2). 

Under a misspecified outcome model and correctly-specified principal score model (but wrong $\xi_{assm}$), the weighting and matching based on $\widehat{\widetilde{\pi}}^1_{as}(\bx_0)$ were generally less biased than estimators that followed matching on Mahalanobis (Figure C2). When both models were misspecified, results were qualitatively similar to the case where $\xi_{assum}=\xi$ (Figure C1 and Figure C2). That taking $\xi_{assm} \ne \xi$ impact only methods that use $\widehat{\widetilde{\pi}}^1_{as}(\bx_0)$ is not surprising, as methods that do not use the principal scores at all do not specify $\xi$. 

\begin{figure}
\centering
\caption{\footnotesize{Bias of different estimators, when monotonicity holds, under correctly specified models (top left), misspecified principal score model (top right), misspecified outcome model (bottom  left), and  both models misspecified (bottom right). %True outcome model included $A$-$\bX_0$ interactions.
% %Results described for matching on
Matching was on the Mahalanobis distance without (Mahal) or with a caliper (Mahal caliper), or on $\widehat{\widetilde{\pi}}^1_{as}(\bx_0)$ (PS).
WLS: weighted least squares;
OLS: ordinary least squares; 
Crude: Crude mean difference.
Weighting (DL): model-based weighting estimator of \cite{ding2017principal}.
%The WLS estimators and the Crude estimator followed matching with replacement. The OLS estimator followed matching without replacement.
True SACE parameter ranged between $2.49-3.40$ for $k=3$, $5.00-5.77$ for $k=5$, and $7.92-8.93$ for $k=10$ under the first outcome model (with the original covariates), and between $-0.01-1.6$ for $k=3$, $10.91-12$ for $k=5$, and $13.72-15.08$ for $k=10$ under the second outcome model (with the transformed covariates).
\label{Fig:biasS1withInterDGMseq}}} \begin{minipage}{.45\textwidth}
\includegraphics[scale=0.3]{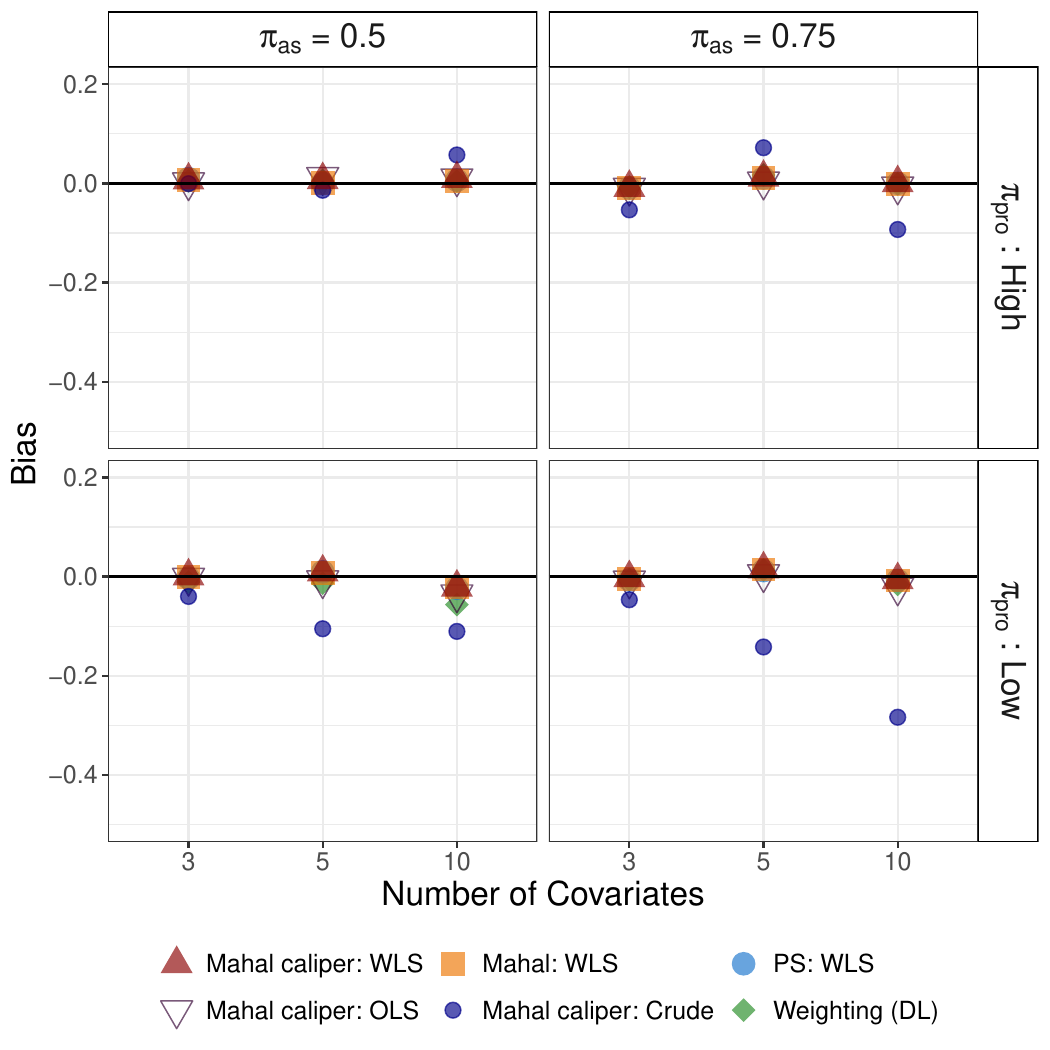}
% general correct wout inter
\end{minipage}%
\begin{minipage}{.45\textwidth}
\centering
\includegraphics[scale=0.3]{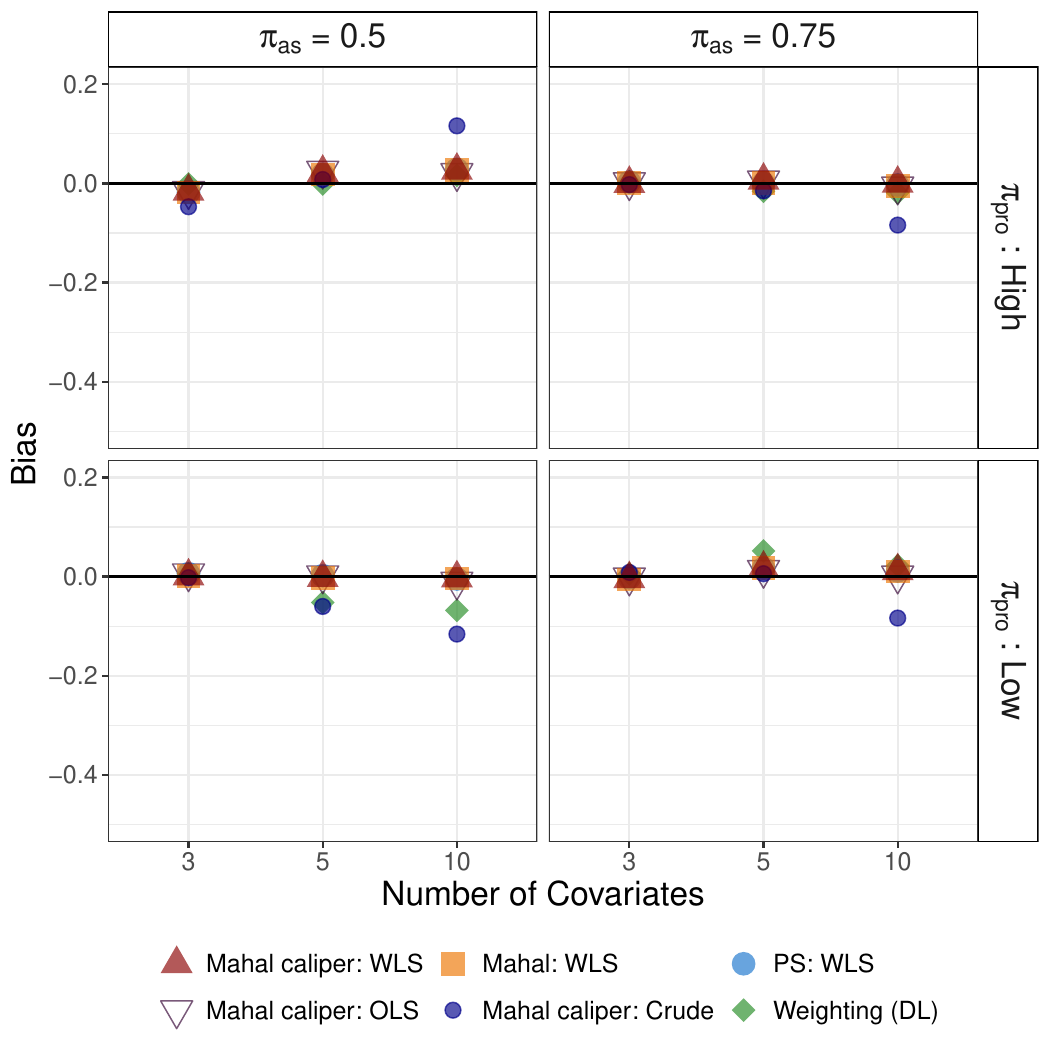}
\end{minipage}

\begin{minipage}{.45\textwidth}
\includegraphics[scale=0.3]{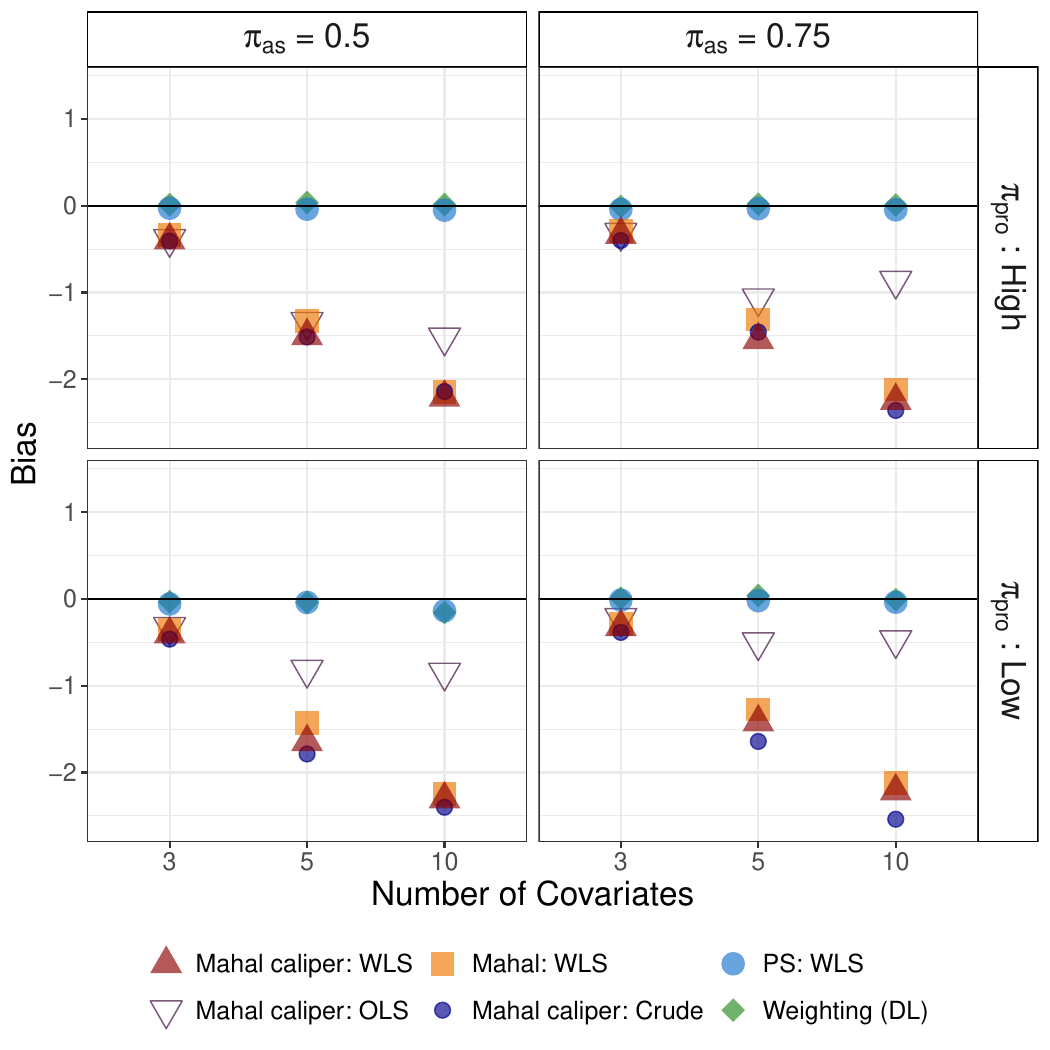}
% general correct wout inter
\end{minipage}%
\begin{minipage}{.45\textwidth}
\centering
\includegraphics[scale=0.3]{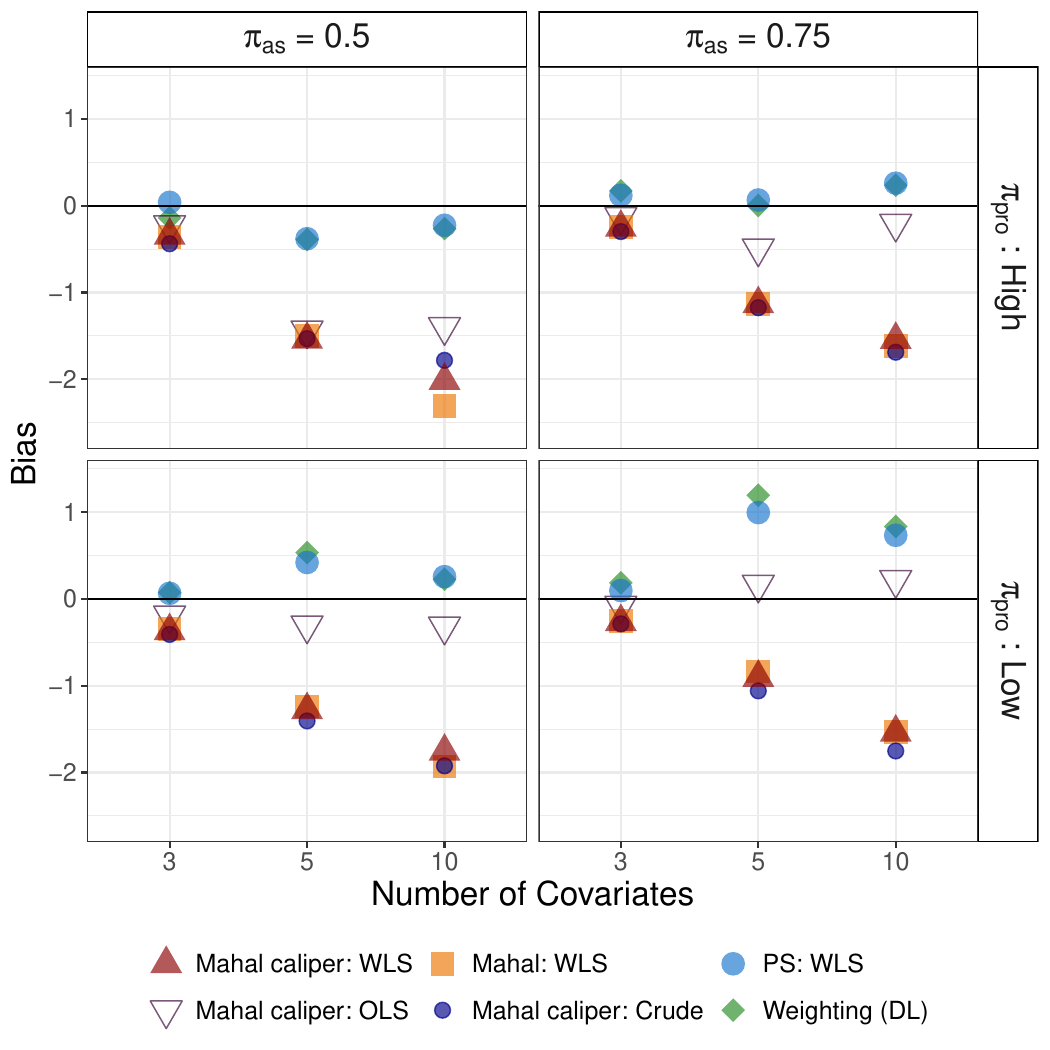}
\end{minipage}
\end{figure}

\begin{table}
\caption{\label{Tab:SimRes} \footnotesize Selected simulation results under monotonicity, low $\pi_{pro}$ and $k=5$ covariates.
Results presented for matching on $\widehat{\widetilde{\pi}}^1_{as}(\bx_0)$ (PS) or using Mahalanobis distance with a caliper (cal).
OLS: ordinary least squares; WLS: weighted least squares;
BC: bias-corrected estimator;
DL: model-based weighting estimator; Emp.SD: empirical standard deviation. Est.SE: estimated standard error; MSE: mean square error; CP95: empirical coverage proportion of 95\% confidence interval.}
\centering
\fbox{
\scriptsize
\begin{tabular}{lccccccccccc}
& \multicolumn{5}{c}{Correctly specified principal score and outcome models} & \multicolumn{5}{c}{Misspecified principal score and outcome models} \\
\em Method & \em Estimator &  \em Mean & \em Emp.SD & \em Est.SE & \em MSE & \em CP95 & Mean & Emp.SD & Est.SE & MSE & CP95 \\ \hline
& & & & & & & & & \\
& \multicolumn{5}{c}{\textbf{Scenario A, SACE = 5}} & \multicolumn{5}{c}{\textbf{Scenario A, SACE = 10.91}} \\
\hline
Matching 
& Crude:cal & 4.92 & 0.23 & 0.20 & 0.06 & 0.90 & 10.16 & 0.54 & 0.44 & 0.85 & 0.54 \\ 
& Crude:PS & 4.90 & 0.30 & 0.28 & 0.10 & 0.90 & 10.74 & 0.67 & 0.69 & 0.48 & 0.95 \\ 
& OLS:cal & 4.99 & 0.17 & 0.18 & 0.03 & 0.97 & 10.61 & 0.57 & 0.42 & 0.41 & 0.76 \\
& OLS:PS & 5.00 & 0.17 & 0.18 & 0.03 & 0.96 & 11.36 & 0.56 & 0.53 & 0.51 & 0.88 \\ 
%& OLS:cal & 4.97 & 0.15 & 0.17 & 0.02 & 0.97 & 10.50 & 0.53 & 0.40 & 0.46 & 0.70 \\ 
Matching 
& Crude:cal & 4.89 & 0.19 & 0.18 & 0.05 & 0.90 & 9.51 & 0.38 & 0.34 & 2.12 & 0.04 \\ 
with & Crude:PS & 5.03 & 0.33 & 0.27 & 0.11 & 0.90 & 10.71 & 0.93 & 0.68 & 0.90 & 0.83 \\ 
replacement & WLS:cal & 5.00 & 0.18 & 0.18 & 0.03 & 0.96 & 9.63 & 0.35 & 0.44 & 1.76 & 0.15 \\ 
& WLS:PS & 5.01 & 0.18 & 0.18 & 0.03 & 0.95 & 11.33 & 0.71 & 0.68 & 0.68 & 0.94 \\ 
& BC:cal & 5.00 & 0.18 & 0.22 & 0.03 & 0.98 & 9.63 & 0.35 & 0.41 & 1.76 & 0.11 \\ 
\multicolumn{2}{c}{Composite}  & 4.93 & 0.64 & 0.64 & 0.41 & 0.95 & 10.76 & 0.88 & 0.88 & 0.80 & 0.94 \\ 
\multicolumn{2}{c}{Naive}  
& 3.94 & 0.42 & 0.41 & 1.28 & 0.27 & 10.98 & 0.67 & 0.69 & 0.45 & 0.96 \\ 
\multicolumn{2}{c}{DL}  & 4.98 & 0.15 &  & 0.02 &  & 11.45 & 0.56 &  & 0.60 &  \\ 
& \multicolumn{5}{c}{\textbf{Scenario B, SACE = 5.77}} & \multicolumn{5}{c}{\textbf{Scenario B, SACE = 11.38}} \\
 \hline
& Crude:cal & 5.70 & 0.22 & 0.16 & 0.05 & 0.83 & 11.37 & 0.57 & 0.39 & 0.33 & 0.81 \\ 
& Crude:PS & 5.64 & 0.32 & 0.21 & 0.12 & 0.80 & 12.44 & 0.61 & 0.56 & 1.49 & 0.52 \\  
& OLS:cal & 5.78 & 0.14 & 0.14 & 0.02 & 0.96 & 11.55 & 0.52 & 0.36 & 0.29 & 0.82 \\
& OLS:PS & 5.79 & 0.14 & 0.14 & 0.02 & 0.96 & 12.41 & 0.46 & 0.45 & 1.27 & 0.36 \\ 
%& OLS:cal & 5.76 & 0.12 & 0.13 & 0.01 & 0.97 & 11.50 & 0.50 & 0.36 & 0.27 & 0.83 \\
Matching 
& Crude:cal & 5.63 & 0.16 & 0.14 & 0.04 & 0.82 & 10.32 & 0.33 & 0.30 & 1.23 & 0.10 \\
with & Crude:PS & 5.78 & 0.25 & 0.20 & 0.06 & 0.90 & 12.51 & 0.79 & 0.55 & 1.90 & 0.46 \\ 
replacement & WLS:cal & 5.79 & 0.15 & 0.15 & 0.02 & 0.95 & 10.47 & 0.32 & 0.37 & 0.92 & 0.31 \\
& WLS:PS & 5.78 & 0.14 & 0.15 & 0.02 & 0.95 & 12.38 & 0.56 & 0.56 & 1.31 & 0.59 \\ 
& BC:cal & 5.79 & 0.15 & 0.17 & 0.02 & 0.98 & 10.47 & 0.32 & 0.36 & 0.92 & 0.28 \\
\multicolumn{2}{c}{Composite}
& 6.56 & 0.56 & 0.57 & 0.93 & 0.73 & 15.96 & 0.78 & 0.76 & 21.55 & 0.00 \\ 
\multicolumn{2}{c}{Naive}  
& 4.50 & 0.33 & 0.33 & 1.73 & 0.02 & 11.80 & 0.59 & 0.57 & 0.52 & 0.88 \\
\multicolumn{2}{c}{DL}  
& 5.79 & 0.12 &  & 0.01 &  & 12.57 & 0.44 &  & 1.62 &  \\
\end{tabular}}
\end{table}

The empirical coverage rates of the $95\%$ confidence intervals of the regression post-matching estimators were close the desirable level, and were robust to wrong $\xi$ values, as long as the outcome model was correctly specified.
When the outcome model was misspecified, the coverage rates of the $95\%$ confidence intervals for estimators that followed matching  on Mahalanobis distance were dramatically lower when $k=5,10$.
In these situations, coverage rates of the regression estimators after matching on $\widehat{\widetilde{\pi}}^1_{as}(\bx_0)$ were noticeably closer to the desirable level than other matching estimators.

When the true outcome model did not include $A-\bX_0$ interactions (Figures C3--C5), the results were qualitatively similar, with the absolute bias being typically lower.
\section{Illustrative data example}
\label{Sec:Data}

The National Supported Work (NSW) Demonstration was an employment training program, carried out in the United States during the early 1970s. The NSW aimed to provide work experience for disadvantaged workers and to help them acquire capabilities required at the labor market. Eligible candidates were randomized to participate in the program or to not receive assistance from the program. Further details about the program can be found elsewhere \citep{lalonde1986evaluating,dehejia1999causal}. Descriptive data, and additional details and results are given in Section D of the Web Appendix.
The dataset consists of $722$ participants in total, of which $297$ (41\%) participated in the program and $425$ (59\%) did not. 

The outcome of interest was earnings (wage) in 1978. However, a challenge in considering the effect of the training program on earnings is that only $526$ participants ($73\%$) were employed during 1978, $230$ ($77\%$) in the treated  and $296$ ($70\%$) in the untreated. Thus, unemployment at 1978 represents ``truncation by death''. Table D14  provides the number and proportion of participants within the four values of $\{A=a,S=s\}$. The causal effect of the program on the overall earnings, setting $Y=0$ for the unemployed, quantifies the effect on earnings and employment status combined, and not on the earnings alone \citep{lee2009training}.  

A number of pre-randomization variables that are possibly shared causes (or good proxies of shared causes) of employment status and earnings in 1978 are available \citep{dehejia1999causal}. These include age, years of education, not having a high school degree, race (white/black/hispanic), marital status (married or not married), and employment status and real earnings in 1975. A description of these covariates is given in Table D15. One may argue that for SPPI to hold, measured covariates should represent pre-trial skills and education, as these are common causes of both future employment status and earnings. While years of education is available, it does not capture the totality of pre-trial education and  skills. For this reason, we consider the employment status and  earnings in 1975 to be key proxies for these covariates.  Nevertheless, it is possible that SPPI or PPI do not hold with the measured covariates, and we therefore also apply to the data the proposed sensitivity analysis.

The covariates were generally balanced due to randomization in the original sample, with the exception of years of education, having a high school degree and employment status in 1975 (Table \ref{Tab:balance}).  Looking at those employed in 1978 ($S=1$), the imbalance got more substantial. Among the employed, compared to the untreated, the treated were older, less likely to be hispanic, and more likely to be married, to own a high school degree, and to be employed in 1975. At each treatment arm, earnings and employment rates in 1975 were higher among the employed, suggesting that participants who managed to acquire a position in 1978 were possibly more skilled before entering the program.

%%%%%%%%%%%%%%%%%%%%%%%%%%%%%%%%%%%%%%%%%%%%%%%%%%%%%%%%%%%%%%%%%%%%%%%%%%%%%%%%%%%
\begin{table}
\caption{\label{Tab:balance} Covariates balance in the full, employed and matched datasets. Means (SDs) are presented for continuous covariates and frequencies (proportions) for discrete covariates, along with standardized mean differences (SMDs). Matching was carried out on Mahalanobis distance with a caliper on $\widehat{\widetilde{\pi}}^1_{as}(\bx_0)$ ($c=0.4SD$).} 
\centering
\fbox{
\scriptsize
\begin{tabular}{l|ccccccccc}
 & \multicolumn{3}{c}{Full} & \multicolumn{3}{c}{Employed} 
 & \multicolumn{3}{c}{Matched} \\
   & Untreated & Treated & SMD & Untreated & Treated & SMD & Untreated & Treated & SMD \\ 
 \hline
 \textbf{Continuous} & & & & & & & & & \\ 
 Age & 24.4 (6.6) & 24.6 (6.7) & 0.03 & 24.1 (6.6) & 24.6 (6.7) & 0.09 & 24.1 (6.6) & 23.9 (6.2) & -0.02 \\ 
 Education & 10.2 (1.6) & 10.4 (1.8) & 0.12 & 10.2 (1.6) & 10.4 (1.9) & 0.13 & 10.2 (1.6) & 10.2 (1.5) & 0.00 \\ 
 Earnings75 & 3.0 (5.2) & 3.1 (4.9) & 0.01 & 3.4 (5.7) & 3.3 (5.1) & -0.02 & 3.4 (5.7) & 3.1 (5.2) & -0.05 \\ 
 \hline
  \textbf{Discrete} & & & & & & & & & \\ 
 Black & 340 (80\%) & 238 (80\%) & 0.00 & 224 (76\%) & 176 (77\%) & 0.02 & 224 (76\%) & 221 (75\%) & -0.02 \\ 
 Hispanic & 48 (11\%) & 28 (9\%) & -0.06 & 41 (14\%) & 25 (11\%) & -0.09 & 41 (14\%) & 36 (12\%) & -0.05 \\  
 Married & 67 (16\%) & 50 (17\%) & 0.03 & 45 (15\%) & 44 (19\%) & 0.11 &
 45 (15\%) & 54 (18\%) & 0.08 \\  
Nodegree & 346 (81\%) & 217 (73\%) & -0.21 & 238 (80\%) & 167 (73\%) & -0.20 & 238 (80\%) & 241 (81\%) & 0.03 \\ 
Employed75 &  247 (58\%) & 186 (63\%) & 0.09 & 182 (61\%) & 150 (65\%) & 0.08 & 182 (61\%) & 182 (61\%) & 0.00 \\ 
\end{tabular}}
\end{table}
%%%%%%%%%%%%%%%%%%%%%%%%%%%%%%%%%%%%%%%%%%%%%%%%%%%%%%%%%%%%%%%%%%%%%%%%%%%%%%%%%%%

\subsection{Results}

The composite outcome approach yielded an estimated difference of 886  US dollars (CI95\%: -71, 1843) between those participated in the program and those who did not. However, because the program increases employment by approximately 8\% ($p = 0.01$, one-sided test), the above estimates cannot be interpreted as a causal effect on the earnings, but as an effect on employment and earnings combined.  The naive difference in the survivors was 409 (CI95\%: -690, 1508).

The always-survivors (here always-employed) stratum proportion  is identifiable under monotonicity and was estimated to be $\hat{\pi}_{as} = 0.70$, which is quite high in comparison to $\hat{\pi}_{pro}=0.08$ and $\hat{\pi}_{ns}=0.22$. Assuming the principal stratum proportions are relatively stable over time, this means the causal effect among the always-survivors stratum is of a high interest in this example, as the always-survivors comprise more than two thirds of the study population.

Turning to the matching, because we had several discrete covariates and a number of continuous covariates, we used the Mahalanobis distance with a caliper on $\hat{\widetilde{\pi}}^1_{as}(\bX_0)$. 
In our main analysis, we included the continuous covariates age, education and earnings in 1975 in the Mahalanobis distance. We used an EM algorithm to estimate the principal score model  with the continuous covariates age and earnings in 1975, and the discrete covariates black, hispanic, marital status, and employment status in 1975. The estimated sequential regression  coefficients  $\widehat{\bgamma}_{S(a)}$ for $a=0,1$
are given in Table D16. %Survivors under no treatment ($S(0)=1$) were more likely to be white or hispanic, not married and employed at 1975, and had higher earnings in 1975.
%Among the never-survivors and protected, protected were older and earned more in 1975, less likely to be employed in 1975, and more likely to be hispanic and married.

As previously noted, a key advantage of matching is that balance can be optimized without looking at the outcome data, and hence without jeopardizing the validity of the analysis. Table D17 presents balance results under several caliper values. The best balance was obtained with $c=0.4SD[\hat{\widetilde{\pi}}^1_{as}(\bX_0)]$   (Equation \eqref{Eq:CaliperDef}). In the matched sample, the covariates were generally balanced, and compared to the employed (survivors) sample, notable improvements were achieved in the balance of the covariates age, education, high-school degree and employment status in 1975 (Table \ref{Tab:balance}).
%Compared to our chosen distance measure, matching using the Mahalanobis distance without a caliper, the balance has been improved for some covariates (e.g. hispanic and \hl{earnings in 1975}), but has deteriorated for employment status in 1975. Matching only on $\hat{\widetilde{\pi}}^1_{as}(\bX_0)$ provided inferior results for most covariates \hl{(Table D3)}.

Because there were less treated than untreated employed, we used matching with replacement, resulting in matching of all untreated employed participants. The alternative  of matching  without replacement  resulted in only $223$ matches out of the $296$ ($75\%$) untreated (Table D18).

Turning to the analysis of the matched sample, beyond the crude difference, we also fitted a linear regression model (using WLS) for the outcome with the covariates age, education, black, hispanic, marital status, and earnings in 1975, and compared models with and without all treatment-covariates interactions. 
For the model with interactions, we estimated the SACE by plugging in the fitted model in \eqref{Eq:gtheta}, and then averaged the obtained quantities across the matched untreated employed. We also calculated the BC estimator, taking the same set of covariates for its regression model. Wald-type 95\% confidence intervals were calculated with  SEs estimated as described in Section \ref{Sec:Sims}.

The crude difference  in the matched dataset was 59 (CI95\%: -977, 1094). The weighted linear regression estimators were 114 (CI95\%: -1226, 1453) without interactions and 55 (CI95\%: -1269, 1379) with interactions. The BC estimator  \eqref{Eq:BCdef} was 68 (CI95\%: -1376, 1512).  The non-parametric aligned-rank test revealed no evidence for rejecting the sharp null hypothesis of no individual causal effect ($p=0.45$).
The model-assisted weighting approach using the principal scores \citep{ding2017principal} estimated the SACE to be 398 (CI95\%: -711, 1280,  using the bootstrap with 500 samples). 

We turn to explore how relaxing the assumptions by using the sensitivity analyses described in Section \ref{Sec:Sens} may affect the conclusions. We used the same matching procedure as in the main analysis. We present here results for the WLS estimator without interactions after matching on the Mahalanobis distance with a $\hat{\widetilde{\pi}}^1_{as}(\bX_0)$ caliper. %Results under different distance measures and using several estimators are presented in Section D of the Web Appendix.

We start with sensitivity analysis for PPI/SPPI under monotonicity. Using a regression model as described in Section \ref{SubSec:SensPPI}, we obtained $0.38 \le \alpha_1 \le 2.64$. Values smaller than one (larger than one) might imply the always-survivors are believed to be more (less) skilled than the protected and hence are expected to have higher (lower) earnings had they participated in the program. 

As can be seen from the left panel of Figure \ref{Fig:SA}, as $\alpha_1$ increases, the SACE decreases. For $\alpha_1 \ge 1.5$, the SACE was estimated to be negative, meaning that the program decreases the mean earnings among the always-survivors. For all $\alpha_1$ values, the estimates were insignificant ($5\%$ significance level). Taking the estimated most extreme values, bounds for the SACE under monotonicity are $(-366, 455)$.
\begin{figure}
\centering
\small\caption{Sensitivity analyses in the NSW data. The left panel presents the SACE estimate as a function of $\alpha_1$, without assuming PPI. The right panel presents SACE estimate as a function of $\xi$ and $\alpha_0$ without assuming  monotonicity.
\label{Fig:SA}
}
\begin{minipage}{.5\textwidth}
\includegraphics[scale=0.32]{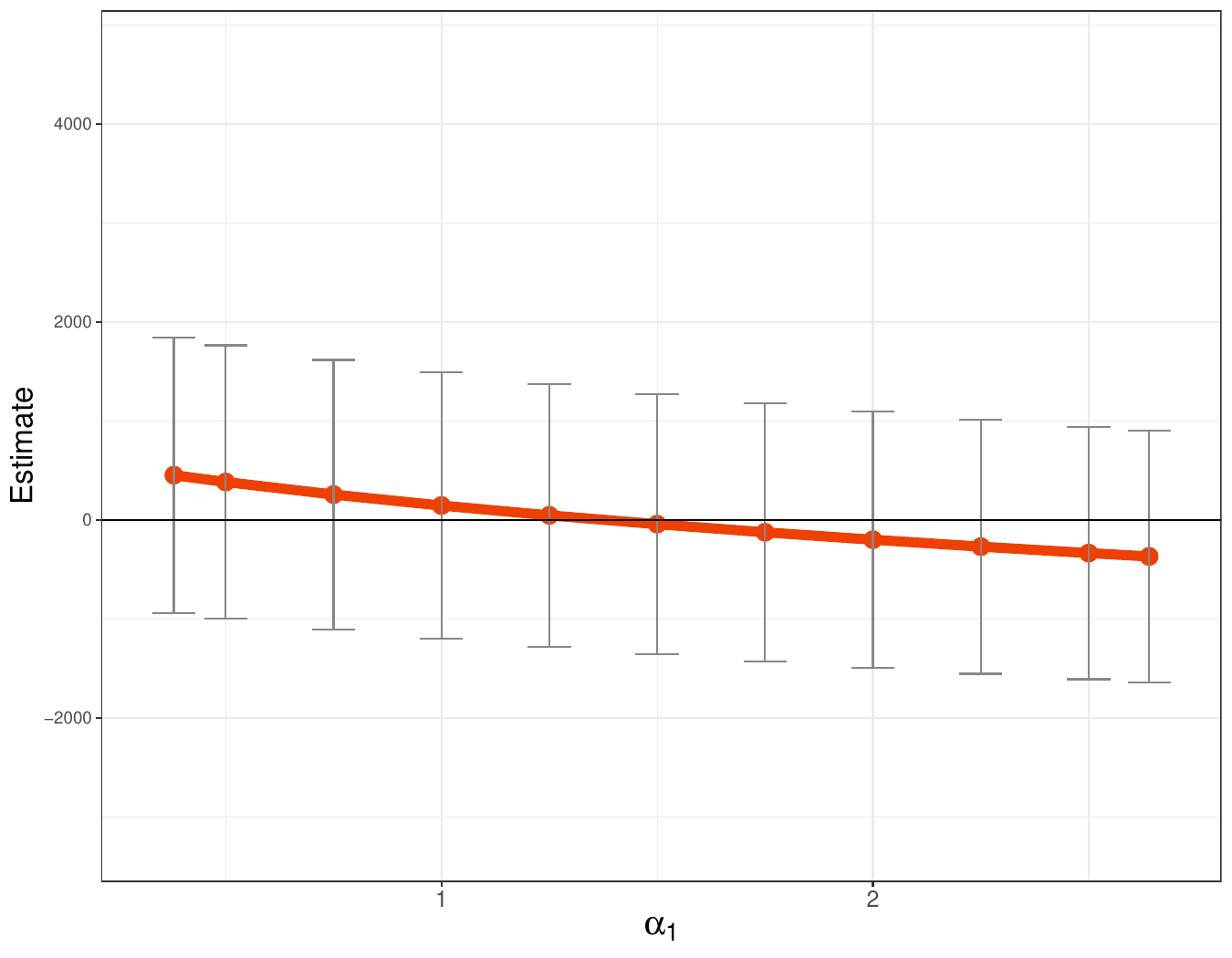}
\end{minipage}%
\begin{minipage}{.5\textwidth}
\centering
\includegraphics[scale=0.32]{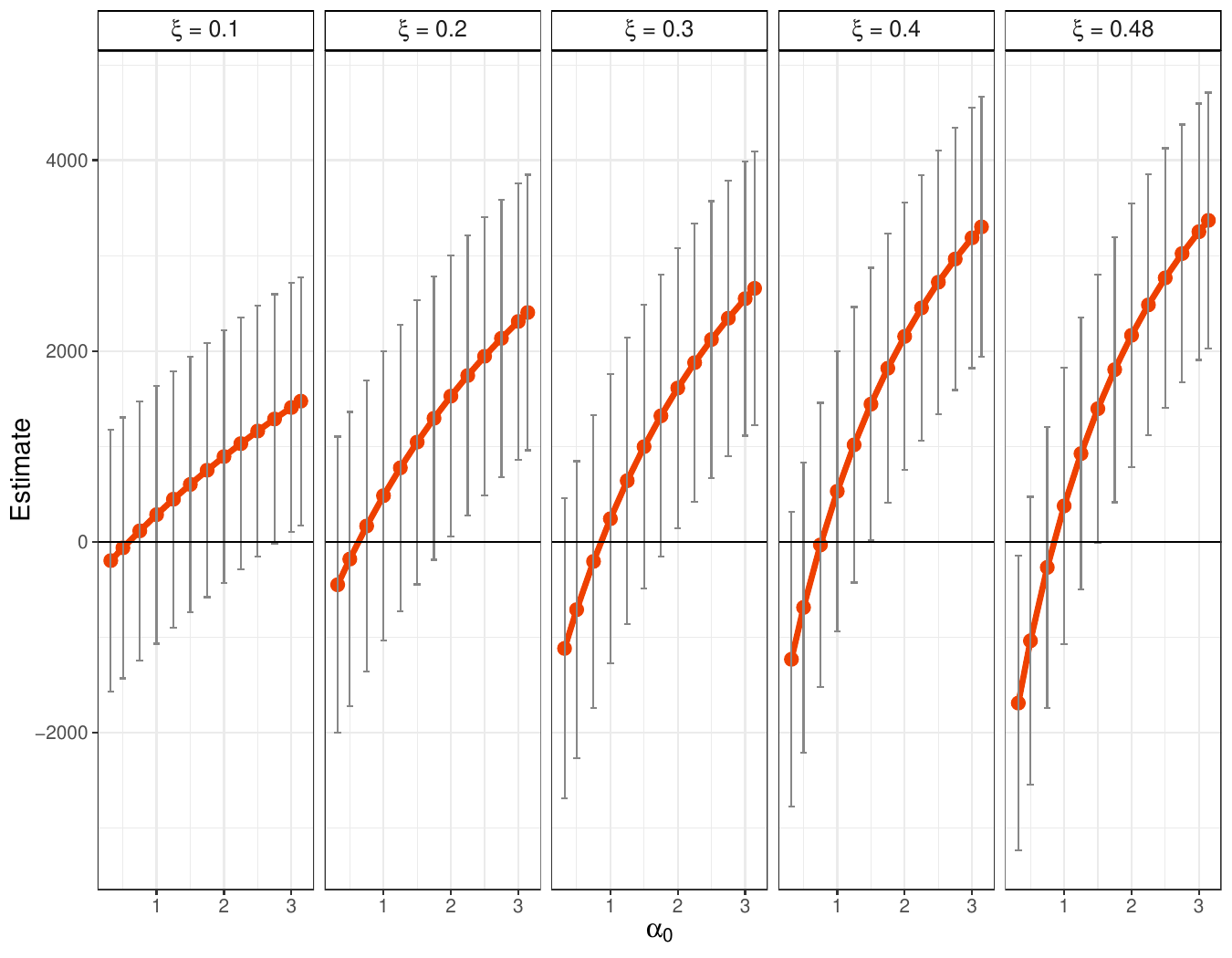}
\end{minipage}
\end{figure}

We now turn to sensitivity analysis for monotonicity (Section \ref{SubSec:Sensmono}).
A possible violation of monotonicity might be that program participants raise their reservation wages (the minimal wage they require before accepting a job) due to the program, and therefore may refuse low-earnings job offers they would have accepted had they did not participate in the program \citep{zhang2009likelihood}. Under CPSR instead of monotonicity, the principal scores are identifiable from the data as a function of $\xi$. In the NSW data,
$\hat{p}_0=0.7$ and $\hat{p}_1=0.77$, thus
$0 \le \xi \le 0.48$. 
Therefore, we considered $\xi \in \{0,0.1,0.2,0.3,0.4,0.48\}$. 
Using a regression model as described in Section \ref{SubSec:Sensmono}, we obtained  $0.32 \le \alpha_0 \le 3.14$. 

For a fixed $\xi>0$, as $\alpha_0$ increases, the SACE increases (Figure \ref{Fig:SA}).  As $\xi$ increases, changes in $\alpha_0$ have larger impact on the SACE estimate,  as expected from Proposition \ref{Prop:SAmono}. For each $\xi$ value, the positive effect of the program becomes significant for large enough $\alpha_0$ value. For example, for $\xi=0.4$, meaning there are 2.5 times more always-survivors than harmed  at each level of $\bX_0$, the SACE is significantly different from zero if the ratio between the mean outcomes of the harmed and the always-survivors when untreated is at least 1.5. Taking the estimated most extreme values, bounds for the SACE under CPSR and PPI are $(-1693, 3369)$. Sensitivity analyses under different distance measures and using several estimators were overall similar (Figures D7 and D8). 

To summarize the conclusions from the data analysis, under monotonicity and PPI (or CPSR and SPPI), no significant causal effect was found on the earnings among those would have been employed regardless of the program participation. Sensitivity analyses revealed that under CPSR and PPI, a positive effect is possible if one believes that the proportion of the harmed stratum is non-negligible and $\alpha_0$ is large enough.

\section{Discussion}
\label{Sec:Discuss}

In this paper, we presented a matching-based approach for estimation of well-defined causal effects in the presence of truncation by death.  Underpinning our approach is that the balance created by randomization and lost due to differential survival can be retrieved by a matching procedure achieving the three properties: survivors only, closeness, and distribution preservation. 

To fix ideas, we focused in this paper on 1:1 matching with or without replacement, and on classical matching distance measures. In practice, 1:k matching can also be used. Other distance measures or newly-developed matching procedures are also possible to use, and might be more attractive, especially when rich data are available and $\bX_0$ is high-dimensional. Regardless of the quality of balance achieved on the observed covariates, researchers should be aware that the identifying assumptions for the SACE are strong and cannot be falsified. Hence, the analysis should be accompanied with the proposed sensitivity analyses. Nevertheless, the methods in this paper assume that all covariates in $\bX_0$ were measured, and without an error. The sensitivity analysis assumes that $\alpha_1$, for example, is not a function of $\bX_0$. If this is not the case, bias is expected. Alternatively, $\alpha_1$ can be replaced by a function $\alpha_1(\bx_0)$ specified by the researchers. Specifying such a function, however, might be challenging in practice.
 
A number of issues concerning matching methods should be highlighted. First, we focused on designs with point-treatment and without loss to follow up. The generalization of matching methods for time-varying treatment is limited compared to other methods, see e.g. \cite{thomas2020matching} for a review of matching methods for time-varying treatments when the goal is to study a static treatment regime. A second issue is that standard matching methods typically do not achieve the semi-parametric efficiency bound \citep{abadie2006large} and can be asymptotically biased \citep{abadie2011bias}. The theoretical properties of matching-based methods are understudied compared to other methods, although they are actively studied  \citep{abadie2022robust}.

Nevertheless, matching methods are one of the most popular tools in the causal inference toolbox of practitioners. We reviewed in this paper the theoretical basis and discussed implementation details for adapting matching methods to adequately overcome truncation by death. 

\bibliographystyle{chicago}
\bibliography{MatchingSACE}

This web appendix includes four major sections. Section \ref{Sec:AdditionalTheory} presents theory and proofs. Section \ref{Sec:EMappendix} provides details on the EM algorithms used for principal scores estimation. Section \ref{Sec:AppSims} includes further details on and results from the simulation study. The final Section \ref{Sec:AppData} presents additional information and results from the NSW data analysis.

\appendix

\appendix

\renewcommand\theequation{\thesection\arabic{equation}}
\renewcommand\thefigure{\thesection\arabic{figure}}
\renewcommand\thetable{\thesection\arabic{table}}

\setcounter{equation}{0}
\setcounter{figure}{0}
\setcounter{table}{0}

\section*{Appendix}
\section{Additional theory and proofs}	
\label{Sec:AdditionalTheory}	

In Section \ref{APPSubsec:Violation of PPI and SPPI} we present a data generating mechanism (DGM), where PPI and SPPI might be violated due to the existence of common causes of survival and the non-survival outcome that are also affected by the treatment.
In Section \ref{APPSubsec:Proofs_further_theory}, we provide proofs for propositions 1-4 from the main text, and present identification formula and proof when randomization is replaced by conditional exchangeability.
In Section \ref{APPSubsec:parametersrequired_identification}, we summarizes the sensitivity parameters required for SACE and principal scores identification under different assumption combinations.
In Section \ref{APPSubsec:bounds}, we provide derive the bounds for the sensitivity parameters.
In Section \ref{SubSec:ProofCSEs}, we provide a proof of Proposition 5, namely a proof for the
alternative presentation of the identification formula of the conditional separable effect $CSE(a_S)$.

\subsection{Violation of PPI and SPPI due to $\bX_{1}$}
\label{APPSubsec:Violation of PPI and SPPI}
We present a DGM under which both PPI and SPPI might be violated because of $\bX_{1}$ being affected by the treatment. For simplicity, we consider both $X_0$ and $X_1$ to be univariate. We then show that under further restrictions, PPI holds but SPPI does not. Consider the following non-parametric structural equation model with independent errors (NPSEM-IE)
\begin{align*}
    X_0 &= \epsilon_{0},\\
    A &= \epsilon_{A},\\
    X_1 &= f_1(A,\epsilon_{1}),\\
    S &= f_S(X_0,A,X_1,\epsilon_{S}),\\
    Y &= f_Y(X_0,A,X_1,S,\epsilon_{Y}),
\end{align*}
where $\epsilon_{0}, \epsilon_{A}, \epsilon_{1}, \epsilon_{S},\epsilon_{Y}$ are independent, $A$ and $S$ are binary and where due to truncation by death $f_Y(x_0,a,x_1,0,\epsilon_{Y})= ^*$ for all $x_0,a,x_1,\epsilon_{Y}$ values. Under this NPSEM-IE we have that 
\begin{align}
\label{AppEq:NPSEM}
\begin{split}
    Y(a) &= f_Y(X_0,a,X_1(a),S(a),\epsilon_{Y}) \\
    & = f_Y(X_0,a,f_1(a,\epsilon_{1}),S(a),\epsilon_{Y}),\\
    S(a) &= f_S(X_0, a, X_1(a), \epsilon_{S}) \\
    & = f_S(X_0,a,f_1(a,\epsilon_{1}),\epsilon_{S}).
\end{split}
\end{align}
Without further assumptions, PPI (and hence SPPI) are not imposed by this model, because even conditionally on $\{X_0, S(1)=1\}$,  we may have that $Y(1) \cancel{\indep} S(0) | X_0, S(1)=1$. This is because $Y(1)$ contains information about $f_1(1,\epsilon_{1})$, which in turn may contain information on $f_1(0,\epsilon_{1})$ and hence on $S(0)$. 

PPI will hold but not SPPI under this NPSEM-IE, for example, if we add the assumption that $X_1$ affects Y directly only when untreated (that is, the treatment inactivates the effect of $X_1$ on $Y$). Formally this means that $f_Y$ from equation \eqref{AppEq:NPSEM} can be written as 
\begin{align*}
Y(0) &= f_{Y,0} = f_Y(X_0,0,X_1(0),S(0),\epsilon_{Y}), \\
Y(1) &= f_{Y,1} = f_Y(X_0,1,S(1),\epsilon_{Y}).
\end{align*}

\subsection{Proofs and further theory}
\label{APPSubsec:Proofs_further_theory}
Throughout the proofs in this section, we denote 
$f(\cdot|\mathcal{Q})$ for either a probability density function, or a joint distribution function (including of a continuous and a binary variable) conditioned on the event $\mathcal{Q}$. For simplicity of presentation only, we assume in all proofs below that all components of $\bX_0$ are continuous.

Before presenting the auxiliary lemmas and the proofs, we define two quantities that will be used. Let $\widetilde{\pi}^{1}_{as} = \frac{\pi_{as}} {\pi_{as} + \pi_{pro}}$ and $\widetilde{\pi}^{0}_{as} = \frac{\pi_{as}} {\pi_{as} + \pi_{har}}$ 
 be the analogue marginal (with respect to $\bX_0$) quantities of $\widetilde{\pi}^{1}_{as}(\bx_0)$ and $\widetilde{\pi}^{0}_{as}(\bx_0)$.

\subsubsection{Auxiliary lemmas}
\label{SubSec:ProofSACEidenAuxLemmas}

%%%%%%%%%%%%%%%%%%%%%%%%%%%%%%%%%%%%%%%%%%%%%%%%%%%%%%%
%%%%%%%%%%%%%%%%%%%%%% Lemma PiTilde %%%%%%%%%%%%%%%%%%%%%%%%
%%%%%%%%%%%%%%%%%%%%%%%%%%%%%%%%%%%%%%%%%%%%%%%%%%%%%%%
 Our first lemma states that under randomization and SUTVA, the proportion of always-survivors among the group $\{A=a,S=1,\bX_0=\bx_0\}$ is the relative size of $\pi_{as}(\bx_0)$ out of the proportion $\Pr[S(a)=1|\bX_0=\bx_0]$.

\begin{lemma} \label{Lem:PiTildeMeaning}
Under randomization and SUTVA 
\begin{align*}
\Pr(G = as|A = 1, S = 1, \bX_0 = \bx_0) &= \widetilde{\pi}^{1}_{as}(\bx_0), \\
%\Pr(G=as|A=1,S=1) &= \widetilde{\pi}^{1}_{as} \\
 \Pr(G = as|A = 0, S = 1, \bX_0 = \bx_0) &= 
\widetilde{\pi}^{0}_{as}(\bx_0).
%\Pr(G=as|A=0,S=1) &= \widetilde{\pi}^{0}_{as}
\end{align*}
\end{lemma}
\begin{proof}
By SUTVA and randomization, we have
\begin{align*}
%\begin{split}
\label{first} 
\Pr(G = as|A = 1, S = 1, \bX_0 = \bx_0) &=
\Pr[G = as|S(1) = 1, \bX_0 = \bx_0] \\
&=
\frac{\Pr(G = as| \bX_0 = \bx_0)} {\Pr[S(1) = 1| \bX_0 = \bx_0]}\\
&=
\frac{\pi_{as}(\bx_0)} {\pi_{as}(\bx_0) + \pi_{pro}(\bx_0)}\\
&=
\widetilde{\pi}^{1}_{as}(\bx_0).
%\end{split}
\end{align*}
 The proof for the  other claim about $\widetilde{\pi}^{0}_{as}(\bx_0)$ is analogous. 
\end{proof}
 \hfill $\blacksquare$\\
 An immediate corollary of this Lemma is that the results also hold without conditioning on $\bX_0$. That is,
 \begin{align*}
\Pr(G = as|A = 1, S = 1) &= \widetilde{\pi}^{1}_{as}, \\
 \Pr(G = as|A = 0, S = 1) &= 
\widetilde{\pi}^{0}_{as}.
\end{align*}
%%%%%%%%%%%%%%%%%%%%%%%%%%%%%%%%%%%%%%%%%%%%%%%%%%%%%%
%%%%%%%%%%%%%%% Lemma Xdist Mono/PPI %%%%%%%%%%%%%%%%%
%%%%%%%%%%%%%%%%%%%%%%%%%%%%%%%%%%%%%%%%%%%%%%%%%%%%%%
The following Lemma would be useful for our proofs, and is also interesting by itself. It states that under randomization, SUTVA and CPSR, the distribution of the covariates $\bX_0$ in the untreated survivors $\{A=0,S=1\}$ is the same as the distribution of $\bX_0$ in the always-survivors $\{S(0)=0,S(1)=1\}$. Since under monotonicity $\xi=0$, this result also holds if we replace CPSR with monotonicity.
\begin{lemma} \label{Lem:XdistASnomono}
Under randomization, SUTVA and CPSR
	\begin{equation*}
	f_{\bX_0}[\bx_0|S(1) = 1, S(0) = 1] = f_{\bX_0}(\bx_0|A = 0, S = 1).
	\end{equation*}	
\end{lemma}
\begin{proof}
First, note that $\xi$ is also the ratio between the marginal stratum proportions. That is,
\begin{equation}
\label{Eq:xiMarg}
\frac{\pi_{har}}{\pi_{as}} = \frac{\int_{\bx_0} f_{\bX_0}(\bx_0) \pi_{har}(\bx_0) d\bx_0} {\int_{\bx_0} f_{\bX_0}(\bx_0) \pi_{as}(\bx_0)d\bx_0} 
= \frac{\int_{\bx_0} f_{\bX_0}(\bx_0)\xi\pi_{as}(\bx_0)d\bx_0}{\int_{\bx_0} f_{\bX_0}(\bx_0) \pi_{as}(\bx_0)d\bx_0} = \xi. 
\end{equation}
Next, by Bayes' theorem and CPSR, the distribution of $\bX_0$ in the harmed and in the always-survivors strata is the same,
\begin{align}
\begin{split}
\label{Eq:Xdist_asharm}
f_X(\bx_0|S(0)=1, S(1)=0) &= \frac{f_{\bX_0}(\bx_0) \pi_{har}(\bx_0)}{\pi_{har}} = \frac{f_{\bX_0}(\bx_0) \xi\pi_{as}(\bx_0)}{\xi \pi_{as}} = \frac{f_{\bX_0}(\bx_0) \pi_{as}(\bx_0)}{ \pi_{as}} \\
&=f_X(\bx_0|S(0)=1, S(1)=1)
\end{split}
\end{align}	 
Now, by law of total probability and Lemma \ref{Lem:PiTildeMeaning} we may write
\begin{align}
\nonumber
& f_{\bX_0}(\bx_0|A=0, S=1) \\[1em]
\nonumber
&= \widetilde{\pi}^{0}_{as} f_{\bX_0}[\bx_0|S(0)=1, S(1)=1, A=0, S=1] \\[1em]
\nonumber
%\label{Eq:byLemma3} 
& + (1 - \widetilde{\pi}^{0}_{as}) f_{\bX_0}(\bx_0|S(0)=1, S(1)=0, A=0, S=1)\\[1em] 
\label{Eq:NoASinX}
&= \widetilde{\pi}^{0}_{as} f_{\bX_0}[\bx_0|S(0)=1, S(1)=1]
+ (1 - \widetilde{\pi}^{0}_{as}) f_{\bX_0}[\bx_0|S(0)=1, S(1)=0]\\[1em] 
&=f_{\bX_0}[\bx_0|S(0)=1, S(1)=1],
\label{Eq:FinalLineS0A1}
\end{align}
where \eqref{Eq:NoASinX} follows from SUTVA and randomization, and \eqref{Eq:FinalLineS0A1} is by \eqref{Eq:Xdist_asharm}.
\end{proof}
 \hfill $\blacksquare$

\subsubsection{Proof of Propositions 1 and 4}
\label{SubSec:ProofSACEidentSPPI}
Proposition 1 is a special case of Proposition 4, when PPI is replaced with the stronger SPPI. Therefore, we first present the proof for Proposition 4, which does not assume SPPI, and then use the fact that under SPPI, $\alpha_0=1$, to derive the conclusion of Proposition 1.
 
First, by law of total expectation and Lemma \ref{Lem:XdistASnomono}, we may write
\begin{align}
\begin{split}
\label{Eq:SACEouter}
E[Y(a)|S(0) = 1, S(1) = 1] &= E_{\bX_0|S(0) = 1, S(1) = 1}\big\{E[Y(a)|S(1) = 1, S(0) = 1, \bX_0]\big\} \\[1em]
& = E_{\bX_0|A = 0, S = 1}\big[\mu_{a,as}(\bX_0)\big]
\end{split}
\end{align}
for $a=0,1$. Next, $\mu_{1,as}(\bX_0)$ is identified by
\begin{align} 
\begin{split}
\label{Eq:SACEm1}
\mu_{1,as}(\bX_0) & = E[Y(1)|S(1) = 1, \bX_0] \\[0.5em]
& = E[Y(1)|A = 1, S(1) = 1, \bX_0] \\[0.5em]
& = E[Y|A = 1, S = 1, \bX_0],
\end{split}
\end{align}
where the first line is by PPI, the second by randomization, and the third by SUTVA.
By setting $a=1$, and substituting (\ref{Eq:SACEm1}) in (\ref{Eq:SACEouter}), it follows that 
\begin{equation}
\label{Eq:Y1_SACE}
E[Y(1)|S(0)=1,S(1)=1] = E_{\bX_0|A=0,S=1}\left\{ E[Y|A = 1, S = 1, \bX_0]\right\}. 
\end{equation}
Next, to identify $E[Y(0)|S(0)=1,S(1)=1]$,
we turn to the identification of $\mu_{0,as}(\bX_0)$. Observe that by SUTVA, randomization, and the fact that $\bX_0$ are pre-treatment variables which are not affected by $A$ we may write
\begin{equation}
\label{Eq:MeanY0atA0S1}
E[Y(0)| G = g, A = 0, S = 1, \bX_0]=E[Y(0)| G = g, \bX_0]=\mu_{0,g}(\bX_0)
\end{equation}
for $g=as,har$. Now, 
\begin{align*}
\nonumber
&E[Y|A=0, S=1, \bX_0] = 
E[Y(0)| A=0, S=1, \bX_0] \\[0.5em] \nonumber
&= \Pr(G=as | A = 0, S = 1, \bX_0) E[Y(0)| G=as, A = 0, S = 1, \bX_0 ]\\[0.5em] \nonumber
&+
\Pr(G=har | A = 0, S = 1, \bX_0) E[Y(0)| G=har, A = 0, S = 1, \bX_0]\\[0.5em] %\label{Eq:UseLem}
&=
 \widetilde{\pi}^{0}_{as}(\bX_0) \mu_{0,as}(\bX_0)+
[1-\widetilde{\pi}^{0}_{as}(\bx_0)]\mu_{0,har}(\bX_0)\\[0.5em] %\label{Eq:UseSensMono}
&=
\big\{\widetilde{\pi}^{0}_{as}(\bX_0) + 
[1-\widetilde{\pi}^{0}_{as}(\bX_0)]\alpha_0\big\}\mu_{0,as}(\bX_0)
\end{align*}
where the first equality is by SUTVA, the second is by the law of total expectation, the third is by Lemma \ref{Lem:PiTildeMeaning} and \eqref{Eq:MeanY0atA0S1}, and  last line is by the definition of $\alpha_0$.
Reorganizing the terms, we obtain
\begin{equation}
\label{Eq:EY0as} 
\mu_{0,as}(\bX_0) = \frac{E(Y|\bX_0, A=0, S=1)}{\widetilde{\pi}^{0}_{as}(\bX_0) + 
[1-\widetilde{\pi}^{0}_{as}(\bx_0)]\alpha_0}. 
\end{equation}
Finally, setting $a=0$, and substituting \eqref{Eq:EY0as} in \eqref{Eq:SACEouter}, we get 
\begin{align}
\nonumber
E[Y(0)|S(0)=1,S(1)=1] &= E_{\bX_0|A=0,S=1}\left\{\frac{E(Y| A = 0, S = 1, \bX_0)}{\widetilde{\pi}^{0}_{as}(\bX_0) + 
[1-\widetilde{\pi}^{0}_{as}(\bx_0)]\alpha_0}\right\} \\
% \label{Eq:UseXiDef}
%  & = (1 + \xi) 
%  E_{\bX_0|A=0,S=1}\left\{\frac{E[Y|A = 1, S = 1, \bX_0]}{(1-\alpha_0) + (1 + \xi) \alpha_0}\right\} \\ 
\label{Eq:Y0_SACE}
 	& = \frac{(1 + \xi)}{(1 + \xi \alpha_0)} 
 	E_{\bX_0|A=0,S=1}[E(Y| A = 1, S = 1, \bX_0)], 
 \end{align}
where in (\ref{Eq:Y0_SACE}) we used that
$\widetilde{\pi}^{0}_{as}(\bx_0)= \frac{1}{1+\xi}$. 
To complete the proof of Proposition 4, we note that the difference between \eqref{Eq:Y1_SACE} and (\ref{Eq:Y0_SACE}) is exactly the SACE.
Proposition 1 follows by setting $\alpha_0=1$.
\hfill $\blacksquare$

%%%%%%%%%%%%%%%%%%%%%%%%%%%%%%%%%%%%%%%%%%%%%%%%%%%%%%%%%%%%%%%%%
\subsubsection{Proof of Proposition 2}
Formally, SPPI entails the two following statements:
\begin{enumerate}[1.]
    \item $Y(1) \indep S(0)|S(1)=1,\bX_0$.
    \item $Y(0) \indep S(1)|S(0)=1,\bX_0$.
\end{enumerate}
Statement 1 is PPI. Statement 2 follows from monotonicity, because $S(1)=1$ deterministically when $S(0)=1$. 
\hfill $\blacksquare$

\subsubsection{Proof of Proposition 3}
\label{SubSec:SAppi}
First, $\mu_{0,as}(\bX_0)$
is identified by
\begin{align} 
\begin{split}
\label{Eq:SACEm0_mono}
\mu_{0,as}(\bX_0) & = E[Y(0)|S(0) = 1, \bX_0] \\[0.5em]
& = E[Y(0)|A = 0, S(0) = 1, \bX_0] \\[0.5em]
& = E(Y|A = 0, S = 1, \bX_0),
\end{split}
\end{align}
where the first equality is by monotonicity, the second is by randomization, and the third by SUTVA.

To identify
$E[Y(0)|S(0)=1,S(1)=1]$, we set $a=0$, and plug
\eqref{Eq:SACEm0_mono} in
\eqref{Eq:SACEouter} (Section \ref{SubSec:ProofSACEidentSPPI}), to obtain
\begin{equation}
\label{Eq:Y0_SACE_SA_PPI}
E[Y(0)|S(0)=1,S(1)=1] = E_{\bX_0|A=0,S=1}[E(Y|A = 0, S = 1, \bX_0)]
\end{equation}

Turning to $E[Y(1)|S(0)=1,S(1)=1]$,  observe that by SUTVA, randomization, and the fact that $\bX_0$ are pre-treatment variables which are not affected by $A$ we may write
\begin{equation}
\label{Eq:MeanY1atA0S1}
E[Y(1) | G = g, A = 1, S = 1, \bX_0] = E[Y(1) |  G = g, \bX_0] = \mu_{1,g}(\bX_0)
\end{equation}
for $g=as,pro$. Now,
\begin{align*}
\nonumber
E(Y&| A = 1, S = 1, \bX_0)\\ &=
E[Y(1)| A = 1, S = 1, \bX_0] \\[0.5em] \nonumber
&= 
\Pr(G = as |A = 1, S = 1, \bX_0) E[Y(1)\ | \ G = as, A = 1, S = 1, \bX_0]\\[0.5em] \nonumber
& + \Pr(G=pro | A = 1, S = 1, \bX_0) E[Y(1)\ | \ G = pro, A = 1, S = 1, \bX_0] \\[0.5em] 
%\label{Eq:UseLemAgain}
&=\widetilde{\pi}^{1}_{as}(\bX_0)\mu_{1,as}(\bX_0) 
+ [(1-\widetilde{\pi}^{1}_{as}(\bX_0)]\mu_{1,pro}(\bX_0)\\[0.5em] 
%\label{Eq:UseSensPPI}
 &= \{\widetilde{\pi}^{1}_{as}(\bX_0) +
[1 - \widetilde{\pi}^{1}_{as}(\bX_0)]\alpha_1\} \mu_{1,as}(\bX_0),
\end{align*}
where the first equality is by SUTVA, the second is by the law of total expectation,   the third is by Lemma \ref{Lem:PiTildeMeaning} and \eqref{Eq:MeanY1atA0S1}, and the last line is by the definition of $\alpha_1$.
Reorganizing the terms, we obtain that
\begin{equation}
\label{Eq:EY1as} 	
\mu_{1,as}(\bX_0) = \frac{E(Y| A = 1, S = 1, \bX_0)}
{\widetilde{\pi}^{1}_{as}(\bX_0) + \alpha_1[1 - \widetilde{\pi}^{1}_{as}(\bX_0)]}.
\end{equation}
Finally, setting $a=1$, and substituting \eqref{Eq:EY1as} in %\eqref{Eq:LIEforY1},
%\eqref{Eq:SACEm0_mono} in
\eqref{Eq:SACEouter} (Section \ref{SubSec:ProofSACEidentSPPI}), we get 
\begin{equation}
\label{Eq:Y1_SACE_SA_PPI}
E[Y(1)|S(0)=1,S(1)=1] = E_{\bX_0|A=0,S=1}\left\{\frac{E[Y| A = 1, S = 1, \bX_0]}
{\widetilde{\pi}^{1}_{as}(\bX_0) + \alpha_1[1 - \widetilde{\pi}^{1}_{as}(\bX_0)]}\right\}
\end{equation}
To complete the proof, we note that the difference between \eqref{Eq:Y1_SACE_SA_PPI} and \eqref{Eq:Y0_SACE_SA_PPI} is exactly the SACE.
\hfill $\blacksquare$
%%%%%%%%%%%%%%%%%%%%%%%%%%%%%%%%%%%%%%%%%%%%%%%%%%%%%%%%%%%%%%%%%

% \subsubsection{Proposition 1 under conditional exchangeability (CE)}
\subsubsection{SACE identification under conditional exchangeability}
\label{SubSec:SACEidentSPPI_CE}

In this Section we show that the SACE can be identified when replacing the randomization assumption by a weaker assumption, namely, the following conditional exchangeability (CE) assumption. We focus on the case the same covariates $\bX_0$ are needed for both SPPI and CE to hold.
\setcounter{assumption}{5}
\begin{assumption}
	Conditional exchangeability. $A \indep \{S(a),Y(a)\} | \bX_0$ for $a=0,1$.
\end{assumption}

As with the case of randomization, we start with a Lemma identifying the distribution of $\bX_0$ among the always survivors. The following Lemma is the analogue of Lemma \ref{Lem:XdistASnomono}, under CE instead of randomization. 
 \begin{lemma}
\label{Lem:XdistASnomono_CE}
Under SUTVA, CE, and CPSR 
\begin{equation*}
f_{\bX_0}[\bx_0|S(0)=1, S(1)=1] =\tilde{f}_{\bX_0}(\bx_0)
\end{equation*}
where $\tilde{f}_{\bX_0}(\bx_0)$ denotes the  probability density function
\begin{equation*}
\tilde{f}_{\bX_0}(\bx_0) = \dfrac{\Pr(S=1|A=0, \bX_0 = \bx_0)}{E_{\bX_0}\big[\Pr(S=1|A=0, \bX_0)\big]} f_{\bX_0}(\bx_0). \end{equation*}
\end{lemma}
\begin{proof}
First, by Bayes' theorem
\begin{align}
\label{X|as_bayes}
f_{\bX_0}(\bx_0|S(0)=1, S(1)=1) &=
\frac{\pi_{as}(\bx_0) f_{\bX_0}(\bx_0)}{\pi_{as}}.
\end{align}

Starting with $\pi_{as}(\bx_0)$ and $\pi_{as}$, we have

\begin{equation}
\label{pi_as_x_LTP_CE}
\pi_{as}(\bx_0) = \Pr(S(0)=1|\bX_0=\bx_0)
\Pr(S(1)=1|S(0)=1, \bX_0 = \bx_0).
\end{equation}
By CE and SUTVA, we obtain
\begin{equation}
\label{S0_identification_CE}
\Pr(S(0)=1|\bX_0 = \bx_0) = \Pr(S=1|A=0, \bX_0 = \bx_0),
\end{equation}
and
\begin{align}
\begin{split}
\label{as_in_S0_identification_CE_CE}
 \Pr(S(1)=1|S(0)=1, \bX_0 = \bx_0) &= \frac{\Pr(S(0)=1, S(1)=1| \bX_0 = \bx_0)}{\Pr(S(0)=1|\bX_0 = \bx_0)} \\
& =
\frac{\pi_{as}(\bx_0)} {\pi_{as}(\bx_0) + \pi_{pro}(\bx_0)} 
\\
&= \frac{1}{1+\xi}.
\end{split}
\end{align}
Substituting (\ref{as_in_S0_identification_CE_CE}) and (\ref{S0_identification_CE}) in 
(\ref{pi_as_x_LTP_CE}), we obtain
\begin{align}
\label{pi_as_x_identification_CE}
\pi_{as}(\bx_0) &=  \frac{1}{1+\xi} \Pr(S=1|A=0, \bX_0 = \bx_0). 
\end{align}
By the law of total probability, $\pi_{as}$ can be written as
\begin{align}
\begin{split}
\label{pi_as_identification_CE}
\pi_{as} &= \int_{\bx_0} f_{\bX_0}(\bx_0) \pi_{as}(\bx_0) d\bx_0 \\
&= \frac{1}{1+\xi} \int_{\bx_0} f_{\bX_0}(\bx_0) \Pr(S=1|A=0, \bX_0 = \bx_0) d\bx_0  \\
&= \frac{1}{1+\xi} E_{\bX_0}\big[\Pr(S=1|A=0, \bX_0)\big], \end{split}
\end{align}
where the second line is by \eqref{pi_as_x_identification_CE}.
To finish the proof, we plug (\ref{pi_as_x_identification_CE}) and (\ref{pi_as_identification_CE}) in (\ref{X|as_bayes}).

% Of note is that $f_{\bX_0}(\bx_0)|S(0)=1, S(1)=1)$ can also be written as
% \begin{align}
% f_{\bX_0}(\bx_0)|S(0)=1, S(1)=1) &= \frac{P(A=0, S=1, \bX_0=\bx_0)}{P(A=0|\bX_0=\bx_0) E_{\bX_0}\big[E(S|A=0,\bX_0)\big]} \nonumber
% \end{align}
\end{proof}
\hfill $\blacksquare$
 
\setcounter{proposition}{5}
\begin{proposition}
\label{Prop:SACEident_CE}
Under SUTVA, CE, SPPI and CPSR (for any $\xi$), the SACE is identified from the observed data by 	
\begin{equation}
\label{Eq:SACEidentCE}
SACE=E_{\tilde{f}_{\bX_0}}\big[E(Y|A=1, S=1, \bX_0) - E(Y|A=0, S=1, \bX_0)\big],
\end{equation}
where $\tilde{f}_{\bX_0}(\bx_0)$ is defined in Lemma \ref{Lem:XdistASnomono_CE}.
\end{proposition}

\begin{proof}
 
First, when replacing randomization with CE,  
\eqref{Eq:SACEouter} (Section \ref{SubSec:ProofSACEidentSPPI}),
can be replaced, by Lemma \ref{Lem:XdistASnomono_CE},
with 
\begin{align}
\begin{split}
\label{Eq:SACEouter_CE}
E[Y(a)|S(0) = 1, S(1) = 1] &= E_{\bX_0|S(0) = 1, S(1) = 1}\big\{E[Y(a)|S(1) = 1, S(0) = 1, \bX_0]\big\} \\[0.5em]
& = E_{\tilde{f}_{\bX_0}}\big[\mu_{a,as}(\bX_0)\big],
\end{split}
\end{align}

for $a=0,1$.
Second, the inner expectations in \eqref{Eq:SACEouter}, i.e. $\mu_{a,as}(\bX_0)$, are identified by
\begin{align} 
\begin{split}
\label{Eq:mu_a_g_CE}
\mu_{a,as}(\bX_0) & = E[Y(a)|S(a) = 1, \bX_0] \\[0.5em]
& = E[Y(a)|A = a, S(a) = 1, \bX_0] \\[0.5em]
& = E[Y|A = a, S = 1, \bX_0],
\end{split}
\end{align}
for $a=0,1$, where the first equality is by SPPI, the second is by CE, and the third is by SUTVA. 
Finally, 
for both $a=1$ and $a=0$,
we plug (\ref{Eq:mu_a_g_CE}) in (\ref{Eq:SACEouter_CE}), and consider the contrast between the two, to obtain the identification formula presented in \eqref{Eq:SACEidentCE}.  
\end{proof}
\hfill $\blacksquare$

\subsection{Summary of parameters required for identification}
\label{APPSubsec:parametersrequired_identification}

Table \ref{Tab:Parameters_required_for_identification} reviews identifiability of the SACE and of the principal scores, and summarizes the sensitivity parameters needed for SACE and principal scores identification under each combination of the assumptions.
For instance, under PPI and monotonicity, both the SACE and the principal scores are point-identifiable, without specifying any further parameters.  Under the weaker combination of SPPI and CPSR, the SACE is point-identifiable (Proposition 1) but the principal scores are not, unless a specific value is assumed for $\xi$. The table shows that whenever monotonicity is assumed, the principal scores are identifiable. If CPSR is assumed instead, the principal scores becomes identifiable only as a function of $\xi$. The table also presents a result not discussed in the main text. Under CPSR only, the SACE is identifiable as a function of all three sensitivity parameters. 

\begin{table}
\caption{\label{Tab:Parameters_required_for_identification} Parameters required for identification of the SACE and the principal scores, under different  assumption combinations. }
\centering
\fbox{\begin{tabular}{|l|c|c|}
 \hline
%  \\[0.05cm]
\multirow{2}{*}{Assumptions} &  \multicolumn{2}{c}{Identification parameters}   \\[0.2cm]
  & SACE & Principal scores  \\[0.05cm]
% & & SACE & principal scores & assumptions \\[0.05cm]
 \hline
 %\\[0.1cm] 
PPI $+$ monotonicity & ---  & --- \\[0.2cm] 
SPPI $+$ CPSR &  --- &  $\xi$  \\[0.2cm]
Monotonicity &  $\alpha_1$ & ---  \\[0.2cm] 
PPI $+$ CPSR &  $\xi, \alpha_0$ & $\xi$  \\[0.2cm] 
CPSR & $\xi, \alpha_0, \alpha_1$ & $\xi$ \\
\end{tabular}}
\end{table}

\newpage
\subsection{Bounds for the sensitivity parameters}
\label{APPSubsec:bounds}
\subsubsection{Bounds for $\xi$}
\label{SubSec:XiBounds}

We now show that $\xi$ is bounded by

\[
\begin{cases}
\left[\max\left(0, \: \frac{p_0 - p_1}{p_1}\right), \:
\frac{(1 - p_1)}{p_0 - (1 - p_1)}\right] &\quad p_0 + p_1 > 1, \\ \\
\left[\max\left(0, \: \frac{p_0 - p_1}{p_1}\right), \:
\infty\right) &\quad p_0 + p_1 \le 1
\end{cases}
\]
First, it can be shown that the principal stratum proportions can be expressed as a function of $p_0, p_1$ and $\xi$ by 
\begin{align*}
\pi_{as} &= \frac{1}{1 + \xi}p_0,\\ 
\pi_{har} &= \frac{\xi}{1 + \xi}p_0, \\  \pi_{pro} &= p_1 - \frac{1}{1 + \xi}  p_0, \\
\pi_{ns} &= 1 - p_1 - \frac{\xi}{1 + \xi}p_0.
\end{align*}
Note that because $p_0\in[0,1]$ and $\xi>0$, then that $\pi_{as},\pi_{har} \in [0,1]$ does not impose any restrictions on $\xi$. Turning to the next two equations, it is easy to see that for any $p_0,p_1\in[0,1]$ we have that $\pi_{pro}$ and $\pi_{ns}$ are equal to or lower than one. However $\pi_{pro}$ and $\pi_{ns}$ are non-negative only if
\begin{align}
\pi_{pro} &= p_1 - \frac{1}{1 + \xi} p_0 \geq 0 , \label{pi pro inequality} \\
\pi_{ns} &= 1 - p_1 - \frac{\xi}{1 + \xi} p_0 \geq 0. \label{pi ns inequality}
\end{align}
From \eqref{pi pro inequality} we obtain $\xi \geq \frac{p_0 - p_1}{p_1} $, which imposes a new restriction only if $p_0 > p_1$. From (\ref{pi ns inequality}) we obtain $\xi \le \frac{(1 - p_1)}{p_0 - (1 - p_1)}$, as long as $p_0 + p_1 > 1$. 
When $p_0 + p_1 \le 1$, $\pi_{as}$ might be zero, while $\pi_{har}$ is positive and hence it is not possible to obtain an upper bound for $\xi$.

\subsubsection{Bounds for $\alpha_1$ and $\alpha_0$}
\label{SubSubSec:alphasBounds}

We now show that
$\alpha_1$ can be bounded by $$
\left[\frac{\min_{\bx_0} E(Y|A=1,S=1,\bx_0)}{\max_{\bx_0} E(Y|A=1,S=1,\bx_0)}, \: \frac{\max_{\bx_0} E(Y|A=1,S=1,\bx_0)}{\min_{\bx_0} E(Y|A=1,S=1,\bx_0)}\right],
$$
and $\alpha_0$ can be bounded by $$
\left[\frac{\min_{\bx_0} E(Y|A=0,S=1,\bx_0)}{\max_{\bx_0} E(Y|A=0,S=1,\bx_0)}, \: \frac{\max_{\bx_0} E(Y|A=0,S=1,\bx_0)}{\min_{\bx_0} E(Y|A=0,S=1,\bx_0)}\right].$$
We start with  $\alpha_1$.
For $g=as,pro$, for $s=0,1$, and
for each value of $\bX_0$, denoted by $\widetilde{\bx}_0$, on one hand we have
\begin{align}
\begin{split}
\label{alpha1_lower_bound}
\mu_{1,g}(\widetilde{\bx}_0) &= 
E[Y(1)|S(1)=1, S(0)=s, \widetilde{\bx}_0] \\
&= 
E[Y|A=1, S(1)=1, S(0)=s, \widetilde{\bx}_0]\\
% \min{_\bx_0}
&\ge \min_{\bx_0} E[Y|A=1, S(1)=1, S(0)=s, \bx_0]\\
&\ge \min_{\bx_0} E[Y|A=1, S(1)=1, \bx_0]\\
&= \min_{\bx_0} E(Y|A=1,S=1,\bx_0),
\end{split}
\end{align}
where second is due to randomization and SUTVA, and the last one is due to SUTVA. On the other hand,
by similar considerations, it follows that
\begin{align}
\begin{split}
\label{alpha1_upper_bound}
\mu_{1,g}(\widetilde{\bx}_0) &=
E[Y(1)|S(1)=1, S(0)=s, \widetilde{\bx}_0]\\ &=
E[Y|A=1, S(1)=1, S(0)=s, \widetilde{\bx}_0]\\
&\le \max_{\bx_0} E[Y|A=1, S(1)=1, S(0)=s, \bx_0]\\
&\le \max_{\bx_0} E[Y|A=1, S(1)=1, \bx_0]\\
&= \max_{\bx_0} E(Y|A=1,S=1,\bx_0).
\end{split}
\end{align}
Finally, putting together \eqref{alpha1_lower_bound} and \eqref{alpha1_upper_bound}, for $g=pro,as$ it follows that for any value of $\bX_0$, denoted by $\widetilde{\bx}_0$, we may write
\begin{align*}
\frac{\min_{\bx_0} E[Y|A=1, S=1, \bx_0]}{\max_{\bx_0} E[Y|A=1, S=1, \bx_0]}
 &\le 
\frac{\mu_{1,pro}(\widetilde{\bx}_0)}{\mu_{1,as}(\widetilde{\bx}_0)} \le \frac{\max_{\bx_0} E[Y|A=1, S=1, \bx_0]}{\min_{\bx_0} E[Y|A=1, S=1, \bx_0]},
\end{align*}
which implies that $\alpha_1 = \frac{\mu_{1,pro}(\bx_0)}{\mu_{1,as}(\bx_0)}$ is bounded by $\left[\frac{\min_{\bx_0} E(Y|A=1,S=1,\bx_0)}{\max_{\bx_0} E(Y|A=1,S=1,\bx_0)}, \: \frac{\max_{\bx_0} E(Y|A=1,S=1,\bx_0)}{\min_{\bx_0} E(Y|A=1,S=1,\bx_0)}\right]$. 
The bounds for $\alpha_0$ are obtained by similar considerations,
for $g=as,har$.
%among units with $S(0)=1$.

%%%%%%%%%%%%%%%%%%%%%%%%%%%%%%%%%%%%%%%%%%%%%%%%%%%%%%%%%%%%%%%%%

\subsection{Proof of Proposition 5}
\label{SubSec:ProofCSEs}
\label{APPSubsec:Proof of Proposition 5}
In the notation of our paper, we can write the identification result of \cite{stensrud2022conditional} as
\begin{equation}
\label{Eq:CSEident}
CSE(a_S) = \int_{\bx_0,\bx_1}E(Y|A = a_Y, S=1, \bx_0, \bx_1)\frac{f(S=1, \bx_1| A = a_S, \bx_0)f(\bx_0)}{\Pr(S=1| A = a_S)}d\bx_0d\bx_{1}\\
\end{equation}
Now, 
\begin{align}
\begin{split}
\label{Eq:CSEproof}
&\frac{f(S=1, \bx_1| A = a_S, \bx_0) f(\bx_0)} {\Pr(S=1| A = a_S)} \\
&= \frac{f(\bx_0, \bx_1| A = a_S, S=1) \Pr(S=1|A=a_S)\Pr(A=a_S)f(\bx_0)}{\Pr(S=1|A=a_S)f(A=a_S, \bx_0)} \\
&= \frac{f(\bx_0, \bx_1| A = a_S, S = 1)\Pr(A = a_S) f(\bx_0)}{f(A = a_S, \bx_0)}\\
&= f(\bx_0, \bx_1| A = a_S, S=1),
\end{split}
\end{align}
where in the second line $f(S = 1, \bx_1| A = a_S, \bx_0)$ was rewritten as a joint distribution divided by the marginal and the fourth line follows by randomization because $f(A=a_S, \bx_0)=\Pr(A=a_S)f(\bx_0)$. Plugging \eqref{Eq:CSEproof} in \eqref{Eq:CSEident} we get Proposition 5.

%%%%%%%%%%%%%%%%%%%%%%%%%%%%%%%%%%%%%%%%%%%%%%%%%%%%%%%%%%%%%%%%%

\section{EM algorithms} 
\label{Sec:EMappendix}

In this section, we present the EM algorithms we have used for the estimation of the principal scores. 
We first present the EM algorithm for estimating the sequential logistic regression model.  
This estimation process requires a value for $\xi$, which can also be zero if the monotonicity assumption is imposed. 
We then describe the multinomial logistic regression model DGM for the principal strata. For this DGM, we give the EM algorithm developed by \cite{ding2017principal} under monotonicity, and then provide the details of needed modifications under CPSR.

Towards the descriptions of the algorithms we define $\bQ_i = (\bZ_{i},A_i,S_i)$, where $\bZ_i=(1, \bX_{0i})$, is the available data for unit $i$.

%%%%%%%%%%%%%%%%%%%%%%%%%%%%%%%%%%%%%%%%%%%%%%%%%%%%%%%%%%%%%%%%%
\subsection{EM algorithm, sequential logistic regression model} 
\label{App: sec:EM_DGM_seq}

With a slight abuse of notations,
let $\gamma_{0,S(a)}$ and $\bgamma_{S(a)}, a=0,1$, denote the intercepts and coefficient vectors, respectively, of the two logistic regressions. The E-step and M-step are given below. Before presenting them, we repeat the definition of the DGM and review a number of useful quantities.

Recall that under the sequential regression model, the following logistic regression model is assumed for $S(0)$.
\begin{equation}
\label{Eq:EMFirstStepSeqRegS0}
\Pr(S(0)=1|\bx_{0i}) = \frac{\exp(\gamma_{0,S(0)} + \bgamma_{S(0)}^T \bx_{0i})} 
{1 + \exp(\gamma_{0,S(0)} + \bgamma_{S(0)}^T \bx_{0i})},
\end{equation}
and $\Pr(S_i(0)=0|\bx_{0i}) = 1 - \Pr(S_i(0)=1|\bx_{0i})$. Note that $\Pr(S_i(0)=1|\bx_{0i}) = \Pr(G \in \{as,har\}|\bx_{0i})$, and
$\Pr(S_i(0)=0|\bx_{0i}) = \Pr(G \in \{ns,pro\}|\bx_{0i})$.
The second part of the model is a logistic regression model for $S_i(1)$ among those with $S_i(0)=0$, 
\begin{equation}
\label{Eq:EMsecondStepSeqRegS1_given_S0=0}
\Pr[S_i(1)=1|S_i(0)=0,\bx_{0i}] = \frac{\exp(\gamma_{0,S(1)} + \bgamma_{S(1)}^T \bx_{0i})}{1+\exp(\gamma_{0,S(1)} + \bgamma_{S(1)}^T \bx_{0i})},
\end{equation}
and $\Pr[S_i(1)=0|S_i(0)=0,\bx_{0i}] = 1 -\Pr[S_i(1)=1|S_i(0)=0,\bx_{0i}]$. Finally, recall also that form CPSR we have.
\begin{equation*}
\Pr[S_i(1)=1|S_i(0)=1,\bx_{0i}] = \frac{1}{1+\xi},
\end{equation*}
and $\Pr[S_i(1)=0|S_i(0)=1,\bx_{0i}] = 1 - \Pr[S_i(1)=1|S_i(0)=1,\bx_{0i}] = \frac{\xi}{1+\xi}$.

To describe the EM algorithm, let $\gamma^k_{0,S(a)}$ and $\bgamma^k_{S(a)}$ , $a=0,1$ denote the values of the coefficients in the $k$--th iteration of the algorithm. Denote also the conditional strata probabilities in the $k$--th iteration by 
\begin{align}
\begin{split}
\label{pis_EM_DGM_seq}
&\pi_{pro}^k(\bx_{0i}) = \left(\frac{1}
{1 + \exp(\gamma^k_{0,S(0)} + (\bgamma_{S(0)}^k)^T \bx_{0i})}\right) \cdot
\left(\frac{\exp(\gamma^k_{0,S(1)} + (\bgamma_{S(1)}^k)^T \bx_{0i})}{1+\exp(\gamma^k_{0,S(1)} + (\bgamma_{S(1)}^k)^T \bx_{0i})}\right),
\\[0.7em]
&\pi_{ns}^k(\bx_{0i}) = \left(\frac{1} 
{1 + \exp(\gamma^k_{0,S(0)} + \bgamma_{S(0)}^k)^T \bx_{0i})}\right) \cdot
\left(\frac{1}{1+\exp(\gamma^k_{0,S(1)} + \bgamma_{S(1)}^k)^T \bx_{0i})}\right), 
\\[0.7em]
&\pi_{as}^k(\bx_{0i},\xi) = \frac{\exp(\gamma^k_{0,S(0)} + \bgamma_{S(0)}^k)^T \bx_{0i})}
{1 + \exp(\gamma^k_{0,S(0)} + \bgamma_{S(0)}^k)^T \bx_{0i})} \cdot \frac{1}{1+\xi},
\\[0.7em]
&\pi_{har}^k(\bx_{0i},\xi) = \frac{\exp(\gamma^k_{0,S(0)} + \bgamma_{S(0)}^k)^T \bx_{0i})}
{1 + \exp(\gamma^k_{0,S(0)} + \bgamma_{S(0)}^k)^T \bx_{0i})} \cdot \frac{\xi}{1+\xi}.
\end{split}
\end{align}
The algorithm starts with choosing initial values for $(\gamma_{0,S(0)}, \bgamma_{S(0)}, \gamma_{0,S(1)},\bgamma_{S(1)})$,
denoted by 
$(\gamma^0_{0,S(0)}, \bgamma^0_{S(0)}, \gamma^0_{0,S(1)},\bgamma^0_{S(1)})$, and setting $k=0$.
Then, at each algorithm iteration $k$, the estimated parameters $(\gamma^k_{0,S(0)}, \bgamma^k_{S(0)}, \gamma^k_{0,S(1)},\bgamma^k_{S(1)})$ are updated according to the following E-step and M-step until convergence. Loosely, at the E-step, for model \eqref{Eq:EMFirstStepSeqRegS0}, the weights $\widetilde{w}^k_i$ are calculated for units with $A_i=1$ as for these units $S_i(0)$ is unknown. These units are then duplicated, and appear with two copies.  In the first copy $S_i(0)=0$, and the calculated weight is the probability that $S_i(0)=0$ given the covariates, the treatment and survival status, and the current values of the regression coefficients. In the second copy, $S(0)=1$ and the calculated weight is the probability that $S(0)=1$ given the observed data and current values of the regression coefficient. For units with $A=0$, $S(0)$ is known, and thus these are not duplicated and receive weight of one. The second part of the E-step for model \eqref{Eq:EMsecondStepSeqRegS1_given_S0=0} also calculate weights, denoted by $\widetilde{\widetilde{w}}^k_i$, and duplicate units as needed. The details are given below. In the M-step the two models are fitted for the augmented datasets that include the weighted and duplicated units. Details of the entire algorithm are given below.

 \begin{enumerate}
 \item \textbf{E-step:}
 \begin{enumerate}
 \item For $i=1,...,n$, calculate $\Pr(S_i(0)=1|\bQ_i)$, using the current values of the parameters, $(\gamma^k_{0,S(0)}, \bgamma^k_{S(0)}, \gamma^k_{0,S(1)},\bgamma^k_{S(1)})$.
 The obtained probabilities are 
\begin{equation}
\label{Eq:PrS0=1GivenQ}
\Pr[S_i(0)=1|\bQ_i] = \left\{\begin{array}{cr}
         0, & A_i=0,S_i=0 \\[0.7em]
   \frac{\pi_{har}^k(\bx_{0i},\xi)}{\pi_{ns}^k(\bx_{0i}) + \pi_{har}^k(\bx_{0i},\xi)}, & A_i=1,S_i=0 \\[0.7em]
1, &  A_i=0,S_i=1 \\[0.7em]
       \frac{\pi_{as}^k(\bx_{0i},\xi)}{\pi_{as}^k(\bx_{0i},\xi) + \pi_{pro}^k(\bx_{0i})}, & A_i=1, S_i=1
\end{array}\right. \\,
%\Pr(S_i(1)=1|S_i(0)=0, \bQ_i) &= \left\{\begin{array}{cr}
 %\Pr(S_i(0)=0|\bQ_i) &= 1 - \Pr(S_i(0)=1|\bQ_i).
\end{equation}
where $\pi_{g}^k(\bx_{0i})$ are defined in  \eqref{pis_EM_DGM_seq}.
\item Create Dataset ``0'', an augmented dataset with duplicated and weighted units, as follows. 
\begin{itemize}
\item For units with $A_i=0$, $\widetilde{w}^k_i=1$ and $S_i(0)$ is the observed survival status. 
\item For units with $A=1$, $S_i(0)$ is unknown. Hence, create two units, one with $S_i(0)=0$  and  $\widetilde{w}^k_i=\Pr[S_i(0)=0|\bQ_i]$ and one  with $S_i(0)=1$ and  $\widetilde{w}^k_i=\Pr[S_i(0)=1|\bQ_i]$.
\end{itemize}
\item Create Dataset ``1'', an augmented dataset with duplicated and weighted units, as follows.
\begin{itemize}
\item Omit units with $\{A=0, S=1\}$. These units are known to have $S_i(0)=1$ and hence should not be included when fitting model \eqref{Eq:EMsecondStepSeqRegS1_given_S0=0}.
\item Units with  $\{A=0, S=0\}$ are part of the desired subset, but their outcome $S(1)$ is unknown. These units are duplicated. One copy with $S_i(1)=1$ and a weight  $\widetilde{\widetilde{w}}^k_i=\Pr[S_i(1)=1|S_i(0)=0, \bx_{0i})$. The second copy with $S_i(1)=0$ and a weight $\widetilde{\widetilde{w}}^k_i=\Pr[S_i(1)=0|S_i(0)=0, \bx_{0i}]$.  The conditional probability (\ref{Eq:EMsecondStepSeqRegS1_given_S0=0}), is given by 
$$
\Pr[S_i(1)=1|S_i(0)=0, \bx_{0i}] = \frac{\pi_{pro}^k(\bx_{0i})}{\pi_{ns}^k(\bx_{0i}) + \pi_{pro}^k(\bx_{0i})} =\frac{\exp(\gamma^k_{0,S(1)} + (\bgamma_{S(1)}^k)^T \bx_{0i})}{1+\exp(\gamma^k_{0,S(1)} + (\bgamma_{S(1)}^k)^T \bx_{0i})}.
$$ 
\item For units with  $A_i=1$, their outcome $S_i(1)$ is known, but it is unknown whether they are part of the desired subset. Therefore, a single copy of each such unit is included in Dataset ``1'', with with weights $\widetilde{\widetilde{w}}^k_i=\Pr(S_i(0)=0|\bQ_i)$, defined in \eqref{Eq:PrS0=1GivenQ}.
\end{itemize}

 \end{enumerate}
\item \textbf{M-step:}

\begin{enumerate}[(I)]
\item \textbf{Logistic regression model for $S(0)$:}
\begin{enumerate}[(i)]
\item Fit the logistic regression model  given in Equation \eqref{Eq:EMFirstStepSeqRegS0} for Dataset ``0'', with the weights $\widetilde{w}^k_i$.
\item Update the parameters to be the estimated $(\gamma^{k+1}_{0,S(0)}, \bgamma^{k+1}_{S(0)})$ from the logistic regression model.
\end{enumerate}

\item \textbf{Logistic regression model for $S(1)$, given $S(0)=0$:}
\begin{enumerate}[(i)]
\item Fit the logistic regression model given in Equation \eqref{Eq:EMsecondStepSeqRegS1_given_S0=0} for Dataset ``1'' , with the weights $\widetilde{\widetilde{w}}^k_i$.
\item Update the parameters to be the estimated $(\gamma^{k+1}_{0,S(1)},\bgamma^{k+1}_{S(1)})$ 
from the logistic regression model.
\end{enumerate}

\end{enumerate}

\end{enumerate}

Note that fitting the model used randomization and SUTVA, namely that
$\Pr(S_i(a)=s|\bX_{0i}=\bx_{0i}) = 
 \Pr(S_i=s|A_i=a, \bX_{0i}=\bx_{0i})$ for a=0,1, s=0,1. Therefore, one can estimate $\gamma_{0,S(0)}$ and  $\bgamma_{S(0)}$
by fitting a logistic regression model among the untreated, and estimate only $\gamma_{0,S(1)}$ and $\bgamma_{S(1)}$ through the EM algorithm. By adopting this approach (as we did), the creation of dataset ``0'', and the estimation of the logistic regression model given in Equation \eqref{Eq:EMFirstStepSeqRegS0}, during the E and M steps, respectively, become redundant.

%%%%%%%%%%%%%%%%%%%%%%%%%%%%%%%%%%%%%%%%%%%%%%%%%%%%%%%%%%%%%%%%%

\subsection{Multinomial logistic regression model for the principal strata and its estimation}
\label{App:DGMmulti_DGM_and_EM}
We first describe the DGM under a multinomial logistic regression for the principal strata proportions before describing the EM algorithms for fitting this model.
\subsubsection{DGM under the multinomial logistic regression model for the principal strata}
\label{SubSec:AppDGMmulti}

Let $V_i = G_i$ if $G_i \in \{ns, pro\}$, and $V_i = ah$ if $G_i \in \{as, har\}$.  Let $\pi_v(\bx_{0i})=\Pr(V_i=v| \bX_{0i}=\bx_{0i})$, $v \in \{ns, pro, ah\}$. 
Consider the following multinomial regression model for $V_i$, with the protected as the reference category.
\begin{equation}
\label{Eq:MultiRegV}
\pi_v(\bx_{0i}) = \frac{\exp(\gamma_{0,v} + \bgamma_v^T \bx_{0i})} 
{1 + \exp(\gamma_{0,ns} + \bgamma_{ns}^T \bx_{0i}) + \exp(\gamma_{0,ah} + \bgamma_{ah}^T \bx_{0i})},
\end{equation}
for $v=ns,ah$, and $\pi_{pro}(\bx_{0i}) = 1 - \pi_{ns}(\bx_{0i}) - \pi_{ah}(\bx_{0i})$.
Among $V_i = ah$, under CPSR   we have
\begin{align}
\begin{split}
\label{Eq:MultiRegV_ah_by_xi}
\pi_{as}(\bx_{0i}) &= \pi_{ah}(\bx_{0i})\frac{1}{1 + \xi},\\
\pi_{har}(\bx_{0i}) &= \pi_{ah}(\bx_{0i})\frac{\xi}{1 + \xi}.
\end{split}
\end{align}

\subsubsection{EM algorithm for the multinomial logistic regression model under monotonicity}
\label{sec:EMmono}
Under monotonicity, the $ah$ group only include always-survivors and we replace the $ah$ notation with $as$.
The algorithm starts with choosing initial values for $(\gamma_{0,ns}, \bgamma_{ns}, \gamma_{0,as},\bgamma_{as})$, denoted by $(\gamma^0_{0,ns}, \bgamma^0_{ns}, \gamma^0_{0,as},\bgamma^0_{as})$ and setting $k=0$. Then, at each algorithm iteration $k=0,1,...,$ until convergence, the  estimated parameters $(\gamma^k_{0,ns},\bgamma^k_{ns}, \gamma^k_{0,as}, \bgamma^k_{as})$ are updated according to the following E-step and M-step.
\begin{enumerate}
\item \textbf{E-step:}
\begin{enumerate}
\item For $i=1,...,n$, calculate $\Pr(G_i=g|\bQ_i)$ according to model Equation (\ref{Eq:MultiRegV}), using the current values of the parameters, $(\gamma^k_{0,ns}, \bgamma^k_{ns}, \gamma^k_{0,as},\bgamma^k_{as})$. The obtained probabilities are 
\begin{align}
\begin{split}
\Pr(G_i=as|\bQ_i) &= \left\{\begin{array}{cr}
    0, & S_i=0 \\[0.3em]
    1, &  A_i=0,S_i=1 \\[0.3em]
    \frac{1}{1 + \exp[-\gamma^k_{0,as} - (\bgamma^k_{as})^T \bx_{0i}]}, & A_i=1, S_i=1 \\[0.3em]
    \end{array}\right. \\[0.7em]
\Pr(G_i=ns|\bQ_i) &= \left\{\begin{array}{cr}
        \frac{1}{1 + \exp[-\gamma^k_{0,ns} - (\bgamma^k_{ns})^T \bx_{0i}]}, & A_i=0, S_i=0 \\[0.3em]
        1, & A_i=1,S_i=0 \\[0.3em]
        0, &  S_i=1 
        \end{array}\right. \\[0.7em]
\Pr(G_i=pro|\bQ_i) &= 1 - \Pr(G_i=as|\bQ_i) - \Pr(G_i=ns|\bQ_i).
\end{split}
\end{align}

\item Create an augmented and weighted dataset as follows. 
For units in $\{A=0,S=1\}$ and $\{A=1,S=0\}$, $w^k_i=1$.
For each unit $i$ in $\{A=1,S=1\}, \{A=0,S=0\}$, we created two units, with weights $w^k_i$, according to its possible strata and $\Pr(G_i=g|\bQ_i)$.
\end{enumerate}
\item \textbf{M-step:}
\begin{enumerate}
\item Fit the multinomial regression model (\ref{Eq:MultiRegV}) for the augmented data but with the weights $w^k_i, i=1,...,n$.
\item Update the parameters to be $(\gamma^{k+1}_{0,ns}, \bgamma^{k+1}_{ns}, \gamma^{k+1}_{0,as},\bgamma^{k+1}_{as})$ from the multinomial regression model.
\end{enumerate}
\end{enumerate}
%%%%%%%%%%%%%%%%%%%%%%%%%%%%%%%%%%%%%%%%%%%%%%%%%%%%%%%%%%%%%%%%%

%%%%%%%%%%%%%%%%%%%%%%%%%%%%%%%%%%%%%%%%%%%%%%%%%%%%%%%%%%%%%%%%%
\subsubsection{EM algorithm for the multinomial logistic regression model without monotonicity} 
\label{sec:EMnoMono}

When monotonicity is replaced with the weaker CPSR, we hanged the reference category to $ah$, as oppose to $pro$, as described in 
Section \ref{SubSec:AppDGMmulti}. Therefore, to avoid confusion, the model coefficients $\bgamma_{0,v},\bgamma_v$ are replaced with $\breve{\bgamma}_{0,v}, \breve{\bgamma}_v$.

The EM step described in Section \ref{sec:EMmono}
is revised as follows. 

The algorithm starts with choosing initial values for $(\breve{\gamma}_{0,ns}, \breve{\bgamma}_{ns}, \breve{\gamma}_{0,pro},\breve{\bgamma}_{pro})$, denoted by $(\breve{\gamma}^0_{0,ns}, \breve{\bgamma}^0_{ns}, \breve{\gamma}^0_{0,pro},\breve{\bgamma}^0_{pro})$ and setting $k=0$. Then, at each algorithm iteration $k=0,1,...,$ until convergence, the estimated parameters $(\breve{\gamma}^k_{0,ns}, \breve{\bgamma}^k_{ns}, \breve{\gamma}^k_{0,pro}, \breve{\bgamma}^k_{pro})$ are updated according to the following E-step and M-step.
\begin{enumerate}
\item \textbf{E-step:}
\begin{enumerate}
\item For $i=1,...,n$, calculate $\Pr(V_i=v|\bQ_i)$ according to model given in Equation \eqref{Eq:MultiRegV}, using the current values of the parameters, $(\breve{\gamma}^k_{0,ns}, \breve{\bgamma}^k_{ns}, \breve{\gamma}^k_{0,pro},\breve{\bgamma}^k_{pro})$. The obtained probabilities are (recall Table 1)
\begin{align}
\begin{split}
\Pr(V_i=ah|\bQ_i) &= \left\{\begin{array}{cr}
    0, & A_i=0,S_i=0, \\[0.7em]
    \frac{\xi}{\xi + (1 + \xi)\exp[\breve{\gamma}^k_{0,ns} + (\breve{\bgamma}_{ns}^k)^T \bx_{0i}]}, & A_i=1,S_i=0 \\[0.7em]
    1, &  A_i=0,S_i=1 \\[0.7em]
    \frac{1}{1 + (1 + \xi)\exp[\breve{\gamma}^k_{0,pro} + (\breve{\bgamma}_{pro}^k)^T \bx_{0i}]}, & A_i=1, S_i=1
    \end{array}\right. \\[0.7em]
\Pr(V_i=ns|\bQ_i) &= \left\{\begin{array}{cr}
\frac{\exp[\breve{\gamma}^k_{0,ns} + (\breve{\bgamma}_{ns}^k)^T \bx_{0i}]}{\exp[\breve{\gamma}^k_{0,ns} + (\breve{\bgamma}_{ns}^k)^T \bx_{0i}] + \exp[\breve{\gamma}^k_{0,pro} + (\breve{\bgamma}_{pro}^k)^T \bx_{0i}]}, & A_i=0, S_i=0
        \\[0.7em]
\frac{\exp[\breve{\gamma}^k_{0,ns} + (\breve{\bgamma}_{ns}^k)^T \bx_{0i}]}{\exp[\breve{\gamma}^k_{0,ns} + (\breve{\bgamma}_{ns}^k)^T \bx_{0i}] + \frac{\xi}{1 + \xi}}, & A_i=1,S_i=0 \\[0.7em]
0, &  S_i=1 
\end{array}\right. \\[0.7em]
\Pr(V_i=pro|\bQ_i) &= 1 - \Pr(V_i=ah|\bQ_i) - \Pr(V_i=ns|\bQ_i).
\end{split}
\end{align}

\item Create an augmented and weighted dataset as follows. 
For units in $\{A=0,S=1\}$, $w^k_i=1$.
For each unit $i$ in $\{A=1,S=1\}, \{A=0,S=0\}$ and $\{A=1,S=0\}$, we created two units, with weights $w^k_i$, according to its possible strata and $\Pr(V_i=v|\bQ_i)$.
\end{enumerate}
\item \textbf{M-step:}
\begin{enumerate}
\item Fit the multinomial regression model \eqref{Eq:MultiRegV} for the augmented data but with the weights $w^k_i, i=1,...,n$.
\item Update the parameters to be $(\breve{\gamma}^{k+1}_{0,ns}, \breve{\bgamma}^{k+1}_{ns}, \breve{\gamma}^{k+1}_{0,pro},\breve{\bgamma}^{k+1}_{pro})$ from the multinomial regression model.
\end{enumerate}
\end{enumerate}
%%%%%%%%%%%%%%%%%%%%%%%%%%%%%%%%%%%%%%%%%%%%%%%%%%%%%%%%%%%%%%%%%
%%%%%%%%%%%%%%%%%%%%%%%%%%%%%%%%%%%%%%%%%%%%%%%%%%%%%%%%%%%%%%%%%

% \section{Further details about the simulation studies}
% \label{Sec:AppSims}

% %\renewcommand\thefigure{C\@arabic\c@figure}
% %\renewcommand\thetable{C \@arabic\c@table}
% \subsection{Data generating mechanism: multinomial logistic regression model (DGM-multi)}
% \label{SubSec:AppDGMmulti}
% The described mechanism is identical to the one described in the main text (DGM-seq), except the strata probabilities generation. 
% The stratum $G_i$ was generated under monotonicity according to a multinomial regression model (DGM-multi), with the protected as the reference stratum,
% \begin{equation}
% \label{Eq:MultiReg}
% \pi_g(\bx_{0i}) = \frac{\exp(\gamma_{0,g} + \bgamma_g^T \bx_{0i})} 
% {1 + \exp(\gamma_{0,as} + \bgamma_{as}^T \bx_{0i}) + \exp(\gamma_{0,ns} + \bgamma_{ns}^T \bx_{0i})},
% \end{equation}
% for $g=as,ns$ and $\pi_{pro}(\bx_{0i}) = 1 - \pi_{as}(\bx_{0i}) - \pi_{ns}(\bx_{0i})$.

% When misspecified, the true principal score model included squared and log terms for two of the covariates (See Tables \ref{Tab:gamma parameter table Scenario A DGMmulti} and \ref{Tab:gamma parameter table Scenario B DGMmulti}). 

\section{Further details about the simulation studies}
\label{Sec:AppSims}

\subsection{Parameters list and stratum proportions}
\label{SubSec:List of parameters}

We now present the values of the parameters we used for the simulation studies described  in Section 7 of the main text and in the additional simulations presented in Section \ref{SubSec:AppSimRes}.
As in Section \ref{App: sec:EM_DGM_seq},
let $\gamma_{0,S(a)}$ and $\bgamma_{S(a)}, a=0,1$ denote the logistic regressions intercepts and coefficients, respectively, where $\bgamma_{S(a)} = (c_{S(a)},...c_{S(a)})$.
For the multinomial regression model for the principal strata, we consider the intercept and the coefficients $\gamma_{0,v}, \bgamma_{v}$, 
described in Section \ref{SubSec:AppDGMmulti},
where $\bgamma_{g} = (c_{g},...c_{g})$, for $v=ns, ah$.

\begin{itemize}
\item Tables \ref{Tab:gamma parameter table Scenario A DGMseq} and \ref{Tab:gamma parameter table Scenario B DGMseq} present the values of the parameters $\gamma_{0,S(a)}, c_{S(a)}$, for $a = 0,1$, 
under a correctly specified (left side) and misspecified (right side) principal score model. For the misspecified principal score model
the modified coefficients of the squared and the exponential  transformations or interactions are presented, while the other coefficients remain the same as in the left panels.

\item Tables \ref{Tab:gamma parameter table Scenario A DGMmulti} and \ref{Tab:gamma parameter table Scenario B DGMmulti} are the the analogue of Tables \ref{Tab:gamma parameter table Scenario A DGMseq} and \ref{Tab:gamma parameter table Scenario B DGMseq}, under the multinomial regression model for the principal strata,
where $\bgamma_{g} = (c_{g},...c_{g})$, for $g=as,ns$.  
Under the multinomial regression model, when misspecified, the true principal score model included squared and log terms for two of the covariates. 

\item Table \ref{Tab:parameter table beta} presents the values of outcome regression model parameters $\beta_{0,0}, \bbeta_{0}, \beta_{0,1}$ and $\bbeta_{1}$.

\item Tables \ref{Tab:AppSimsPiDGMseq} and \ref{Tab:AppSimsPiDGMmulti} give the strata proportions for each scenario, under the sequential logistic regression model and under the multinomial logistic regression model, respectively.
\end{itemize}

\begin{table}
\caption{\label{Tab:gamma parameter table Scenario A DGMseq}True $\gamma_{0,S(0)}, c_{S(0)}, \gamma_{0,S(1)}, c_{S(1)}$ values in the simulations, Scenario A.} 
\centering
\fbox{
\scriptsize	  
\begin{tabular}{lllllllll}
\multicolumn{1}{l}{} &
\multicolumn{4}{l}{\textbf{PS correctly specified}} & 
%\multicolumn{1}{l}{} & 
\multicolumn{4}{l}{\textbf{PS correctly specified}} \\[0.25cm]
$\pi_{pro}$& coefficient & $k=3$ & $k=5$ & $k=10$ 
& transformation & $k=3$ & $k=5$ & $k=10$ \\[0.25cm]
\hline
\multirow{4}*{High} 
%&  \multirow{2}*{as} 
& $\gamma_{0,S(0)}$ & -0.69 & -0.07 & -0.3 & square & -0.92 & -0.06 & -0.12 \\[0.25cm]
& $c_{S(0)}$ & 0.46 & 0.03 & 0.06 & int/exp & 0.92 & 0.06 & 0.12 \\[0.25cm] 
%& \multirow{2}*{ns} 
& $\gamma_{0,S(1)}$ & 0.5 & 0.41 & 0.17 & square & -1.12 & -0.4 & -0.3 \\[0.25cm]
& $c_{S(1)}$ & 0.56 & 0.2 & 0.15 & int/exp & 1.12 & 0.4 & 0.3 \\[0.25cm]
\hline
\multirow{4}*{Low} 
%&  \multirow{2}*{as} 
& $\gamma_{0,S(0)}$ & -0.1 & 0.00 & 0.05 & square & -0.14 & 0 & 0.02 \\[0.25cm] 
& $c_{S(0)}$ & 0.07 & 0.00 & -0.01 & int/exp & 0.14 & 0 & -0.02 \\[0.25cm]
%& \multirow{2}*{ns} 
&  $\gamma_{0,S(1)}$ & -0.9 & -0.26 & -0.21 & square & 0.9 & 1.26 & 0.54 \\[0.25cm] 
& $c_{S(1)}$ & -0.45 & -0.63 & -0.27 & int/exp & -0.9 & -1.26 & -0.54 \\ 
%\hline
\end{tabular}}
\end{table}

\begin{table}
\caption{\label{Tab:gamma parameter table Scenario B DGMseq}True $\gamma_{0,S(0)}, c_{S(0)}, \gamma_{0,S(1)}, c_{S(1)}$ values in the simulations, Scenario B.} 
\centering
\fbox{
\scriptsize	  
\begin{tabular}{lllllllll}
\multicolumn{1}{l}{} &
\multicolumn{4}{l}{\textbf{PS correctly specified}} & 
%\multicolumn{1}{l}{} & 
\multicolumn{4}{l}{\textbf{PS correctly specified}} \\[0.25cm]
$\pi_{pro}$ & coefficient & $k=3$ & $k=5$ & $k=10$ 
& transformation & $k=3$ & $k=5$ & $k=10$ \\[0.25cm]
\hline
\multirow{4}*{High} 
%&  \multirow{2}*{as} 
& $\gamma_{0,S(0)}$ & 0.46 & 0.39 & -0.5 & square & -1.12 & -0.66 & -0.76 \\[0.25cm]
& $c_{S(0)}$ & 0.56 & 0.33 & 0.38 & int/exp & 1.12 & 0.66 & 0.76 \\[0.25cm] 
%& \multirow{2}*{ns} 
& $\gamma_{0,S(1)}$ & 1.4 & 0.85 & 0.5 & square & 0.1 & -1.7 & -1.00 \\[0.25cm]
& $c_{S(1)}$ & -0.05 & 0.85 & 0.5 & int/exp & -0.1 & 1.7 & 1.00 \\[0.25cm]
\hline
\multirow{4}*{Low} 
%&  \multirow{2}*{as} 
& $\gamma_{0,S(0)}$ & 0.51 & -0.26 & -0.51 & square & -1.02 & -1.5 & -0.76 \\[0.25cm] 
& $c_{S(0)}$ & 0.51 & 0.75 & 0.38 & int/exp & 1.02 & 1.5 & 0.76 \\[0.25cm]
%& \multirow{2}*{ns} 
&  $\gamma_{0,S(1)}$ & -0.19 & -0.17 & 0.24 & square & 0.94 & 0.84 & 0.50 \\[0.25cm] 
& $c_{S(1)}$ & -0.47 & -0.42 & -0.25 & int/exp & -0.94 & -0.84 & -0.50\\ 
%\hline
\end{tabular}}
\end{table}

\begin{table}
\caption{\label{Tab:gamma parameter table Scenario A DGMmulti} True $\gamma_{0,as}, c_{as}, \gamma_{0,ns}, c_{ns}$ values in the simulations, Scenario A.
The principal strata were generated according to the multinomial logistic regression model.} 
\centering
\fbox{
\scriptsize	  
\begin{tabular}{llllllllll}
\multicolumn{2}{l}{} &
\multicolumn{4}{l}{\textbf{PS correctly specified}} & 
%\multicolumn{1}{l}{} &
\multicolumn{4}{l}{\textbf{PS misspecification}} \\[0.25cm]
$\pi_{pro}$ & G & coefficient & $k=3$ & $k=5$ & $k=10$ 
& transformation & $k=3$ & $k=5$ & $k=10$ \\[0.25cm]
\hline
\multirow{4}*{High} 
&  \multirow{2}*{as} 
& $\gamma_{0,as}$ & -0.1 & -0.05 & -0.954 & square & -0.81 & -0.48 & -0.75 \\[0.25cm]
& & $c_{as}$ & 0.27 & 0.16 & 0.25 & log & 0.81 & 0.48 & 0.75 \\[0.25cm] 
& \multirow{2}*{ns} 
& $\gamma_{0,ns}$ & -0.52 & -0.4 & -0.31 & square & 1.8 & 0.75 & 0.48 \\[0.25cm]
& & $c_{ns}$ & -0.6 & -0.25 & -0.16 & log & -1.8 & -0.75 & -0.48 \\[0.25cm]
\hline
\multirow{4}*{Low} 
&  \multirow{2}*{as} 
& $\gamma_{0,as}$ & -0.07 & 0.45 & 0.024 & square & -3.72 & -2.25 & -1.23 \\[0.25cm] 
& & $c_{as}$ & 1.24 & 0.75 & 0.41 & log & 3.72 & 2.25 & 1.23 \\[0.25cm]
& \multirow{2}*{ns} 
&  $\gamma_{0,ns}$ & 1.25 & 0.62 & 0.26 & square & -0.72 & -1.8 & -0.96 \\[0.25cm] 
& & $c_{ns}$ & 0.24 & 0.6 & 0.32 & log & 0.7 & 1.80 & 0.96 \\ 
\end{tabular}}
\end{table}

\begin{table}
\caption{\label{Tab:gamma parameter table Scenario B DGMmulti} True $\gamma_{0,as}, c_{as}, \gamma_{0,ns}, c_{ns}$ values in the simulations, Scenario B.
The principal strata were generated according to the multinomial logistic regression model.} 
\centering
\fbox{
\scriptsize	  
\begin{tabular}{llllllllll}
\multicolumn{2}{l}{} &
\multicolumn{4}{l}{\textbf{PS correctly specified}} & 
%\multicolumn{1}{l}{} &
\multicolumn{4}{l}{\textbf{PS misspecification}} \\[0.25cm]
$\pi_{pro}$& G & coefficient & $k=3$ & $k=5$ & $k=10$ 
& transformation & $k=3$ & $k=5$ & $k=10$ \\[0.25cm]
\hline
\multirow{4}*{High} 
&  \multirow{2}*{as} 
& $\gamma_{0,as}$ & 0.6 & -0.5 & -1.01 & square & -3.975 & -4.5 & -2.4 \\[0.25cm]
& & $c_{as}$ & 1.325 & 1.5 & 0.8 & log & 3.975 & 4.5 & 2.4 \\[0.25cm] 
& \multirow{2}*{ns} 
& $\gamma_{0,ns}$ & -0.25 & -0.25 & -1.3 & square & -0.75 & -0.75 & -1.35 \\[0.25cm]
& & $c_{ns}$ & 0.25 & 0.25 & 0.45 & log & 0.75 & 0.75 & 1.35 \\[0.25cm]
\hline
\multirow{4}*{Low} 
&  \multirow{2}*{as} 
& $\gamma_{0,as}$ & 0.84 & -0.12 & -0.45 & square & -3.6 & -3.75 & -1.83 \\[0.25cm] 
& & $c_{as}$ & 1.2 & 1.25 & 0.61 & log & 3.6 & 3.75 & 1.83 \\[0.25cm]
& \multirow{2}*{ns} 
&  $\gamma_{0,ns}$ & 0.32 & 0.45 & 0.4 & square & 0.6 & 0.6 & 0.06 \\[0.25cm] 
& & $c_{ns}$ & -0.2 & -0.2 & -0.02 & log & -0.6 & -0.6 & -0.06\\ 
\end{tabular}}
\end{table}

\begin{table}
\caption{\label{Tab:parameter table beta} True $\beta_{0,0}, \bbeta_{0}, \beta_{0,1}$, $\bbeta_{1}$ values in the simulations.}
\centering
\fbox{
\scriptsize	
\begin{tabular}{l|lllllllllllll}
\hline\\[0.01em]
Outcome model & $k$ & $a$ & $\beta_{0,a}$ & $\beta_{1,a}$ & $\beta_{2,a}$ & $\beta_{3,a}$ & $\beta_{4,a}$ & $\beta_{5,a}$ & $\beta_{6,a}$ & $\beta_{7,a}$ & $\beta_{8,a}$ & $\beta_{9,a}$ & $\beta_{10,a}$ \\[0.25em]
\hline
\multirow{6}*{True model}\\ 
\multirow{6}*{does not include}\\
\multirow{6}*{$A-\bX_0$ interactions}
& \multirow{2}*{3} & 1  & 22 & 3 & 4 & 5 &  &  &  &  &  &  &  \\
& & 0 & 20 & 3 & 4 & 5 &  &  &  &  &  &  &  \\[0.1em] 
\cline{2-14} \\
& \multirow{2}*{5} & 1 & 22 & 3 & 4 & 5 & 1 & 3 &  &  &  &  &  \\[0.1em] 
& & 0 & 20 & 3 & 4 & 5 & 1 & 3 &  &  &  &  &  \\[0.1em] 
\cline{2-14}\\
& \multirow{2}*{10}  & 1 & 22 & 5 & 2 & 1 & 3 & 5 & 5 & 2 & 1 & 3 & 5 \\[0.1em]
&  & 0 & 20 & 5 & 2 & 1 & 3 & 5 & 5 & 2 & 1 & 3 & 5 \\[0.1em]
\hline
\multirow{6}*{True model includes} \\ 
\multirow{6}*{$A-\bX_0$ interactions}
& \multirow{2}*{3} & 1 & 22 & 5 & 2 & 1 &  &  &  &  &  & &  \\[0.1em]
&  & 0 & 20 & 3 & 3 & 0 &  &  &  &  &  &  &  \\[0.1em]
\cline{2-14}\\
& \multirow{2}*{5} & 1 & 22 & 5 & 2 & 1 & 3 & 5 &  &  &  &  &  \\[0.1em]
&  & 0 &  20 & 3 & 3 & 0 & 1 & 3 &  &  &  &  &  \\[0.1em]
\cline{2-14}\\
& \multirow{2}*{10} & 1 & 22 & 5 & 2 & 1 & 3 & 5 & 5 & 2 & 1 & 3 & 5 \\[0.1em]
& & 0 & 20 & 3 & 3 & 0 & 1 & 3 & 3 & 3 & 0 & 1 & 3 \\[0.1em]  
\end{tabular}}
\end{table}    

\begin{table}
\caption{\label{Tab:AppSimsPiDGMseq}Strata proportions for each scenario in the simulations, under correctly specified and misspecified principal score model.} 
\centering
\fbox{
\begin{tabular}{c|ccccccc}
 \hline
 & & \multicolumn{3}{c}{High $\pi_{pro}$} & \multicolumn{3}{c}{Low $\pi_{pro}$} \\
\cline{2-8}\\
& Scenario & $\pi_{ah}$ & $\pi_{pro}$ & $\pi_{ns}$&$\pi_{ah}$ & $\pi_{pro}$ & $\pi_{ns}$ \\
\hline\\
\multirow{2}*{Correctly specified} \\ 
\multirow{2}*{models}
& A  & 0.50 & 0.35 & 0.15 & 0.50 & 0.10 & 0.40 \\
& B & 0.75 & 0.20 & 0.05 & 0.75 & 0.10 & 0.15 \\
\hline
\multirow{2}*{Misspecified principal} \\ 
\multirow{2}*{score model ($k=3$)}
& A  & 0.55 & 0.22 & 0.23 & 0.53 & 0.14 & 0.33 \\
& B & 0.71 & 0.24 & 0.05 & 0.71 & 0.19 & 0.10 \\
\hline
\multirow{2}*{Misspecified principal} \\ 
\multirow{2}*{score model ($k=5$)}
& A  & 0.51 & 0.34 & 0.15 & 0.50 & 0.13 & 0.37 \\
& B & 0.73 & 0.13 & 0.14 & 0.70 & 0.19 & 0.11 \\
\hline
\multirow{2}*{Misspecified principal} \\ 
\multirow{2}*{score model ($k=10$)}
& A  & 0.52 & 0.32 & 0.16 & 0.49 & 0.11 & 0.40 \\
& B & 0.73 & 0.16 & 0.11 & 0.73 & 0.15 & 0.12 \\
\end{tabular}}
\end{table}

\begin{table}
\caption{\label{Tab:AppSimsPiDGMmulti} Strata proportions at each scenario in the simulations, under correctly specified and misspecified principal score model. 
The principal strata were generated according to the multinomial logistic regression model.} 
\centering
\fbox{
\begin{tabular}{c|ccccccc}
 \hline
 & & \multicolumn{3}{c}{High $\pi_{pro}$} & \multicolumn{3}{c}{Low $\pi_{pro}$} \\
\cline{2-8}\\
& Scenario & $\pi_{as}$ & $\pi_{pro}$ & $\pi_{ns}$&$\pi_{as}$ & $\pi_{pro}$ & $\pi_{ns}$ \\
\hline\\
\multirow{2}*{Correctly specified} \\ 
\multirow{2}*{principal score model}
& A  & 0.50 & 0.35 & 0.15 & 0.50 & 0.10 & 0.40 \\
& B & 0.75 & 0.14 & 0.11 & 0.75 & 0.10 & 0.15 \\
\hline
\multirow{2}*{Misspecified principal} \\ 
\multirow{2}*{score model ($k=3$)}
& A  & 0.54 & 0.25 & 0.21 & 0.60 & 0.15 & 0.25 \\
& B & 0.73 & 0.20 & 0.07 & 0.73 & 0.07 & 0.20 \\
\hline
\multirow{2}*{Misspecified principal} \\ 
\multirow{2}*{score model ($k=5$)}
& A  & 0.52 & 0.29 & 0.19 & 0.50 & 0.18 & 0.32 \\
& B & 0.74 & 0.19 & 0.07 & 0.73 & 0.07 & 0.20 \\
\hline
\multirow{2}*{Misspecified principal} \\ 
\multirow{2}*{score model ($k=10$)}
& A  & 0.53 & 0.30 & 0.17 & 0.50 & 0.15 & 0.35 \\
& B & 0.74 & 0.19 & 0.07 & 0.74 & 0.10 & 0.16 \\
\end{tabular}}
\end{table}

\clearpage
\newpage

\subsection{Additional simulation results}
\label{SubSec:AppSimRes}
This section presents additional simulation results which are summarized in Section 7.3 of the main text. We present the following results:

\begin{itemize}
     \item Table \ref{Tab:AppSimResWithInterDGMseq} extends Table 2 of the main text, when monotonicity holds ($\xi_{assm}=\xi=0$), showing results for $k=3,10$, instead of $k=5$ shown in Table 2.
     \item Table \ref{Tab:AppSimResWoutInterDGMseq} is the analogue of Table \ref{Tab:AppSimResWithInterDGMseq},
     when the true outcome model did not include $A$-$\bX_0$ interactions.
     
    \item Table \ref{Tab:AppSimResWithInterSeveralcorrectXiDGMseq} presents the results of SACE estimators, when $\xi_{assm} = \xi$, under several values of $\xi$, with $k=5$, when the true outcome model outcome included all $A$-$\bX_0$ interactions.
    \item Table \ref{Tab:AppSimResWithInterSeveralwrongXiDGMseq} presents the results of SACE estimators when $\xi$ is assumed to be $0$, but the true value of $\xi$ is varying between $0$ and $0.2$, with $k=5$, when the true outcome model outcome included all $A$-$\bX_0$ interactions.

   \item Table \ref{Tab:AppSimResWithInterFullDGMmulti} presents selected SACE estimators under a multinomial regression model for the principal strata, when the true outcome model outcome included all $A$-$\bX_0$ interactions.
    In the left panel, both models were correctly specified and in the left panel, the principal score model was misspecified. 

    \item Figure \ref{Fig:AppbiasS2withInterDGMseq} shows results of selected SACE estimators when $\xi_{assm} = \xi$, under several values of $\xi$, with $k=5$, and the true outcome model outcome included all $A$-$\bX_0$ interactions.
    \item Figure \ref{Fig:AppbiasS3withInterDGMseq} shows results of SACE estimators when $\xi_{assm} = 0$ (i.e. monotonicity is assumed) but the true value of $\xi$ is varying between $0$ and $0.2$, with $k=5$, and the true outcome model outcome included all $A$-$\bX_0$ interactions.
    \item Figure \ref{Fig:AppbiasS1woutInterDGMseq}, Figure \ref{Fig:AppbiasS2woutInterDGMseq} and Figure \ref{Fig:AppbiasS3woutInterDGMseq} are the analogue of Figure 2 of the main text, Figure \ref{Fig:AppbiasS2withInterDGMseq}, and Figure \ref{Fig:AppbiasS3withInterDGMseq}, respectively, when the true outcome model did not include $A$-$\bX_0$ interactions.
    \item Figure \ref{Fig:AppbiasS1withInterDGMmulti} is the analogue of the top left and top right panels of Figure 1 of the main text,|
    under a multinomial regression model for the principal strata. 
 
\end{itemize}

\clearpage

%%%%%%%%%%%%%%%%%%%%%%%%%%%%%%%%%%%%%%%%%%%%%%%%%%%%%%%%%%%%%%%%%%%%%%%%%%%%%%%%
\begin{table}
\caption{\label{Tab:AppSimResWithInterDGMseq} \footnotesize 
Selected simulation results when monotonicity holds ($\xi_{assm}=\xi=0$), with low $\pi_{pro}$ and $k = 3,5$ covariares.
True outcome model included $A$-$\bX_0$ interactions. 
Results described for matching on $\widehat{\widetilde{\pi}}^1_{as}(\bx_0)$ (PS) or using Mahalanobis distance with a caliper (cal).
OLS: least squares with interactions; WLS: weighted least squares with interactions; DL: model-based weighting estimator of \cite{ding2017principal}; Emp.SD: empirical standard deviation. Est.SE: estimated standard error; MSE: mean square error; CP95: empirical coverage proportion of 95\% confidence interval.}
\centering
\fbox{
\scriptsize
\begin{tabular}{lccccccccccc}
& \multicolumn{5}{c}{Correctly specified principal score and outcome models} & \multicolumn{5}{c}{Misspecified principal score and outcome models} \\
\em Method & \em Estimator &  \em Mean & \em Emp.SD & \em Est.SE & \em MSE & \em CP95 & Mean & Emp.SD & Est.SE & MSE & CP95 \\ \hline
& & & & & & & & & \\
& \multicolumn{4}{c}{\textbf{Scenario A, SACE = 3.07}} & \multicolumn{2}{c}{\textbf{k=3}} &  \multicolumn{4}{c}{\textbf{Scenario A, SACE = 0.8}} \\
\hline
Matching & Crude:cal & 3.02 & 0.15 & 0.14 & 0.02 & 0.91 & 0.40 & 0.52 & 0.36 & 0.43 & 0.76 \\ 
& Crude:PS & 3.02 & 0.21 & 0.21 & 0.05 & 0.93 & 0.64 & 0.65 & 0.61 & 0.44 & 0.94 \\ 
& OLS:cal & 3.07 & 0.13 & 0.13 & 0.02 & 0.96 & 0.61 & 0.42 & 0.33 & 0.22 & 0.87 \\ 
&  OLS:PS & 3.07 & 0.13 & 0.13 & 0.02 & 0.95 & 0.85 & 0.46 & 0.42 & 0.22 & 0.92 \\ 
%& OLS:cal & 3.06 & 0.12 & 0.12 & 0.01 & 0.96 & 0.65 & 0.39 & 0.31 & 0.18 & 0.88 \\ 
Matching &  Crude:cal & 3.03 & 0.14 & 0.13 & 0.02 & 0.92 & 0.39 & 0.34 & 0.32 & 0.28 & 0.81 \\ 
with &  Crude:PS & 3.06 & 0.25 & 0.21 & 0.06 & 0.90 & 0.69 & 0.69 & 0.61 & 0.48 & 0.92 \\ 
replacement &  WLS:cal & 3.07 & 0.14 & 0.13 & 0.02 & 0.94 & 0.43 & 0.32 & 0.36 & 0.24 & 0.92 \\ 
& WLS:PS & 3.07 & 0.14 & 0.14 & 0.02 & 0.94 & 0.86 & 0.47 & 0.49 & 0.22 & 0.96 \\ 
&  BC:cal & 3.07 & 0.14 & 0.16 & 0.02 & 0.97 & 0.43 & 0.32 & 0.37 & 0.24 & 0.94 \\
 \multicolumn{2}{c}{Composite}  & 3.90 & 0.57 & 0.57 & 1.01 & 0.69 & 4.59 & 0.81 & 0.80 & 15.04 & 0.00 \\ 
 \multicolumn{2}{c}{Naive} & 2.52 & 0.31 & 0.30 & 0.41 & 0.54 & 0.18 & 0.70 & 0.70 & 0.87 & 0.87 \\ 
  \multicolumn{2}{c}{DL} & 3.06 & 0.11 & & 0.01 & & 0.87 & 0.46 & & 0.21 & \\ 
& & & & & & & & & \\
 & \multicolumn{4}{c}{\textbf{Scenario B, SACE = 3.23}} & \multicolumn{2}{c}{\textbf{k=3}} & \multicolumn{4}{c}{\textbf{Scenario B, SACE = 0.45}} \\
 \hline
Matching & Crude:cal & 3.17 & 0.13 & 0.11 & 0.02 & 0.86 & 0.25 & 0.40 & 0.28 & 0.20 & 0.78 \\ 
&  Crude:PS & 3.16 & 0.20 & 0.16 & 0.04 & 0.88 & 0.62 & 0.50 & 0.48 & 0.27 & 0.92 \\ 
&  OLS:cal & 3.22 & 0.10 & 0.10 & 0.01 & 0.96 & 0.38 & 0.33 & 0.25 & 0.11 & 0.88 \\ 
&  OLS:PS & 3.22 & 0.10 & 0.10 & 0.01 & 0.95 & 0.56 & 0.34 & 0.32 & 0.13 & 0.89 \\ 
%& OLS:cal & 3.21 & 0.09 & 0.10 & 0.01 & 0.97 & 0.41 & 0.30 & 0.24 & 0.10 & 0.88 \\ 
Matching & Crude:cal & 3.18 & 0.11 & 0.10 & 0.01 & 0.90 & 0.16 & 0.25 & 0.24 & 0.14 & 0.82 \\ 
with &  Crude:PS & 3.24 & 0.20 & 0.16 & 0.04 & 0.90 & 0.64 & 0.53 & 0.47 & 0.31 & 0.90 \\ 
replacement &  WLS:cal & 3.23 & 0.11 & 0.11 & 0.01 & 0.96 & 0.19 & 0.23 & 0.28 & 0.12 & 0.94 \\ 
&  WLS:PS & 3.22 & 0.11 & 0.11 & 0.01 & 0.94 & 0.55 & 0.34 & 0.36 & 0.13 & 0.94 \\ 
&  BC:cal & 3.23 & 0.11 & 0.13 & 0.01 & 0.98 & 0.19 & 0.23 & 0.29 & 0.12 & 0.95 \\ 
\multicolumn{2}{c}{Composite}  & 4.65 & 0.48 & 0.48 & 2.25 & 0.16 & 5.99 & 0.71 & 0.69 & 31.16 & 0.00 \\ 
\multicolumn{2}{c}{Naive} & 2.60 & 0.24 & 0.25 & 0.46 & 0.27 & 0.04 & 0.54 & 0.55 & 0.46 & 0.90 \\ 
\multicolumn{2}{c}{DL} & 3.22 & 0.09 & & 0.01 &  & 0.64 & 0.33 &  & 0.14 &  \\  
& & & & & & & & & \\
& \multicolumn{4}{c}{\textbf{Scenario A, SACE = 7.95}} & \multicolumn{2}{c}{\textbf{k=10}} & \multicolumn{4}{c}{\textbf{Scenario A, SACE = 13.7}} \\
\hline
Matching & Crude:cal & 7.84 & 0.40 & 0.36 & 0.17 & 0.91 & 13.06 & 0.68 & 0.58 & 0.90 & 0.74 \\ 
&  Crude:PS & 7.83 & 0.46 & 0.50 & 0.23 & 0.96 & 13.56 & 0.80 & 0.79 & 0.66 & 0.95 \\ 
&  OLS:cal & 7.92 & 0.24 & 0.25 & 0.06 & 0.96 & 13.41 & 0.56 & 0.48 & 0.41 & 0.85 \\ 
&  OLS:PS & 7.93 & 0.24 & 0.25 & 0.06 & 0.96 & 13.98 & 0.59 & 0.56 & 0.41 & 0.93 \\ 
%& OLS:cal & 7.90 & 0.21 & 0.22 & 0.05 & 0.96 & 13.32 & 0.53 & 0.46 & 0.44 & 0.80 \\ 
Matching &  Crude:cal & 7.84 & 0.36 & 0.31 & 0.14 & 0.90 & 11.80 & 0.59 & 0.47 & 4.05 & 0.06 \\ 
with &  Crude:PS & 7.92 & 0.67 & 0.50 & 0.45 & 0.87 & 13.67 & 1.08 & 0.79 & 1.18 & 0.85 \\ 
replacement & WLS:cal & 7.93 & 0.24 & 0.25 & 0.06 & 0.95 & 11.96 & 0.46 & 0.50 & 3.31 & 0.06 \\ 
& WLS:PS & 7.93 & 0.24 & 0.25 & 0.06 & 0.96 & 13.98 & 0.74 & 0.71 & 0.61 & 0.94 \\ 
&  BC:cal & 7.93 & 0.24 & 0.31 & 0.06 & 0.99 & 11.96 & 0.46 & 0.51 & 3.31 & 0.06 \\ 
\multicolumn{2}{c}{Composite} & 7.25 & 0.81 & 0.81 & 1.15 & 0.86 & 11.73 & 0.99 & 1.02 & 4.96 & 0.50 \\ 
\multicolumn{2}{c}{Naive} & 6.98 & 0.57 & 0.58 & 1.27 & 0.61 & 13.30 & 0.77 & 0.79 & 0.77 & 0.92 \\ 
\multicolumn{2}{c}{DL} & 7.90 & 0.21 &  & 0.05 &  & 13.95 & 0.54 &  & 0.34 &  \\ 
& & & & & & & & & \\
 & \multicolumn{4}{c}{\textbf{Scenario B, SACE = 8.93}} & \multicolumn{2}{c}{\textbf{k=10}} & \multicolumn{4}{c}{\textbf{Scenario B, SACE = 14.71}} \\
 \hline
Matching & Crude:cal & 8.85 & 0.33 & 0.26 & 0.12 & 0.87 & 14.84 & 0.66 & 0.49 & 0.46 & 0.86 \\ 
& Crude:PS & 8.75 & 0.44 & 0.32 & 0.22 & 0.85 & 15.43 & 0.70 & 0.62 & 1.00 & 0.76 \\ 
&  OLS:cal & 8.90 & 0.18 & 0.20 & 0.03 & 0.96 & 14.93 & 0.51 & 0.43 & 0.31 & 0.89 \\ 
&  OLS:PS & 8.92 & 0.19 & 0.20 & 0.04 & 0.96 & 15.46 & 0.50 & 0.48 & 0.82 & 0.68 \\ 
%& OLS:cal & 8.89 & 0.16 & 0.18 & 0.03 & 0.97 & 14.91 & 0.50 & 0.40 & 0.29 & 0.88 \\ 
Matching &  Crude:cal & 8.64 & 0.27 & 0.23 & 0.15 & 0.72 & 12.95 & 0.45 & 0.37 & 3.27 & 0.01 \\
with &  Crude:PS & 8.93 & 0.40 & 0.31 & 0.16 & 0.88 & 15.53 & 0.85 & 0.61 & 1.42 & 0.66 \\ 
replacement &  WLS:cal & 8.92 & 0.19 & 0.20 & 0.04 & 0.96 & 13.17 & 0.36 & 0.40 & 2.50 & 0.04 \\ 
&  WLS:PS & 8.92 & 0.19 & 0.20 & 0.04 & 0.96 & 15.44 & 0.62 & 0.59 & 0.92 & 0.79 \\ 
&  BC:cal & 8.92 & 0.19 & 0.25 & 0.04 & 0.99 & 13.17 & 0.36 & 0.41 & 2.50 & 0.04 \\ 
\multicolumn{2}{c}{Composite} & 9.48 & 0.72 & 0.73 & 0.83 & 0.88 & 17.59 & 0.89 & 0.90 & 9.10 & 0.10 \\ 
\multicolumn{2}{c}{Naive} & 7.37 & 0.46 & 0.47 & 2.63 & 0.09 & 14.30 & 0.67 & 0.66 & 0.61 & 0.89 \\ 
\multicolumn{2}{c}{DL} & 8.91 & 0.17 &  & 0.03 &  & 15.54 & 0.47 &  & 0.92 &  \\
\end{tabular}}
\end{table}

\begin{table}
\caption{\label{Tab:AppSimResWoutInterDGMseq} \footnotesize Selected simulation results when monotonicity holds  ($\xi_{assm}=\xi=0$), with low $\pi_{pro}$ and $k=3, 5, 10$ covariates.
True outcome model did not include $A$-$\bX_0$ interactions. 
Results described for matching on $\widehat{\widetilde{\pi}}^1_{as}(\bx_0)$ (PS) or using Mahalanobis distance with a caliper (cal).
OLS: least squares with interactions; WLS: weighted least squares with interactions; DL: model-based weighting estimator of \cite{ding2017principal}; Emp.SD: empirical standard deviation. Est.SE: estimated standard error; MSE: mean square error; CP95: empirical coverage proportion of 95\% confidence interval.}
\centering
\fbox{
\tiny 
%\scriptsize
\begin{tabular}{lccccccccccc}
& \multicolumn{5}{c}{Correctly specified principal score and outcome models} & \multicolumn{5}{c}{Misspecified principal score and outcome models} \\
\em Method & \em Estimator &  \em Mean & \em Emp.SD & \em Est.SE & \em MSE & \em CP95 & Mean & Emp.SD & Est.SE & MSE & CP95 \\ \hline
& & & & & & & & & \\
& \multicolumn{4}{c}{\textbf{Scenario A, SACE = 2}} & \multicolumn{2}{c}{\textbf{k=3}} &  \multicolumn{4}{c}{\textbf{Scenario A, SACE = 2}} \\
\hline
Matching & Crude:cal & 1.92 & 0.14 & 0.10 & 0.03 & 0.79 & 1.24 & 0.79 & 0.47 & 1.19 & 0.55 \\ 
& Crude:PS & 1.92 & 0.24 & 0.20 & 0.06 & 0.88 & 2.22 & 1.04 & 1.00 & 1.12 & 0.93 \\ 
&  OLS:cal & 2.00 & 0.06 & 0.06 & 0.00 & 0.95 & 1.64 & 0.69 & 0.40 & 0.60 & 0.62 \\ 
&  OLS:PS & 2.00 & 0.06 & 0.06 & 0.00 & 0.95 & 2.62 & 0.74 & 0.66 & 0.94 & 0.81 \\ 
%& OLS:cal & 2.00 & 0.06 & 0.06 & 0.00 & 0.95 & 1.66 & 0.67 & 0.38 & 0.57 & 0.61 \\ 
Matching &  Crude:cal & 1.94 & 0.10 & 0.09 & 0.01 & 0.86 & 1.04 & 0.39 & 0.34 & 1.07 & 0.16 \\ 
with & Crude:PS & 2.00 & 0.25 & 0.20 & 0.06 & 0.90 & 2.31 & 1.19 & 0.99 & 1.51 & 0.89 \\ 
replacement & WLS:cal & 2.00 & 0.08 & 0.08 & 0.01 & 0.94 & 1.11 & 0.33 & 0.49 & 0.90 & 0.60 \\ 
&  WLS:PS & 2.00 & 0.08 & 0.08 & 0.01 & 0.96 & 2.64 & 0.83 & 0.86 & 1.09 & 0.92 \\ 
&  BC:cal & 2.00 & 0.08 & 0.08 & 0.01 & 0.94 & 1.11 & 0.33 & 0.35 & 0.90 & 0.20 \\ 
\multicolumn{2}{c}{Composite} & 3.41 & 0.65 & 0.64 & 2.42 & 0.41 & 6.96 & 1.13 & 1.11 & 25.89 & 0.01 \\ 
\multicolumn{2}{c}{Naive} & 1.16 & 0.44 & 0.44 & 0.90 & 0.51 & 1.88 & 1.13 & 1.11 & 1.28 & 0.95 \\ 
\multicolumn{2}{c}{DL} & 1.99 & 0.06 &  & 0.00 &  & 2.73 & 0.81 &  & 1.18 &  \\ 
& & & & & & & & & \\
& \multicolumn{4}{c}{\textbf{Scenario B, SACE = 2}} & \multicolumn{2}{c}{\textbf{k=3}} &  \multicolumn{4}{c}{\textbf{Scenario B, SACE = 2}} \\
\hline
Matching & Crude:cal & 1.91 & 0.14 & 0.08 & 0.03 & 0.65 & 1.69 & 0.66 & 0.36 & 0.53 & 0.65 \\ 
& Crude:PS & 1.89 & 0.24 & 0.12 & 0.07 & 0.77 & 2.98 & 0.82 & 0.78 & 1.65 & 0.74 \\ 
&  OLS:cal & 2.00 & 0.05 & 0.05 & 0.00 & 0.95 & 1.95 & 0.59 & 0.30 & 0.35 & 0.65 \\ 
&  OLS:PS & 2.00 & 0.05 & 0.05 & 0.00 & 0.95 & 2.91 & 0.55 & 0.50 & 1.14 & 0.52 \\ 
%& OLS:cal & 2.00 & 0.05 & 0.05 & 0.00 & 0.95 & 1.96 & 0.58 & 0.28 & 0.34 & 0.64 \\ 
Matching & Crude:cal & 1.94 & 0.08 & 0.06 & 0.01 & 0.79 & 1.31 & 0.27 & 0.25 & 0.54 & 0.19 \\ 
with & Crude:PS & 1.99 & 0.14 & 0.11 & 0.02 & 0.88 & 3.00 & 0.90 & 0.77 & 1.81 & 0.72 \\ 
replacement & WLS:cal & 2.00 & 0.06 & 0.06 & 0.00 & 0.95 & 1.38 & 0.23 & 0.37 & 0.43 & 0.69 \\ 
&  WLS:PS & 2.00 & 0.06 & 0.06 & 0.00 & 0.94 & 2.86 & 0.59 & 0.63 & 1.10 & 0.75 \\ 
&  BC:cal & 2.00 & 0.06 & 0.06 & 0.00 & 0.95 & 1.38 & 0.23 & 0.26 & 0.43 & 0.26 \\ 
\multicolumn{2}{c}{Composite} & 3.71 & 0.57 & 0.57 & 3.24 & 0.14 & 9.54 & 1.00 & 0.98 & 57.88 & 0.00 \\ 
\multicolumn{2}{c}{Naive} & 1.05 & 0.34 & 0.35 & 1.02 & 0.22 & 2.44 & 0.87 & 0.89 & 0.96 & 0.92 \\ 
\multicolumn{2}{c}{DL} & 2.00 & 0.05 & & 0.00 &  & 3.10 & 0.56 &  & 1.53 &  \\
  & & & & & & & & & \\
& \multicolumn{4}{c}{\textbf{Scenario A, SACE = 2}} & \multicolumn{2}{c}{\textbf{k=5}} &  \multicolumn{4}{c}{\textbf{Scenario A, SACE = 2}} \\
\hline
Matching & Crude:cal & 1.94 & 0.21 & 0.15 & 0.05 & 0.82 & 1.63 & 0.38 & 0.28 & 0.28 & 0.65 \\ 
&  Crude:PS & 1.90 & 0.30 & 0.26 & 0.10 & 0.88 & 1.91 & 0.52 & 0.57 & 0.28 & 0.97 \\ 
&  OLS:cal & 2.00 & 0.06 & 0.06 & 0.00 & 0.94 & 1.89 & 0.27 & 0.19 & 0.09 & 0.75 \\ 
&  OLS:PS & 2.00 & 0.06 & 0.06 & 0.00 & 0.96 & 2.27 & 0.29 & 0.28 & 0.15 & 0.84 \\ 
%& OLS:cal & 2.00 & 0.06 & 0.06 & 0.00 & 0.94 & 1.90 & 0.28 & 0.18 & 0.09 & 0.75 \\ 
Matching & Crude:cal & 1.89 & 0.14 & 0.12 & 0.04 & 0.80 & 1.34 & 0.24 & 0.22 & 0.49 & 0.17 \\ 
with &  Crude:PS & 2.02 & 0.30 & 0.25 & 0.09 & 0.90 & 1.90 & 0.68 & 0.57 & 0.47 & 0.90 \\ 
replacement & WLS:cal & 2.00 & 0.08 & 0.08 & 0.01 & 0.95 & 1.40 & 0.17 & 0.22 & 0.40 & 0.12 \\ 
& WLS:PS & 2.00 & 0.08 & 0.08 & 0.01 & 0.94 & 2.26 & 0.36 & 0.36 & 0.19 & 0.91 \\ 
&  BC:cal & 2.00 & 0.08 & 0.08 & 0.01 & 0.95 & 1.40 & 0.17 & 0.18 & 0.40 & 0.06 \\ 
\multicolumn{2}{c}{Composite} & 3.43 & 0.69 & 0.68 & 2.50 & 0.44 & 5.59 & 0.82 & 0.81 & 13.56 & 0.02 \\ 
\multicolumn{2}{c}{Naive} & 0.94 & 0.48 & 0.48 & 1.38 & 0.38 & 2.18 & 0.57 & 0.58 & 0.36 & 0.94 \\ 
\multicolumn{2}{c}{DL} & 2.00 & 0.06 &  & 0.00 &  & 2.36 & 0.31 &  & 0.22 &  \\ 
  & & & & & & & & & \\
& \multicolumn{4}{c}{\textbf{Scenario B, SACE = 2}} & \multicolumn{2}{c}{\textbf{k=5}} &  \multicolumn{4}{c}{\textbf{Scenario B, SACE = 2}} \\
\hline
Matching & Crude:cal & 1.92 & 0.22 & 0.12 & 0.06 & 0.68 & 2.04 & 0.39 & 0.21 & 0.16 & 0.70 \\ 
&  Crude:PS & 1.84 & 0.33 & 0.18 & 0.13 & 0.76 & 2.68 & 0.47 & 0.45 & 0.69 & 0.67 \\ 
&  OLS:cal & 2.00 & 0.05 & 0.05 & 0.00 & 0.96 & 2.14 & 0.25 & 0.14 & 0.08 & 0.66 \\ 
&  OLS:PS & 2.00 & 0.05 & 0.05 & 0.00 & 0.96 & 2.63 & 0.23 & 0.21 & 0.45 & 0.18 \\ 
%& OLS:cal & 2.00 & 0.05 & 0.05 & 0.00 & 0.96 & 2.14 & 0.25 & 0.14 & 0.08 & 0.65 \\ 
Matching & Crude:cal & 1.84 & 0.11 & 0.10 & 0.04 & 0.64 & 1.54 & 0.18 & 0.16 & 0.24 & 0.22 \\ 
with & Crude:PS & 1.99 & 0.22 & 0.18 & 0.05 & 0.89 & 2.70 & 0.59 & 0.45 & 0.84 & 0.61 \\ 
replacement & WLS:cal & 2.00 & 0.07 & 0.06 & 0.00 & 0.94 & 1.59 & 0.12 & 0.15 & 0.18 & 0.18 \\ 
&  WLS:PS & 2.00 & 0.06 & 0.06 & 0.00 & 0.95 & 2.62 & 0.28 & 0.27 & 0.46 & 0.39 \\ 
&  BC:cal & 2.00 & 0.07 & 0.06 & 0.00 & 0.94 & 1.59 & 0.12 & 0.13 & 0.18 & 0.09 \\ 
\multicolumn{2}{c}{Composite} & 3.71 & 0.61 & 0.61 & 3.31 & 0.20 & 8.25 & 0.72 & 0.69 & 39.54 & 0.00 \\ 
\multicolumn{2}{c}{Naive} & 0.70 & 0.36 & 0.37 & 1.81 & 0.06 & 2.47 & 0.49 & 0.47 & 0.46 & 0.82 \\ 
\multicolumn{2}{c}{DL} & 2.00 & 0.05 & & 0.00 &  & 2.68 & 0.22 & & 0.51 &  \\
  & & & & & & & & & \\
& \multicolumn{4}{c}{\textbf{Scenario A, SACE = 2}} & \multicolumn{2}{c}{\textbf{k=10}} &  \multicolumn{4}{c}{\textbf{Scenario A, SACE = 2}} \\
\hline
Matching & Crude:cal & 1.92 & 0.42 & 0.36 & 0.19 & 0.90 & 1.37 & 0.79 & 0.66 & 1.02 & 0.79 \\ 
& Crude:PS & 1.90 & 0.51 & 0.57 & 0.27 & 0.96 & 1.86 & 0.93 & 0.96 & 0.88 & 0.95 \\ 
&  OLS:cal & 2.00 & 0.06 & 0.06 & 0.00 & 0.94 & 1.72 & 0.56 & 0.46 & 0.39 & 0.86 \\ 
&  OLS:PS & 2.00 & 0.06 & 0.06 & 0.00 & 0.95 & 2.28 & 0.60 & 0.57 & 0.45 & 0.92 \\ 
%& OLS:cal & 2.00 & 0.06 & 0.06 & 0.00 & 0.94 & 1.75 & 0.55 & 0.44 & 0.37 & 0.85 \\ 
Matching & Crude:cal & 1.91 & 0.35 & 0.31 & 0.13 & 0.91 & 0.10 & 0.67 & 0.57 & 4.06 & 0.11 \\ 
with & Crude:PS & 1.99 & 0.68 & 0.56 & 0.47 & 0.89 & 1.97 & 1.17 & 0.95 & 1.36 & 0.89 \\ 
replacement & WLS:cal & 2.00 & 0.08 & 0.08 & 0.01 & 0.95 & 0.26 & 0.46 & 0.49 & 3.23 & 0.05 \\ 
&  WLS:PS & 2.00 & 0.08 & 0.08 & 0.01 & 0.95 & 2.28 & 0.72 & 0.72 & 0.60 & 0.94 \\ 
&  BC:cal & 2.00 & 0.08 & 0.08 & 0.01 & 0.94 & 0.26 & 0.46 & 0.49 & 3.23 & 0.04 \\ 
\multicolumn{2}{c}{Composite} & 4.29 & 0.90 & 0.89 & 6.02 & 0.27 & 5.94 & 1.14 & 1.17 & 16.85 & 0.08 \\ 
\multicolumn{2}{c}{Naive} & 1.05 & 0.69 & 0.69 & 1.38 & 0.72 & 1.60 & 0.98 & 0.99 & 1.12 & 0.94 \\ 
\multicolumn{2}{c}{DL} & 1.99 & 0.06 &  & 0.00 &  & 2.36 & 0.58 &  & 0.47 & 0 \\ &
& & & & & & & & & \\\multicolumn{4}{c}{\textbf{Scenario B, SACE = 2}} & \multicolumn{2}{c}{\textbf{k=10}} &  \multicolumn{4}{c}{\textbf{Scenario B, SACE = 2}} \\
\hline
Matching & Crude:cal & 1.95 & 0.36 & 0.24 & 0.13 & 0.80 & 2.14 & 0.77 & 0.54 & 0.61 & 0.81 \\ 
&  Crude:PS & 1.83 & 0.47 & 0.33 & 0.25 & 0.83 & 2.70 & 0.81 & 0.73 & 1.16 & 0.79 \\ 
& OLS:cal & 2.00 & 0.05 & 0.05 & 0.00 & 0.95 & 2.23 & 0.51 & 0.38 & 0.31 & 0.82 \\ 
&  OLS:PS & 2.00 & 0.05 & 0.05 & 0.00 & 0.95 & 2.74 & 0.51 & 0.47 & 0.80 & 0.64 \\ 
%& OLS:cal & 2.00 & 0.05 & 0.05 & 0.00 & 0.95 & 2.24 & 0.51 & 0.38 & 0.31 & 0.81 \\ 
Matching & Crude:cal & 1.72 & 0.25 & 0.21 & 0.14 & 0.69 & 0.23 & 0.51 & 0.44 & 3.39 & 0.03 \\ 
with & Crude:PS & 2.01 & 0.40 & 0.32 & 0.16 & 0.89 & 2.82 & 0.91 & 0.72 & 1.49 & 0.73 \\ 
replacement & WLS:cal & 2.00 & 0.07 & 0.07 & 0.00 & 0.94 & 0.45 & 0.35 & 0.38 & 2.54 & 0.01 \\ 
& WLS:PS & 2.00 & 0.07 & 0.06 & 0.00 & 0.95 & 2.72 & 0.60 & 0.58 & 0.88 & 0.79 \\ 
&  BC:cal & 2.00 & 0.07 & 0.07 & 0.00 & 0.94 & 0.45 & 0.35 & 0.39 & 2.54 & 0.01 \\ 
\multicolumn{2}{c}{Composite} & 4.31 & 0.81 & 0.81 & 5.95 & 0.20 & 8.30 & 1.07 & 1.08 & 40.75 & 0.00 \\ 
\multicolumn{2}{c}{Naive} & 0.45 & 0.53 & 0.55 & 2.69 & 0.19 & 1.58 & 0.85 & 0.84 & 0.90 & 0.93 \\ 
\multicolumn{2}{c}{DL} & 2.00 & 0.05 & & 0.00 & & 2.80 & 0.50 & & 0.88 &  \\  \\ 
\end{tabular}}
\end{table}
%%%%%%%%%%%%%%%%%%%%%%%%%%%%%%%%%%%%%%%%%%%%%%%%%%%%%%%%%%%%

%%%%%%%%%%%%%%%%%%%%%%%%%%%%%%%%%%%%%%%%%%%%%%%%%%%%%%%%%%%%
\begin{table}
\caption{\label{Tab:AppSimResWithInterSeveralcorrectXiDGMseq} \footnotesize Selected simulation results 
%for several $\xi$ values,
%with $\xi_{assm}=\xi=0.05,0.1,0.2$, 
for several $\xi$ values, and with $\xi_{assm}=\xi$, and $k=5$ covariates.
True outcome included $A$-$\bX_0$ interactions. 
Results described for matching on $\widehat{\widetilde{\pi}}^1_{as}(\bx_0)$ (PS) or using Mahalanobis distance with a caliper (cal).
OLS: least squares; WLS: weighted least squares; DL: model-based weighting estimator of \cite{ding2017principal}; Emp.SD: empirical standard deviation. Est.SE: estimated standard error; MSE: mean square error; CP95: empirical coverage proportion of 95\% confidence interval.}
\centering
\fbox{
%\tiny % \scriptsize
\tiny
\begin{tabular}{lccccccccccc}
& \multicolumn{5}{c}{Correctly specified principal score and outcome models} & \multicolumn{5}{c}{Misspecified principal score and outcome models} \\
\em Method & \em Estimator &  \em Mean & \em Emp.SD & \em Est.SE & \em MSE & \em CP95 & Mean & Emp.SD & Est.SE & MSE & CP95 \\ \hline
& & & & & & & & & \\
& \multicolumn{4}{c}{\textbf{Scenario A, SACE = 4.99}} & \multicolumn{2}{c}{\textbf{$\xi=0.05$}} &  \multicolumn{4}{c}{\textbf{Scenario A, SACE = 10.9}} \\
\hline
Matching & Crued:cal & 4.90 & 0.24 & 0.21 & 0.07 & 0.89 & 10.30 & 0.59 & 0.46 & 0.72 & 0.64 \\ 
&  Crued:PS & 4.79 & 0.36 & 0.28 & 0.17 & 0.84 & 10.76 & 0.70 & 0.69 & 0.51 & 0.93 \\ 
&  OLS:cal & 4.96 & 0.17 & 0.18 & 0.03 & 0.95 & 10.78 & 0.56 & 0.44 & 0.33 & 0.84 \\ 
&  OLS:PS & 5.00 & 0.18 & 0.18 & 0.03 & 0.94 & 11.40 & 0.56 & 0.54 & 0.56 & 0.87 \\ 
%& OLS:cal & 4.95 & 0.16 & 0.17 & 0.03 & 0.96 & 10.67 & 0.54 & 0.42 & 0.34 & 0.81 \\ 
Matching &  Crude:cal & 4.88 & 0.19 & 0.18 & 0.05 & 0.89 & 9.51 & 0.38 & 0.35 & 2.08 & 0.05 \\ 
with &  Crude:PS & 5.01 & 0.34 & 0.28 & 0.12 & 0.89 & 10.83 & 0.92 & 0.69 & 0.85 & 0.84 \\ 
replacement &  WLS:cal & 5.00 & 0.18 & 0.18 & 0.03 & 0.95 & 9.63 & 0.36 & 0.45 & 1.75 & 0.16 \\ 
&  WLS:PS & 5.00 & 0.19 & 0.19 & 0.04 & 0.94 & 11.38 & 0.69 & 0.69 & 0.70 & 0.94 \\ 
&  BC:cal& 5.00 & 0.18 & 0.22 & 0.03 & 0.98 & 9.63 & 0.36 & 0.41 & 1.75 & 0.12 \\ 
\multicolumn{2}{c}{Composite} & 4.20 & 0.64 & 0.64 & 1.03 & 0.78 & 9.85 & 0.88 & 0.88 & 1.88 & 0.76 \\ 
\multicolumn{2}{c}{Naive} & 3.92 & 0.44 & 0.42 & 1.34 & 0.29 & 10.95 & 0.68 & 0.69 & 0.46 & 0.96 \\ 
\multicolumn{2}{c}{DL} & 4.98 & 0.15 &  & 0.02 &  & 11.51 & 0.58 &  & 0.70 & \\
    & & & & & & & & & \\
 & \multicolumn{4}{c}{\textbf{Scenario B, SACE = 5.78}} & \multicolumn{2}{c}{\textbf{$\xi=0.05$}} & \multicolumn{4}{c}{\textbf{Scenario B, SACE = 11.39}} \\
 \hline
Matching &  Crued:cal & 5.70 & 0.23 & 0.17 & 0.06 & 0.83 & 11.60 & 0.59 & 0.42 & 0.40 & 0.83 \\ 
&  Crued:PS & 5.40 & 0.47 & 0.22 & 0.36 & 0.61 & 12.38 & 0.64 & 0.56 & 1.41 & 0.58 \\ 
&  OLS:cal & 5.76 & 0.14 & 0.15 & 0.02 & 0.96 & 11.81 & 0.53 & 0.39 & 0.46 & 0.78 \\ 
&  OLS:PS & 5.79 & 0.14 & 0.14 & 0.02 & 0.94 & 12.52 & 0.50 & 0.46 & 1.55 & 0.30 \\ 
%& OLS:cal & 5.75 & 0.12 & 0.14 & 0.01 & 0.98 & 11.75 & 0.52 & 0.38 & 0.40 & 0.80 \\ 
Matching & Crude:cal & 5.63 & 0.16 & 0.14 & 0.05 & 0.80 & 10.29 & 0.33 & 0.30 & 1.30 & 0.09 \\ 
with &  Crude:PS & 5.77 & 0.26 & 0.20 & 0.07 & 0.89 & 12.52 & 0.74 & 0.55 & 1.83 & 0.49 \\ 
replacement &  WLS:cal & 5.79 & 0.15 & 0.15 & 0.02 & 0.96 & 10.44 & 0.32 & 0.37 & 0.99 & 0.26 \\ 
&  WLS:PS & 5.79 & 0.15 & 0.15 & 0.02 & 0.94 & 12.43 & 0.56 & 0.57 & 1.41 & 0.58 \\ 
&  BC:cal& 5.79 & 0.15 & 0.17 & 0.02 & 0.98 & 10.44 & 0.32 & 0.36 & 0.99 & 0.25 \\ 
\multicolumn{2}{c}{Composite} & 5.41 & 0.59 & 0.58 & 0.48 & 0.91 & 14.59 & 0.80 & 0.79 & 10.89 & 0.01 \\ 
\multicolumn{2}{c}{Naive} & 4.43 & 0.33 & 0.33 & 1.92 & 0.02 & 11.85 & 0.58 & 0.58 & 0.54 & 0.90 \\ 
\multicolumn{2}{c}{DL} & 5.79 & 0.12 & & 0.02 &  & 12.63 & 0.47 &  & 1.78 & \\ 
& & & & & & & & & \\
 & \multicolumn{4}{c}{\textbf{Scenario A, SACE = 4.99}} & \multicolumn{2}{c}{\textbf{$\xi=0.1$}} &  \multicolumn{4}{c}{\textbf{Scenario A, SACE = 10.9}} \\
\hline
Matching &  Crued:cal & 4.88 & 0.25 & 0.21 & 0.07 & 0.87 & 10.41 & 0.61 & 0.48 & 0.62 & 0.73 \\ 
&  Crued:PS & 4.63 & 0.44 & 0.28 & 0.32 & 0.70 & 10.79 & 0.70 & 0.69 & 0.50 & 0.95 \\ 
&  OLS:cal & 4.94 & 0.17 & 0.18 & 0.03 & 0.96 & 10.93 & 0.54 & 0.45 & 0.29 & 0.89 \\ 
&  OLS:PS & 5.00 & 0.18 & 0.18 & 0.03 & 0.94 & 11.46 & 0.58 & 0.54 & 0.64 & 0.85 \\ 
%& OLS:cal & 4.93 & 0.16 & 0.17 & 0.03 & 0.97 & 10.80 & 0.52 & 0.44 & 0.28 & 0.88 \\ 
Matching &  Crude:cal & 4.89 & 0.20 & 0.18 & 0.05 & 0.90 & 9.52 & 0.40 & 0.35 & 2.07 & 0.05 \\ 
with &  Crude:PS & 5.02 & 0.34 & 0.27 & 0.12 & 0.89 & 10.85 & 0.95 & 0.69 & 0.90 & 0.85 \\ 
replacement &  WLS:cal & 5.00 & 0.19 & 0.18 & 0.04 & 0.94 & 9.64 & 0.36 & 0.45 & 1.72 & 0.17 \\ 
&  WLS:PS & 5.00 & 0.19 & 0.19 & 0.04 & 0.94 & 11.43 & 0.74 & 0.71 & 0.82 & 0.92 \\ 
&  BC:cal& 5.00 & 0.19 & 0.22 & 0.04 & 0.98 & 9.64 & 0.36 & 0.42 & 1.72 & 0.13 \\ 
\multicolumn{2}{c}{Composite} & 3.60 & 0.65 & 0.64 & 2.34 & 0.43 & 8.97 & 0.89 & 0.88 & 4.53 & 0.41 \\ 
\multicolumn{2}{c}{\textbf{Naive}} & 3.87 & 0.41 & 0.42 & 1.42 & 0.24 & 10.97 & 0.72 & 0.70 & 0.52 & 0.95 \\ 
\multicolumn{2}{c}{DL} & 4.99 & 0.16 &  & 0.02 &  & 11.54 & 0.58 & & 0.75 &  \\ 
& & & & & & & & & \\
  & \multicolumn{4}{c}{\textbf{Scenario A, SACE = 5.78}} & \multicolumn{2}{c}{\textbf{$\xi=0.1$}} &  \multicolumn{4}{c}{\textbf{Scenario B, SACE = 11.39}} \\
\hline
Matching &  Crued:cal & 5.68 & 0.25 & 0.17 & 0.07 & 0.79 & 11.69 & 0.58 & 0.43 & 0.43 & 0.82 \\ 
&  Crued:PS & 5.02 & 0.56 & 0.23 & 0.90 & 0.32 & 12.22 & 0.64 & 0.56 & 1.10 & 0.66 \\ 
&  OLS:cal & 5.73 & 0.14 & 0.15 & 0.02 & 0.96 & 11.96 & 0.47 & 0.41 & 0.55 & 0.71 \\ 
&  OLS:PS & 5.79 & 0.14 & 0.14 & 0.02 & 0.96 & 12.65 & 0.51 & 0.47 & 1.86 & 0.23 \\ 
%& OLS:cal & 5.72 & 0.12 & 0.14 & 0.02 & 0.97 & 11.89 & 0.46 & 0.40 & 0.46 & 0.74 \\ 
Matching &  Crude:cal & 5.62 & 0.16 & 0.14 & 0.05 & 0.79 & 10.29 & 0.31 & 0.30 & 1.30 & 0.07 \\ 
with &  Crude:PS & 5.78 & 0.25 & 0.20 & 0.06 & 0.89 & 12.47 & 0.76 & 0.55 & 1.74 & 0.50 \\ 
replacement &  WLS:cal & 5.78 & 0.14 & 0.15 & 0.02 & 0.96 & 10.44 & 0.32 & 0.37 & 0.99 & 0.26 \\ 
&  WLS:PS & 5.79 & 0.14 & 0.15 & 0.02 & 0.96 & 12.41 & 0.58 & 0.57 & 1.40 & 0.60 \\ 
&  BC:cal& 5.78 & 0.14 & 0.17 & 0.02 & 0.99 & 10.44 & 0.32 & 0.37 & 0.99 & 0.26 \\ 
\multicolumn{2}{c}{Composite} & 4.34 & 0.60 & 0.60 & 2.43 & 0.33 & 13.34 & 0.80 & 0.80 & 4.47 & 0.30 \\ 
\multicolumn{2}{c}{Naive} & 4.37 & 0.35 & 0.34 & 2.09 & 0.02 & 11.82 & 0.58 & 0.59 & 0.53 & 0.90 \\ 
\multicolumn{2}{c}{DL} & 5.78 & 0.13 &  & 0.02 &  & 12.66 & 0.45 &  & 1.84 &  \\ 
 & & & & & & & & & \\
 & \multicolumn{4}{c}{\textbf{Scenario B, SACE = 4.99}} & \multicolumn{2}{c}{\textbf{$\xi=0.1$}} & \multicolumn{4}{c}{\textbf{Scenario A, SACE = 10.9}} \\
 \hline
Matching & Crued:cal & 4.80 & 0.27 & 0.22 & 0.11 & 0.79 & 10.53 & 0.62 & 0.50 & 0.52 & 0.81 \\ 
&  Crued:PS & 4.18 & 0.54 & 0.30 & 0.95 & 0.36 & 10.80 & 0.69 & 0.70 & 0.48 & 0.96 \\ 
&  OLS:cal & 4.87 & 0.18 & 0.18 & 0.05 & 0.91 & 11.11 & 0.54 & 0.48 & 0.34 & 0.91 \\ 
&  OLS:PS & 5.00 & 0.18 & 0.18 & 0.03 & 0.94 & 11.62 & 0.63 & 0.56 & 0.92 & 0.77 \\ 
%& OLS:cal & 4.85 & 0.16 & 0.18 & 0.05 & 0.89 & 10.98 & 0.52 & 0.46 & 0.28 & 0.92 \\ 
Matching &  Crude:cal & 4.88 & 0.20 & 0.18 & 0.06 & 0.88 & 9.53 & 0.42 & 0.35 & 2.05 & 0.09 \\ 
with &  Crude:PS & 5.00 & 0.34 & 0.28 & 0.12 & 0.89 & 10.89 & 0.98 & 0.68 & 0.96 & 0.84 \\ 
replacement &  WLS:cal & 5.00 & 0.19 & 0.19 & 0.04 & 0.94 & 9.65 & 0.38 & 0.46 & 1.70 & 0.19 \\ 
&  WLS:PS & 5.00 & 0.19 & 0.19 & 0.04 & 0.94 & 11.46 & 0.74 & 0.72 & 0.87 & 0.92 \\ 
&  BC:cal& 5.00 & 0.19 & 0.23 & 0.04 & 0.98 & 9.65 & 0.38 & 0.42 & 1.70 & 0.15 \\ 
\multicolumn{2}{c}{Composite} & 2.42 & 0.68 & 0.64 & 7.07 & 0.03 & 7.40 & 0.84 & 0.88 & 12.90 & 0.02 \\ 
\multicolumn{2}{c}{Naive} & 3.77 & 0.44 & 0.44 & 1.68 & 0.22 & 10.97 & 0.72 & 0.72 & 0.52 & 0.96 \\ 
\multicolumn{2}{c}{DL} & 4.98 & 0.16 &  & 0.03 & & 11.61 & 0.59 &  & 0.86 &  \\ 
& & & & & & & & & \\
 & \multicolumn{4}{c}{\textbf{Scenario B, SACE = 5.78}} & \multicolumn{2}{c}{\textbf{$\xi=0.2$}} & \multicolumn{4}{c}{\textbf{Scenario B, SACE = 11.39}} \\
 \hline
Matching &  Crued:cal & 5.63 & 0.25 & 0.18 & 0.08 & 0.77 & 11.76 & 0.60 & 0.46 & 0.50 & 0.80 \\ 
&  Crued:PS & 4.42 & 0.45 & 0.26 & 2.05 & 0.04 & 12.01 & 0.61 & 0.57 & 0.75 & 0.80 \\ 
&  OLS:cal & 5.68 & 0.14 & 0.15 & 0.03 & 0.92 & 12.06 & 0.45 & 0.43 & 0.66 & 0.65 \\ 
&  OLS:PS & 5.79 & 0.14 & 0.15 & 0.02 & 0.95 & 13.03 & 0.58 & 0.50 & 3.02 & 0.09 \\ 
%& OLS:cal & 5.67 & 0.12 & 0.14 & 0.03 & 0.92 & 11.99 & 0.45 & 0.41 & 0.56 & 0.71 \\ 
Matching &  Crude:cal & 5.62 & 0.16 & 0.14 & 0.05 & 0.78 & 10.26 & 0.33 & 0.30 & 1.38 & 0.07 \\ 
with &  Crude:PS & 5.78 & 0.26 & 0.20 & 0.07 & 0.88 & 12.47 & 0.80 & 0.56 & 1.81 & 0.51 \\ 
replacement &  WLS:cal & 5.79 & 0.15 & 0.15 & 0.02 & 0.95 & 10.42 & 0.32 & 0.37 & 1.04 & 0.25 \\ 
&  WLS:PS & 5.79 & 0.15 & 0.15 & 0.02 & 0.95 & 12.49 & 0.60 & 0.59 & 1.57 & 0.57 \\ 
&  BC:cal& 5.79 & 0.15 & 0.18 & 0.02 & 0.98 & 10.42 & 0.32 & 0.37 & 1.04 & 0.24 \\ 
\multicolumn{2}{c}{Composite} & 2.57 & 0.60 & 0.62 & 10.66 & 0.00 & 11.16 & 0.83 & 0.83 & 0.74 & 0.94 \\ 
\multicolumn{2}{c}{Naive} & 4.28 & 0.33 & 0.35 & 2.37 & 0.01 & 11.84 & 0.58 & 0.60 & 0.54 & 0.90 \\ 
\multicolumn{2}{c}{DL} & 5.79 & 0.13 &  & 0.02 & & 12.73 & 0.47 &  & 2.02 & \\ 
\end{tabular}}
\end{table}
%%%%%%%%%%%%%%%%%%%%%%%%%%%%%%%%%%%%%%%%%%%%%%%%%%%%%%%%%%%%

%%%%%%%%%%%%%%%%%%%%%%%%%%%%%%%%%%%%%%%%%%%%%%%%%%%%%%%%%%%%
\begin{table}
\caption{\label{Tab:AppSimResWithInterSeveralwrongXiDGMseq} \footnotesize Selected simulation results for several $\xi$ values, with $\xi_{assm}=0$, and $k=5$ covariates.
True outcome included $A$-$\bX_0$ interactions. 
Results described for matching on $\widehat{\widetilde{\pi}}^1_{as}(\bx_0)$ (PS) or using Mahalanobis distance with a caliper (cal).
OLS: least squares with interactions; WLS: weighted least squares with interactions; DL: model-based weighting estimator of \cite{ding2017principal}; Emp.SD: empirical standard deviation. Est.SE: estimated standard error; MSE: mean square error; CP95: empirical coverage proportion of 95\% confidence interval.}
\centering
\fbox{
%\tiny % \scriptsize
\tiny
\begin{tabular}{lccccccccccc}
& \multicolumn{5}{c}{Correctly specified principal score and outcome models} & \multicolumn{5}{c}{Misspecified principal score and outcome models} \\
\em Method & \em Estimator &  \em Mean & \em Emp.SD & \em Est.SE & \em MSE & \em CP95 & Mean & Emp.SD & Est.SE & MSE & CP95 \\ \hline
& & & & & & & & & \\
& \multicolumn{4}{c}{\textbf{Scenario A, SACE = 4.99}} & \multicolumn{2}{c}{\textbf{$\xi=0.05$}} &  \multicolumn{4}{c}{\textbf{Scenario A, SACE = 10.9}} \\
\hline
Matching & Crued:cal & 4.89 & 0.25 & 0.21 & 0.07 & 0.89 & 10.20 & 0.56 & 0.45 & 0.82 & 0.58 \\ 
&  Crued:PS & 4.80 & 0.37 & 0.28 & 0.17 & 0.83 & 10.64 & 0.68 & 0.69 & 0.53 & 0.92 \\ 
&  OLS:cal & 4.98 & 0.17 & 0.18 & 0.03 & 0.96 & 10.75 & 0.56 & 0.43 & 0.34 & 0.84 \\ 
&  OLS:PS & 5.00 & 0.17 & 0.18 & 0.03 & 0.96 & 11.36 & 0.57 & 0.53 & 0.54 & 0.88 \\ 
%& OLS:cal & 4.96 & 0.15 & 0.17 & 0.02 & 0.97 & 10.62 & 0.53 & 0.42 & 0.36 & 0.79 \\ 
Matching &  Crude:cal & 4.86 & 0.20 & 0.18 & 0.06 & 0.88 & 9.52 & 0.40 & 0.35 & 2.09 & 0.06 \\ 
with &  Crude:PS & 5.04 & 0.35 & 0.28 & 0.12 & 0.86 & 10.70 & 0.92 & 0.68 & 0.89 & 0.83 \\ 
replacement &  WLS:cal & 5.01 & 0.18 & 0.18 & 0.03 & 0.96 & 9.65 & 0.36 & 0.45 & 1.71 & 0.18 \\ 
&  WLS:PS & 5.01 & 0.18 & 0.19 & 0.03 & 0.95 & 11.35 & 0.71 & 0.69 & 0.71 & 0.93 \\ 
&  BC:cal& 5.01 & 0.18 & 0.22 & 0.03 & 0.98 & 9.65 & 0.36 & 0.41 & 1.71 & 0.14 \\ 
 \multicolumn{2}{c}{Composite}  & 4.22 & 0.65 & 0.64 & 1.02 & 0.77 & 9.76 & 0.88 & 0.88 & 2.07 & 0.74 \\ 
 \multicolumn{2}{c}{Naive}  & 3.90 & 0.42 & 0.42 & 1.35 & 0.26 & 10.94 & 0.69 & 0.69 & 0.47 & 0.96 \\ 
 \multicolumn{2}{c}{DL}  & 4.94 & 0.15 &  & 0.02 &  & 11.31 & 0.57 &  & 0.48 &  \\ 
 & & & & & & & & & \\
& \multicolumn{4}{c}{\textbf{Scenario A, SACE = 5.78}} & \multicolumn{2}{c}{\textbf{$\xi=0.05$}} &  \multicolumn{4}{c}{\textbf{Scenario A, SACE = 11.39}} \\
\hline
Matching &  Crued:cal & 5.66 & 0.24 & 0.17 & 0.07 & 0.8 & 11.54 & 0.59 & 0.41 & 0.37 & 0.83 \\ 
&  Crued:PS & 5.40 & 0.46 & 0.22 & 0.35 & 0.6 & 12.35 & 0.63 & 0.57 & 1.32 & 0.61 \\ 
&  OLS:cal & 5.76 & 0.14 & 0.14 & 0.02 & 0.96 & 11.78 & 0.53 & 0.39 & 0.43 & 0.77 \\ 
&  OLS:PS & 5.79 & 0.14 & 0.14 & 0.02 & 0.96 & 12.49 & 0.47 & 0.47 & 1.44 & 0.33 \\ 
%& OLS:cal & 5.74 & 0.11 & 0.14 & 0.01 & 0.97 & 11.71 & 0.51 & 0.38 & 0.36 & 0.8 \\ 
Matching &  Crude:cal & 5.62 & 0.16 & 0.14 & 0.05 & 0.8 & 10.30 & 0.32 & 0.3 & 1.28 & 0.08 \\ 
with &  Crude:PS & 5.76 & 0.25 & 0.2 & 0.06 & 0.9 & 12.44 & 0.77 & 0.56 & 1.71 & 0.52 \\ 
replacement &  WLS:cal & 5.79 & 0.14 & 0.15 & 0.02 & 0.96 & 10.44 & 0.32 & 0.37 & 0.99 & 0.25 \\ 
&  WLS:PS & 5.79 & 0.14 & 0.15 & 0.02 & 0.96 & 12.37 & 0.55 & 0.57 & 1.28 & 0.63 \\ 
&  BC:cal& 5.79 & 0.14 & 0.17 & 0.02 & 0.98 & 10.44 & 0.32 & 0.36 & 0.99 & 0.25 \\ 
 \multicolumn{2}{c}{Composite}  & 5.42 & 0.58 & 0.58 & 0.46 & 0.9 & 14.59 & 0.78 & 0.78 & 10.86 & 0.02 \\ 
 \multicolumn{2}{c}{Naive}  & 4.42 & 0.34 & 0.33 & 1.94 & 0.03 & 11.80 & 0.57 & 0.58 & 0.50 & 0.91 \\ 
 \multicolumn{2}{c}{DL}  & 5.69 & 0.12 &  & 0.02 &  & 12.46 & 0.43 &  & 1.34 &  \\
& & & & & & & & & \\
& \multicolumn{4}{c}{\textbf{Scenario B, SACE = 4.99}} & \multicolumn{2}{c}{\textbf{$\xi=0.1$}} &  \multicolumn{4}{c}{\textbf{Scenario A, SACE = 10.9}} \\
\hline
Matching &  Crued:cal & 4.83 & 0.27 & 0.22 & 0.10 & 0.82 & 10.29 & 0.60 & 0.47 & 0.74 & 0.68 \\ 
&  Crued:PS & 4.65 & 0.44 & 0.29 & 0.31 & 0.74 & 10.61 & 0.69 & 0.69 & 0.56 & 0.91 \\ 
&  OLS:cal & 4.96 & 0.17 & 0.18 & 0.03 & 0.95 & 10.91 & 0.56 & 0.45 & 0.32 & 0.88 \\ 
&  OLS:PS & 5.00 & 0.18 & 0.18 & 0.03 & 0.94 & 11.39 & 0.58 & 0.54 & 0.57 & 0.86 \\ 
%& OLS:cal & 4.93 & 0.16 & 0.17 & 0.03 & 0.96 & 10.76 & 0.54 & 0.43 & 0.31 & 0.86 \\ 
Matching &  Crude:cal & 4.84 & 0.21 & 0.19 & 0.07 & 0.84 & 9.53 & 0.38 & 0.35 & 2.03 & 0.05 \\ 
with &  Crude:PS & 4.98 & 0.33 & 0.28 & 0.11 & 0.9 & 10.60 & 0.97 & 0.68 & 1.02 & 0.81 \\ 
replacement &  WLS:cal & 5.00 & 0.19 & 0.18 & 0.04 & 0.94 & 9.65 & 0.34 & 0.45 & 1.69 & 0.18 \\ 
&  WLS:PS & 5.00 & 0.19 & 0.19 & 0.04 & 0.94 & 11.31 & 0.70 & 0.69 & 0.65 & 0.95 \\ 
&  BC:cal& 5.00 & 0.19 & 0.22 & 0.04 & 0.98 & 9.65 & 0.34 & 0.42 & 1.69 & 0.14 \\ 
  \multicolumn{2}{c}{Composite}  & 3.57 & 0.63 & 0.64 & 2.41 & 0.4 & 8.93 & 0.86 & 0.88 & 4.65 & 0.39 \\ 
 \multicolumn{2}{c}{Naive}  & 3.86 & 0.42 & 0.42 & 1.46 & 0.26 & 10.98 & 0.70 & 0.71 & 0.50 & 0.95 \\ 
 \multicolumn{2}{c}{DL}  & 4.88 & 0.15 &  & 0.04 &  & 11.19 & 0.57 &  & 0.40 &  \\
 & & & & & & & & & \\
& \multicolumn{4}{c}{\textbf{Scenario A, SACE = 5.78}} & \multicolumn{2}{c}{\textbf{$\xi=0.1$}} &  \multicolumn{4}{c}{\textbf{Scenario A, SACE = 11.39}} \\
\hline
Matching &  Crued:cal & 5.63 & 0.26 & 0.17 & 0.09 & 0.74 & 11.71 & 0.61 & 0.43 & 0.48 & 0.8 \\ 
&  Crued:PS & 5.02 & 0.55 & 0.24 & 0.88 & 0.3 & 12.24 & 0.66 & 0.58 & 1.16 & 0.68 \\ 
&  OLS:cal & 5.75 & 0.13 & 0.15 & 0.02 & 0.97 & 11.99 & 0.50 & 0.41 & 0.62 & 0.66 \\ 
&  OLS:PS & 5.79 & 0.13 & 0.14 & 0.02 & 0.97 & 12.64 & 0.52 & 0.48 & 1.85 & 0.23 \\ 
%& OLS:cal & 5.72 & 0.11 & 0.14 & 0.01 & 0.97 & 11.92 & 0.48 & 0.4 & 0.52 & 0.71 \\ 
Matching &  Crude:cal & 5.62 & 0.15 & 0.14 & 0.05 & 0.8 & 10.29 & 0.34 & 0.3 & 1.32 & 0.08 \\ 
with &  Crude:PS & 5.76 & 0.26 & 0.21 & 0.07 & 0.88 & 12.44 & 0.78 & 0.56 & 1.71 & 0.55 \\ 
replacement &  WLS:cal & 5.79 & 0.14 & 0.15 & 0.02 & 0.97 & 10.43 & 0.33 & 0.37 & 1.01 & 0.26 \\ 
&  WLS:PS & 5.78 & 0.14 & 0.15 & 0.02 & 0.96 & 12.38 & 0.57 & 0.57 & 1.32 & 0.61 \\ 
&  BC:cal& 5.79 & 0.14 & 0.17 & 0.02 & 0.98 & 10.43 & 0.33 & 0.37 & 1.01 & 0.26 \\ 
 \multicolumn{2}{c}{Composite}  & 4.35 & 0.59 & 0.6 & 2.37 & 0.34 & 13.36 & 0.82 & 0.8 & 4.59 & 0.31 \\ 
 \multicolumn{2}{c}{Naive}  & 4.38 & 0.34 & 0.34 & 2.05 & 0.01 & 11.85 & 0.59 & 0.58 & 0.56 & 0.89 \\ 
 \multicolumn{2}{c}{DL}  & 5.59 & 0.12 &  & 0.05 &  & 12.36 & 0.45 &  & 1.14 &  \\
 & & & & & & & & & \\
& \multicolumn{4}{c}{\textbf{Scenario A, SACE = 4.99}} & \multicolumn{2}{c}{\textbf{$\xi=0.2$}} &  \multicolumn{4}{c}{\textbf{Scenario A, SACE = 10.9}} \\
\hline
Matching &  Crued:cal & 4.70 & 0.32 & 0.23 & 0.19 & 0.69 & 10.31 & 0.65 & 0.49 & 0.76 & 0.69 \\ 
&  Crued:PS & 4.16 & 0.54 & 0.31 & 0.98 & 0.34 & 10.56 & 0.71 & 0.69 & 0.61 & 0.91 \\ 
&  OLS:cal & 4.93 & 0.17 & 0.18 & 0.03 & 0.95 & 11.11 & 0.55 & 0.47 & 0.34 & 0.9 \\ 
&  OLS:PS & 5.00 & 0.18 & 0.18 & 0.03 & 0.95 & 11.52 & 0.62 & 0.55 & 0.77 & 0.8 \\ 
%& OLS:cal & 4.89 & 0.15 & 0.18 & 0.04 & 0.95 & 10.92 & 0.53 & 0.45 & 0.28 & 0.9 \\ 
Matching &  Crude:cal & 4.79 & 0.21 & 0.19 & 0.08 & 0.78 & 9.50 & 0.38 & 0.35 & 2.10 & 0.05 \\ 
with &  Crude:PS & 4.94 & 0.36 & 0.28 & 0.13 & 0.88 & 10.51 & 0.96 & 0.68 & 1.06 & 0.78 \\ 
replacement &  WLS:cal & 5.01 & 0.19 & 0.18 & 0.04 & 0.94 & 9.65 & 0.36 & 0.46 & 1.67 & 0.2 \\ 
&  WLS:PS & 5.00 & 0.18 & 0.19 & 0.03 & 0.96 & 11.34 & 0.70 & 0.7 & 0.69 & 0.93 \\ 
&  BC:cal& 5.01 & 0.19 & 0.22 & 0.04 & 0.98 & 9.65 & 0.36 & 0.42 & 1.67 & 0.14 \\ 
 \multicolumn{2}{c}{Composite}  & 2.42 & 0.63 & 0.64 & 6.99 & 0.02 & 7.38 & 0.86 & 0.88 & 13.11 & 0.02 \\ 
 \multicolumn{2}{c}{Naive}  & 3.74 & 0.41 & 0.44 & 1.74 & 0.18 & 10.94 & 0.73 & 0.72 & 0.53 & 0.94 \\ 
 \multicolumn{2}{c}{DL}  & 4.77 & 0.15 &  & 0.07 &  & 10.96 & 0.57 &  & 0.33 &  \\ 
& & & & & & & & & \\
& \multicolumn{4}{c}{\textbf{Scenario A, SACE = 5.78}} & \multicolumn{2}{c}{\textbf{$\xi=0.2$}} &  \multicolumn{4}{c}{\textbf{Scenario A, SACE = 11.39}} \\
\hline
Matching &  Crued:cal & 5.54 & 0.28 & 0.18 & 0.14 & 0.65 & 11.73 & 0.66 & 0.46 & 0.55 & 0.78 \\ 
&  Crued:PS & 4.42 & 0.45 & 0.27 & 2.06 & 0.04 & 11.96 & 0.66 & 0.59 & 0.76 & 0.83 \\ 
&  OLS:cal & 5.72 & 0.15 & 0.15 & 0.02 & 0.93 & 12.15 & 0.48 & 0.43 & 0.82 & 0.57 \\ 
&  OLS:PS & 5.79 & 0.15 & 0.15 & 0.02 & 0.95 & 13.04 & 0.58 & 0.51 & 3.05 & 0.1 \\ 
%& OLS:cal & 5.68 & 0.12 & 0.14 & 0.02 & 0.93 & 12.05 & 0.48 & 0.42 & 0.66 & 0.65 \\ 
Matching &  Crude:cal & 5.61 & 0.16 & 0.14 & 0.05 & 0.75 & 10.29 & 0.32 & 0.3 & 1.32 & 0.08 \\ 
with &  Crude:PS & 5.76 & 0.29 & 0.22 & 0.08 & 0.88 & 12.35 & 0.86 & 0.58 & 1.66 & 0.59 \\ 
replacement &  WLS:cal & 5.79 & 0.15 & 0.15 & 0.02 & 0.95 & 10.46 & 0.31 & 0.38 & 0.96 & 0.29 \\ 
&  WLS:PS & 5.79 & 0.15 & 0.15 & 0.02 & 0.96 & 12.47 & 0.61 & 0.59 & 1.54 & 0.58 \\ 
&  BC:cal& 5.79 & 0.15 & 0.18 & 0.02 & 0.98 & 10.46 & 0.31 & 0.37 & 0.96 & 0.27 \\ 
  \multicolumn{2}{c}{Composite}  & 2.55 & 0.64 & 0.62 & 10.86 & 0 & 11.18 & 0.83 & 0.83 & 0.73 & 0.94 \\ 
 \multicolumn{2}{c}{Naive}  & 4.27 & 0.35 & 0.35 & 2.39 & 0.01 & 11.85 & 0.62 & 0.6 & 0.60 & 0.88 \\ 
 \multicolumn{2}{c}{DL}  & 5.40 & 0.12 &  & 0.16 &  & 12.09 & 0.47 &  & 0.71 &  \\ 
\end{tabular}}
\end{table}
%%%%%%%%%%%%%%%%%%%%%%%%%%%%%%%%%%%%%%%%%%%%%%%%%%%%%%%%%%%%

%%%%%%%%%%%%%%%%%%%%%%%%%%%%%%%%%%%%%%%%%%%%%%%%%%%%%%%%%%%%%
\begin{table}
\caption{\label{Tab:AppSimResWithInterFullDGMmulti} \footnotesize
Selected simulation results when monotonicity holds ($\xi=\xi_{assm}=0$), for Scenarios A and B, with high $\pi_{pro}$, and $k=3, 5, 10$ covariates.
The principal strata were generated according to the multinomial logistic regression model.
% under a multinomial logistic regression model for the principal strata.
True outcome model included $A$-$\bX_0$ interactions. 
Results described for matching on $\widehat{\widetilde{\pi}}^1_{as}(\bx_0)$ (PS) or using Mahalanobis distance with a caliper (cal).
OLS: least squares with interactions; WLS: weighted least squares with interactions; DL: model-based weighting estimator of \cite{ding2017principal}; Emp.SD: empirical standard deviation. Est.SE: estimated standard error; MSE: mean square error; CP95: empirical coverage proportion of 95\% confidence interval.}
\centering
\fbox{
\scriptsize
\begin{tabular}{lccccccccccc}
& \multicolumn{5}{c}{Correctly specified principal score and outcome models} & \multicolumn{5}{c}{Misspecified principal score model} \\
\em Method & \em Estimator &  \em Mean & \em Emp.SD & \em Est.SE & \em MSE & \em CP95 & Mean & Emp.SD & Est.SE & MSE & CP95 \\ \hline
& & & & & & & & & \\
& \multicolumn{4}{c}{\textbf{Scenario A, SACE = 3.36}} & \multicolumn{2}{c}{\textbf{k=3}} &  \multicolumn{4}{c}{\textbf{Scenario A, SACE = 3.71}} \\
\hline
Matching 
& Crude &  3.33 & 0.13 & 0.13 & 0.02 & 0.94 & 3.67 & 0.14 & 0.12 & 0.02 & 0.91 \\
& OLS & 3.35 & 0.12 & 0.12 & 0.01 & 0.95 & 3.69 & 0.13 & 0.12 & 0.02 & 0.94 \\ 
& OLS-I & 3.36 & 0.13 & 0.12 & 0.02 & 0.95 & 3.71 & 0.13 & 0.12 & 0.02 & 0.94 \\  Matching 
& Crude & 3.34 & 0.13 & 0.12 & 0.02 & 0.94 & 3.69 & 0.14 & 0.12 & 0.02 & 0.90 \\ 
 with & BC & 3.36 & 0.13 & 0.14 & 0.02 & 0.98 & 3.71 & 0.14 & 0.14 & 0.02 & 0.96 \\ 
replacement& WLS & 3.35 & 0.13 & 0.13 & 0.02 & 0.97 & 3.70 & 0.14 & 0.13 & 0.02 & 0.95 \\ 
& WLS-I & 3.36 & 0.13 & 0.13 & 0.02 & 0.95 & 3.71 & 0.14 & 0.13 & 0.02 & 0.94 \\ 
\multicolumn{2}{c}{Composite}  & 10.73 & 0.51 & 0.52 & 54.58 & 0.00 & 8.47 & 0.53 & 0.53 & 22.91 & 0.00 \\ 

\multicolumn{2}{c}{Naive}  & 2.61 & 0.25 & 0.26 & 0.62 & 0.16 & 3.40 & 0.27 & 0.26 & 0.17 & 0.78 \\ 
\multicolumn{2}{c}{DL} & 3.36 & 0.11 &  & 0.01 &  & 3.84 & 0.13 &  & 0.03 &  \\
& & & & & & & & & \\
 & \multicolumn{4}{c}{\textbf{Scenario B, SACE = 3.38}} & \multicolumn{2}{c}{\textbf{k=3}} & \multicolumn{4}{c}{\textbf{Scenario B, SACE = 3.64}} \\
 \hline
Matching & Crude & 3.30 & 0.14 & 0.11 & 0.03 & 0.81 & 3.48 & 0.18 & 0.11 & 0.05 & 0.70 \\
& OLS & 3.36 & 0.09 & 0.10 & 0.01 & 0.96 & 3.62 & 0.11 & 0.10 & 0.01 & 0.93 \\ 
& OLS-I & 3.37 & 0.10 & 0.10 & 0.01 & 0.93 & 3.63 & 0.12 & 0.10 & 0.01 & 0.92 \\ 
Matching  & Crude & 3.32 & 0.11 & 0.10 & 0.01 & 0.90 & 3.62 & 0.12 & 0.10 & 0.02 & 0.88 \\ 
with  & BC & 3.37 & 0.11 & 0.13 & 0.01 & 0.97 & 3.64 & 0.12 & 0.13 & 0.02 & 0.96 \\ 
replacement  & WLS & 3.37 & 0.11 & 0.11 & 0.01 & 0.96 & 3.63 & 0.12 & 0.12 & 0.01 & 0.93 \\ 
& WLS-I & 3.37 & 0.11 & 0.11 & 0.01 & 0.94 & 3.64 & 0.12 & 0.11 & 0.02 & 0.91 \\ 
 \multicolumn{2}{c}{Composite}  & 5.57 & 0.45 & 0.46 & 5.02 & 0.00 & 8.98 & 0.56 & 0.45 & 28.87 & 0.00 \\ 
\multicolumn{2}{c}{Naive}   &  2.34 & 0.23 & 0.23 & 1.12 & 0.00 & 3.66 & 0.22 & 0.23 & 0.05 & 0.96 \\ 
\multicolumn{2}{c}{DL} & 3.38 & 0.09 &  & 0.01 &  & 3.79 & 0.14 &  & 0.04 &  \\ 

& & & & & & & & & \\
& \multicolumn{4}{c}{\textbf{Scenario A, SACE = 5.65}} & \multicolumn{2}{c}{\textbf{k=5}} &  \multicolumn{4}{c}{\textbf{Scenario A, SACE = 4.92}} \\
\hline
Matching 
& Crude & 5.60 & 0.19 & 0.18 & 0.04 & 0.93 & 4.93 & 0.20 & 0.19 & 0.04 & 0.93 \\
& OLS & 5.63 & 0.17 & 0.17 & 0.03 & 0.95 & 
4.94 & 0.16 & 0.16 & 0.03 & 0.95 \\ 
& OLS-I & 5.64 & 0.17 & 0.18 & 0.03 & 0.96 &
4.92 & 0.17 & 0.17 & 0.03 & 0.96 \\ Matching 
& Crude & 5.62 & 0.19 & 0.18 & 0.04 & 0.93& 4.92 & 0.20 & 0.18 & 0.04 & 0.92 \\ 
 with & BC & 5.64 & 0.18 & 0.20 & 0.03 & 0.98 &
 4.92 & 0.18 & 0.20 & 0.03 & 0.97 \\ 
 replacement& WLS & 5.64 & 0.18 & 0.18 & 0.03 &0.96 & 
4.93 & 0.18 & 0.18 & 0.03 & 0.96 \\ 
& WLS-I & 5.64 & 0.18 & 0.18 & 0.03 & 0.95 & 4.92 & 0.18 & 0.18 & 0.03 & 0.94 \\
\multicolumn{2}{c}{Composite}  & 13.0 & 0.59 & 0.60 & 54.31 & 0.00 & 11.3 & 0.64 & 0.62 & 40.58 & 0.00 \\ 

\multicolumn{2}{c}{Naive}  & 4.69 & 0.36 & 0.36 & 1.05 & 0.23 & 4.47 & 0.37 & 0.37 & 0.34 & 0.78 \\ 
\multicolumn{2}{c}{DL} & 5.65 & 0.16 &  & 0.03 &  & 5.14 & 0.17 &  & 0.08 &\\
& & & & & & & & & \\
 & \multicolumn{4}{c}{\textbf{Scenario B, SACE = 6}} & \multicolumn{2}{c}{\textbf{k=5}} &  \multicolumn{4}{c}{\textbf{Scenario B, SACE = 4.86}} \\
 \hline
Matching & Crude & 5.89 & 0.22 & 0.16 & 0.06 & 0.80
& 4.43 & 0.30 & 0.19 & 0.28 & 0.43 \\ 
& OLS & 5.96 & 0.11 & 0.13 & 0.01 & 0.97 
& 4.78 & 0.14 & 0.14 & 0.03 & 0.92 \\ 
& OLS-I & 5.98 & 0.13 & 0.14 & 0.02 & 0.96 & 4.86 & 0.16 & 0.14 & 0.02 & 0.91 \\ 
Matching  & Crude & 5.80 & 0.15 & 0.14 & 0.06 & 0.72 
& 4.75 & 0.17 & 0.15 & 0.04 & 0.86 \\ 
with  & BC & 5.99 & 0.14 & 0.17 & 0.02 & 0.98 &
4.86 & 0.16 & 0.18 & 0.03 & 0.96 \\ 
replacement  & WLS & 5.95 & 0.13 & 0.15 & 0.02 & 0.96 
& 4.85 & 0.16 & 0.16 & 0.02 & 0.95 \\
& WLS-I & 5.99 & 0.14 & 0.15 & 0.02 & 0.96 & 4.86 & 0.16 & 0.15 & 0.03 & 0.92 \\
 \multicolumn{2}{c}{Composite}  & 7.51 & 0.55 & 0.55 & 2.61 & 0.21 & 10.3 & 0.66 & 0.55 & 29.71 & 0.00 \\ 
\multicolumn{2}{c}{Naive}   & 4.22 & 0.31 & 0.31 & 3.25 & 0.00 & 4.53 & 0.34 & 0.33 & 0.23 & 0.81\\ 
\multicolumn{2}{c}{DL} & 6.00 & 0.12 &  & 0.015 & & 4.95 & 0.16 &  & 0.035 &\\
& & & & & & & & & \\
& \multicolumn{4}{c}{\textbf{Scenario A, SACE = 9.48}} & \multicolumn{2}{c}{\textbf{k=10}} & \multicolumn{4}{c}{\textbf{Scenario A, SACE = 8.63}} \\
\hline
Matching 
& Crude & 9.34 & 0.30 & 0.28 & 0.11 & 0.90 & 8.72 & 0.34 & 0.30 & 0.12 & 0.92 \\
& OLS &  9.45 & 0.22 & 0.22 & 0.05 & 0.94 & 8.66 & 0.23 & 0.22 & 0.06 & 0.93 \\ 
& OLS-I & 9.48 & 0.23 & 0.24 & 0.05 & 0.96 & 8.63 & 0.25 & 0.24 & 0.06 & 0.95 \\  Matching 
& Crude & 9.38 & 0.29 & 0.26 & 0.09 & 0.90 & 8.69 & 0.33 & 0.28 & 0.11 & 0.89 \\ 
 with & BC & 9.49 & 0.24 & 0.29 & 0.06 & 0.97 & 8.63 & 0.25 & 0.29 & 0.06 & 0.98 \\ 
 replacement & WLS & 9.46 & 0.24 & 0.24 & 0.06 & 0.96 & 8.66 & 0.25 & 0.24 & 0.06 & 0.94 \\ 
& WLS-I & 9.49 & 0.24 & 0.25 & 0.06 & 0.95 & 8.63 & 0.25 & 0.25 & 0.06 & 0.94 \\ 
\multicolumn{2}{c}{Composite}  & 17.20 & 0.78 & 0.77 & 60.13 & 0.00 & 15.83 & 0.81 & 0.79 & 52.49 & 0.00 \\ 
\multicolumn{2}{c}{Naive}  & 6.71 & 0.50 & 0.50 & 7.93 & 0.00 & 6.83 & 0.52 & 0.51 & 3.50 & 0.06 \\ 
\multicolumn{2}{c}{DL} & 9.48 & 0.23 &  & 0.05 &  & 8.91 & 0.25 &  & 0.15 &  \\ 
& & & & & & & & & \\
 & \multicolumn{4}{c}{\textbf{Scenario B, SACE = 9.12}} & \multicolumn{2}{c}{\textbf{k=10}} & \multicolumn{4}{c}{\textbf{Scenario B, SACE = 8.34}} \\
 \hline
Matching & Crude & 8.99 & 0.33 & 0.26 & 0.12 & 0.87 & 7.5 & 0.44 & 0.32 & 0.919 & 0.29 \\ 
& OLS & 9.08 & 0.15 & 0.18 & 0.02 & 0.97 & 8.17 & 0.17 & 0.19 & 0.06 & 0.87 \\ 
& OLS-I & 9.10 & 0.19 & 0.20 & 0.03 & 0.96 & 8.33 & 0.21 & 0.20 & 0.04 & 0.93 \\ 
Matching  & Crude & 8.73 & 0.26 & 0.23 & 0.21 & 0.58 & 7.96 & 0.31 & 0.27 & 0.25 & 0.67 \\ 
with  & BC & 9.12 & 0.19 & 0.24 & 0.04 & 0.99 & 8.33 & 0.22 & 0.26 & 0.05 & 0.97 \\ 
replacement & WLS & 9.04 & 0.18 & 0.20 & 0.04 & 0.95 & 8.26 & 0.20 & 0.21 & 0.04 & 0.95 \\ 
& WLS-I & 9.12 & 0.19 & 0.20 & 0.04 & 0.96 & 8.33 & 0.22 & 0.21 & 0.05 & 0.94 \\ 
 \multicolumn{2}{c}{Composite} & 10.42 & 0.69 & 0.70 & 2.17 & 0.53 & 13.87 & 0.80 & 0.70 & 31.18 & 0.00 \\ 
\multicolumn{2}{c}{Naive}  & 6.67 & 0.46 & 0.46 & 6.18 & 0.00 & 7.01 & 0.46 & 0.46 & 2.00 & 0.17 \\ 
\multicolumn{2}{c}{DL} & 9.13 & 0.17 &  & 0.03 &  & 8.42 & 0.22 &  & 0.05 &  \\ 
\end{tabular}}
\end{table}

% MAHAL AND PS 
% start new scale
%%%%%%%%%%%%%%%%%%%%%%%%%%%%%%%%%%%%%%%%%%%%%%%%%%%%%%%%%%%%%
%%%%%%%%%%%%%%%%%%%%%%%%%%%%%%%%%%%%%%%%%%%%%%%%%%%%%%%%%%%%%
\begin{figure}
\centering
\caption{\footnotesize{Bias of different estimators, with $k=5$, for several $\xi$ values, and $\xi_{assm}=\xi$, under correctly specified models (top left), misspecified principal score model (top right), misspecified outcome model (bottom  left), and misspecified principal score model and outcome model (bottom  right). True outcome model outcome included $A$-$\bX_0$ interactions. 
%Results described for matching on 
Matching was on the Mahalanobis distance without (Mahal) or with a caliper (Mahal caliper), or on $\widehat{\widetilde{\pi}}^1_{as}(\bx_0)$ (PS).
WLS: weighted least squares with interactions;
OLS: ordinary least squares with interactions; 
Crude: Crude mean difference.
Weighting (DL): model-based weighting estimator of \cite{ding2017principal}.
The true SACE parameter ranged between $4.99 - 5.78$ under the first outcome model (with the original covariates), and between $10.9 - 12.01$ under the second outcome model (with the transformed covariates).
The OLS estimator after matching without replacement on the Mahalanobis distance with a caliper is affected by $\xi$ values, because as $\xi$ increases, $\pi_{as}$ decreases while $\pi_{pro}$ increases, resulting in smaller pool of always-survivors in $\{A=1, S=1\}$. One may expect that matching without replacement is more frail to a smaller pool than matching with replacement.
\label{Fig:AppbiasS2withInterDGMseq}}} \begin{minipage}{.5\textwidth}
\includegraphics[scale=0.415]{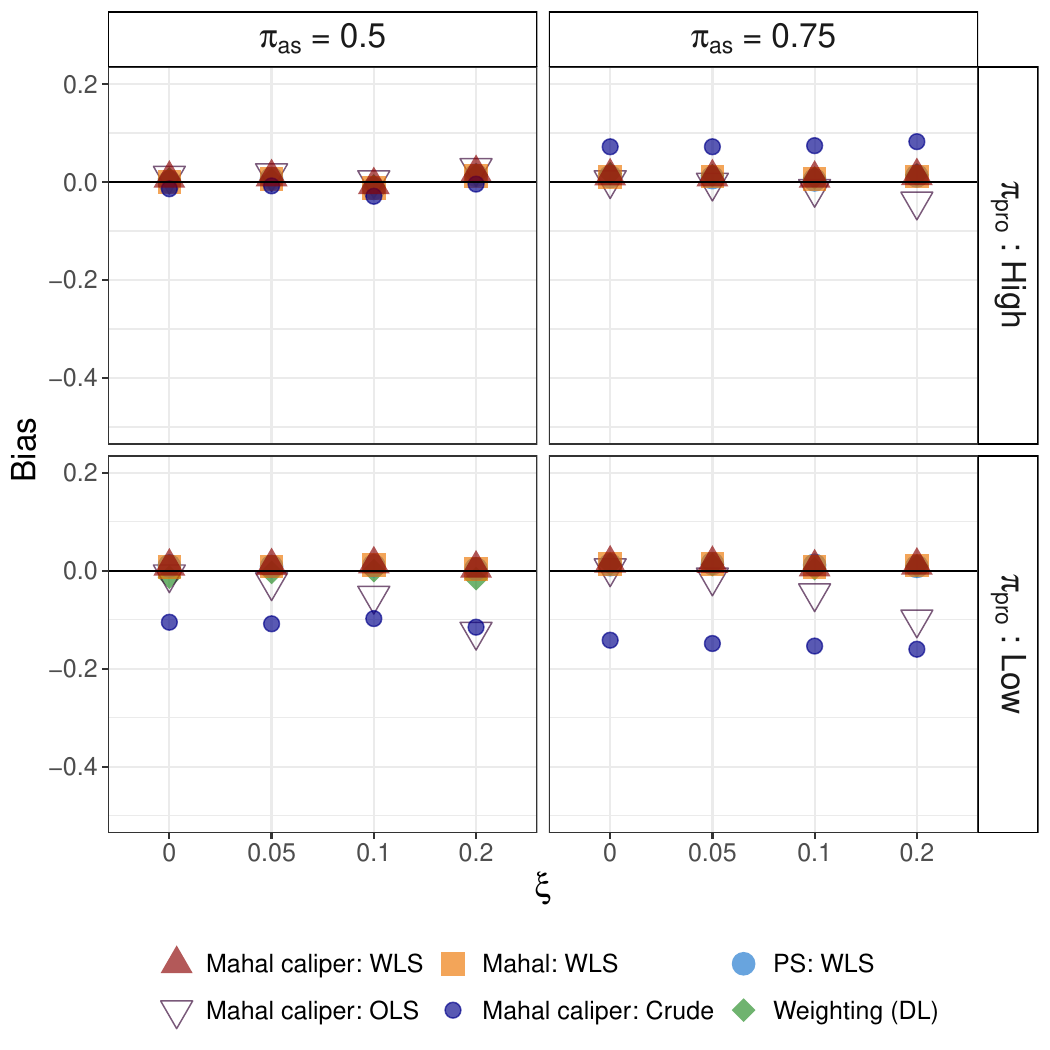}
% general correct wout inter
\end{minipage}%
\begin{minipage}{.475\textwidth}
\centering
\includegraphics[scale=0.415]{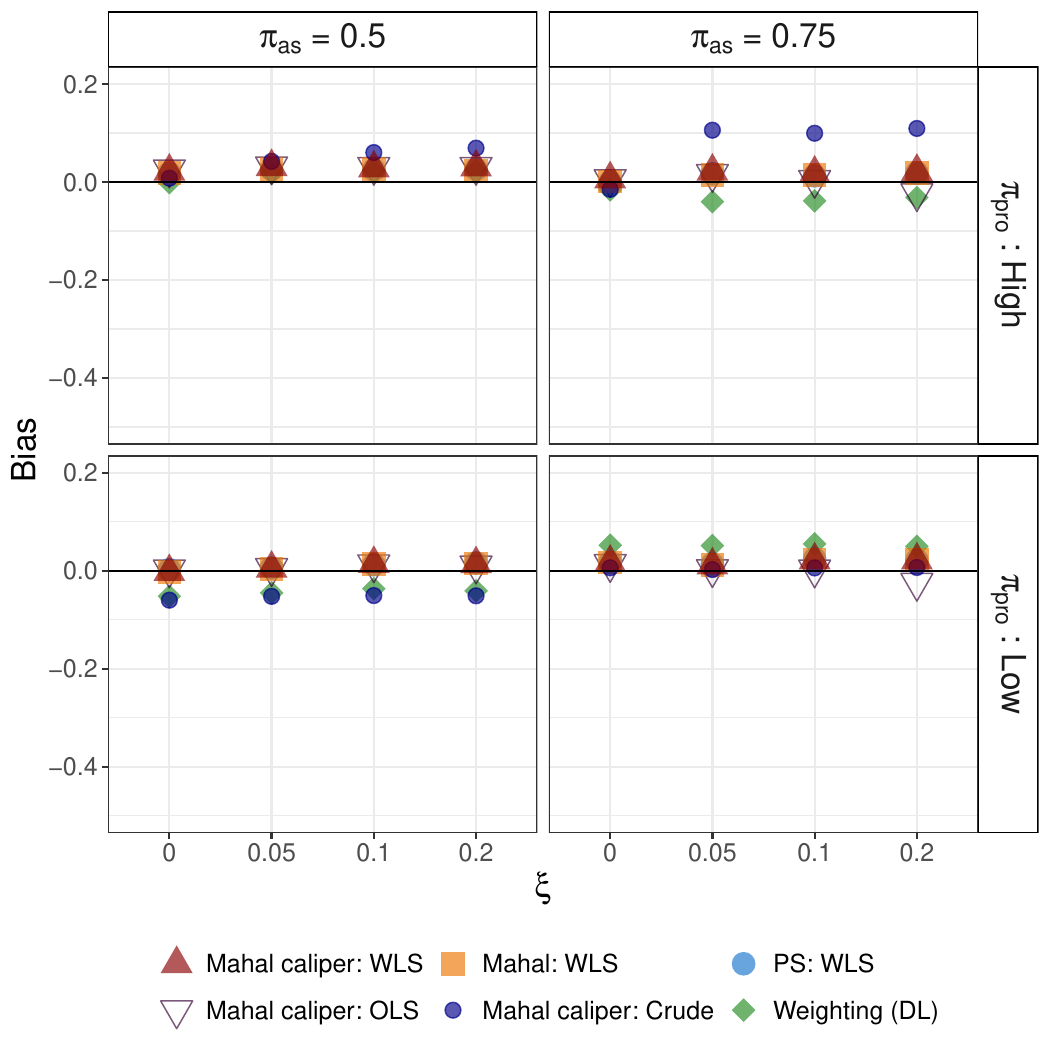}
\end{minipage}

\begin{minipage}{.5\textwidth}
\includegraphics[scale=0.415]{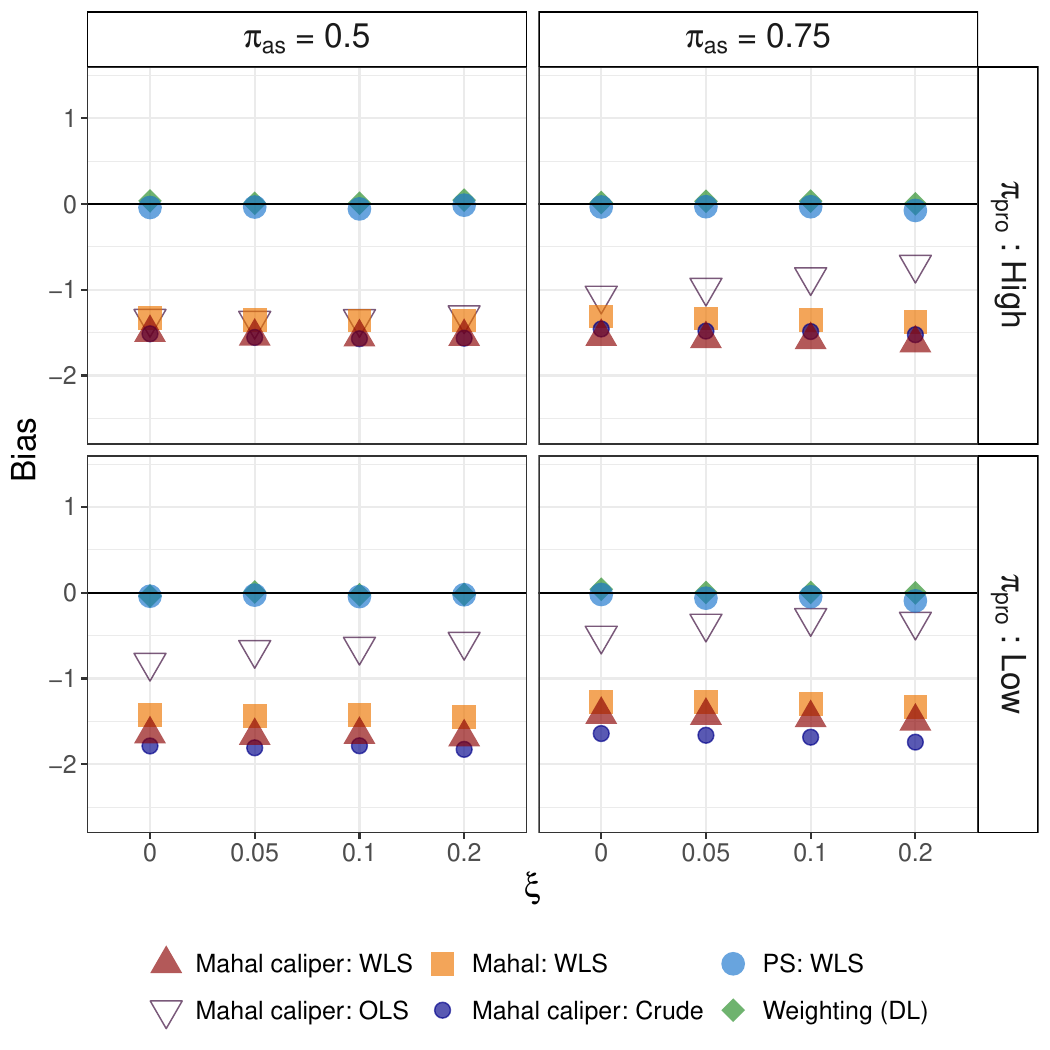}
% general correct wout inter
\end{minipage}%
\begin{minipage}{.475\textwidth}
\centering
\includegraphics[scale=0.415]{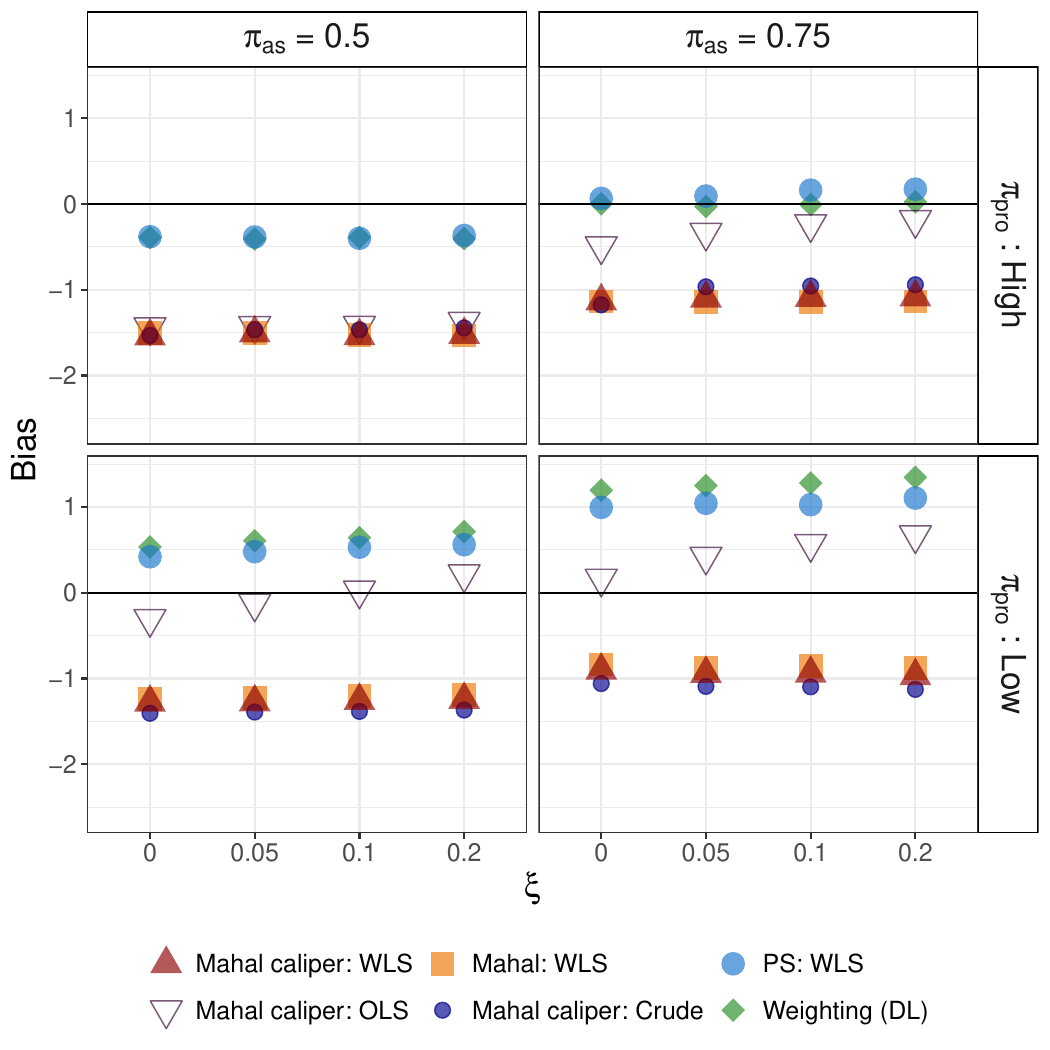}
\end{minipage}
\end{figure}
%%%%%%%%%%%%%%%%%%%%%%%%%%%%%%%%%%%%%%%%%%%%%%%%%%%%%%%%%%%%%

%%%%%%%%%%%%%%%%%%%%%%%%%%%%%%%%%%%%%%%%%%%%%%%%%%%%%%%%%%%%%
\begin{figure}
\centering
\caption{\footnotesize{Bias of different estimators, with $k=5$, for several $\xi$ values, and $\xi_{assm}=0$, under correctly specified models (top left), misspecified principal score model (top right), misspecified outcome model (bottom  left), and misspecified principal score model and outcome model (bottom  right). True outcome model outcome included $A$-$\bX_0$ interactions.
Matching was on the Mahalanobis distance without (Mahal) or with a caliper (Mahal caliper), or on $\widehat{\widetilde{\pi}}^1_{as}(\bx_0)$ (PS).
WLS: weighted least squares with interactions;
OLS: ordinary least squares with interactions; 
Crude: Crude mean difference.
Weighting (DL): model-based weighting estimator of \cite{ding2017principal}.
% The WLS estimators and the Crude estimator followed matching with replacement. The OLS estimator followed matching without replacement.
The true SACE parameter ranged between $4.99 - 5.78$ under the first outcome model (with the original covariates), and between $10.9 - 12.01$ under the second outcome model (with the transformed covariates). \label{Fig:AppbiasS3withInterDGMseq}}} \begin{minipage}{.5\textwidth}
\includegraphics[scale=0.415]{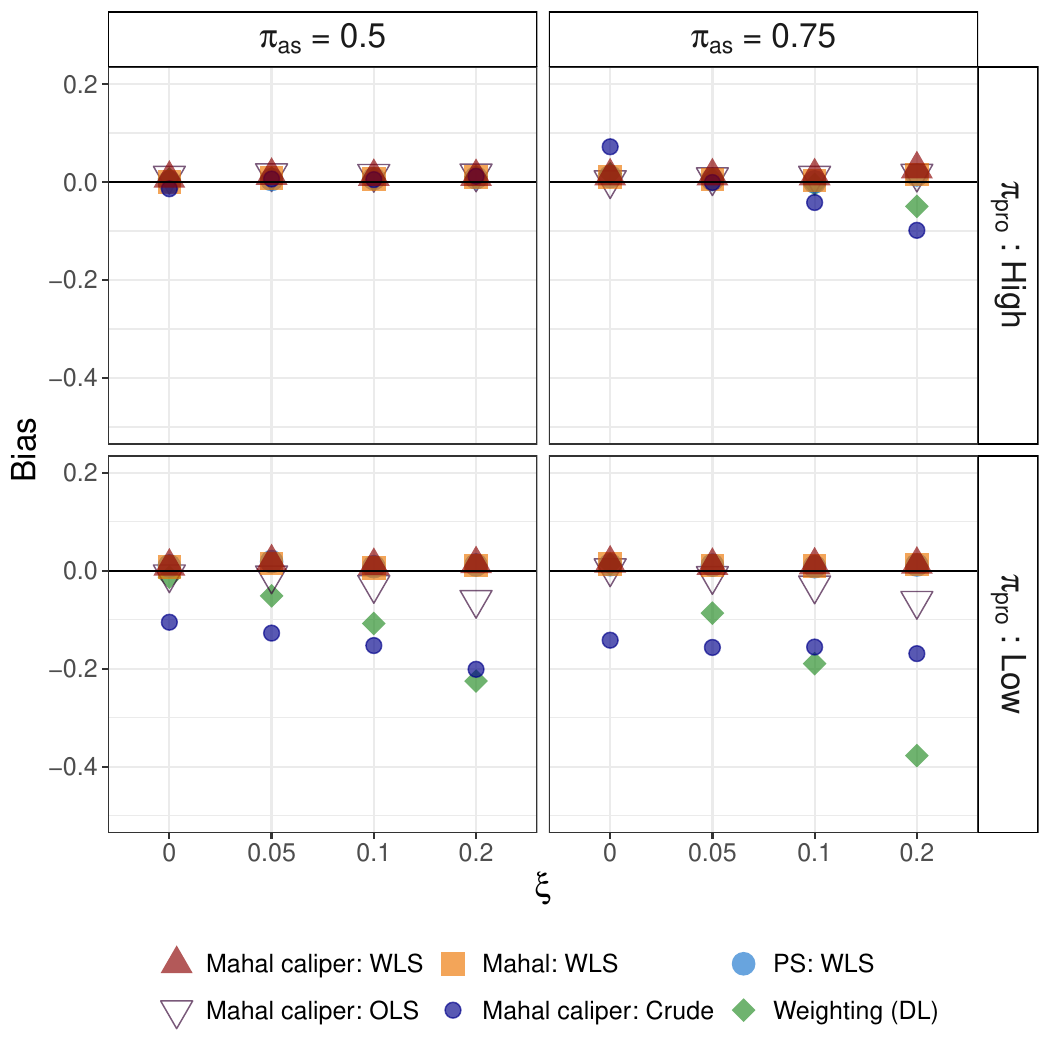}
% general correct wout inter
\end{minipage}%
\begin{minipage}{.475\textwidth}
\centering
\includegraphics[scale=0.415]{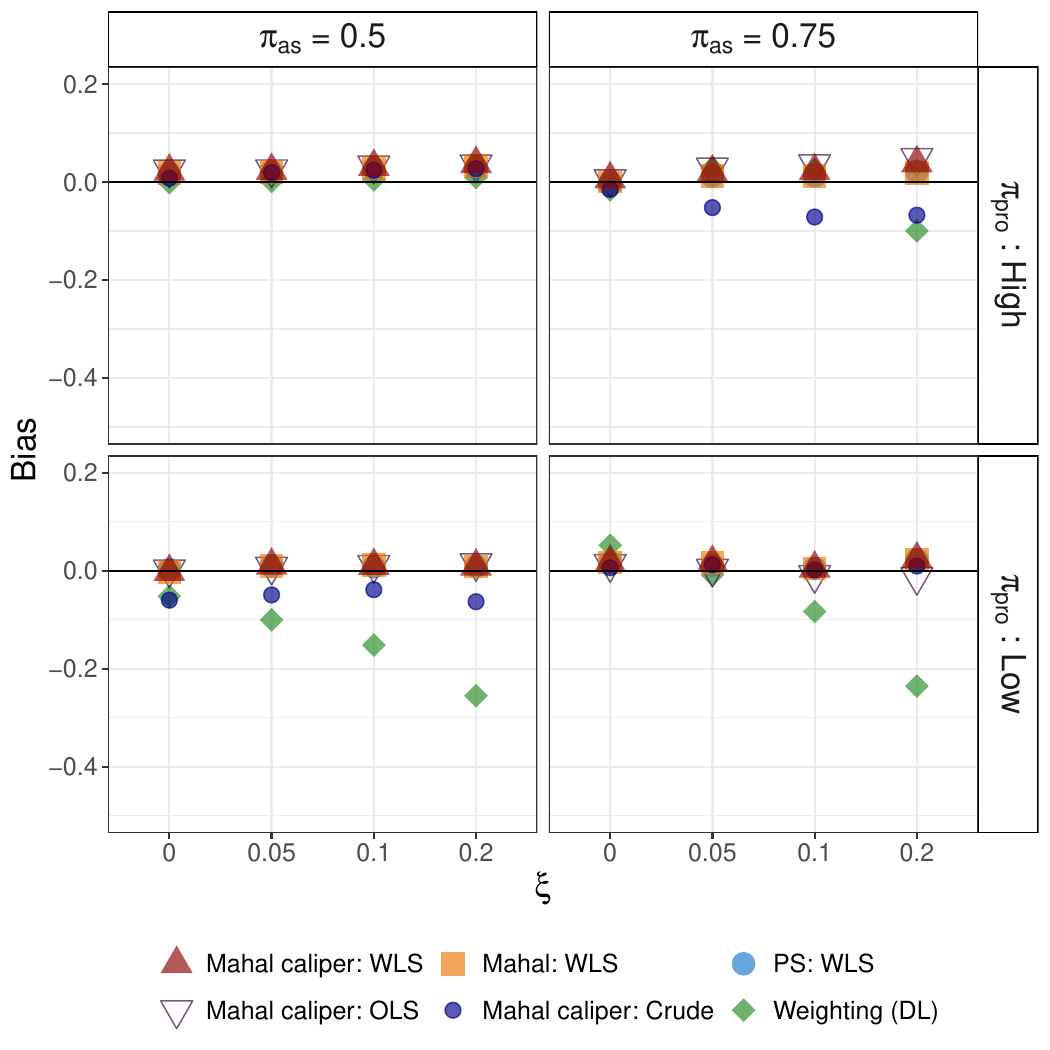}
\end{minipage}

\begin{minipage}{.5\textwidth}
\includegraphics[scale=0.415]{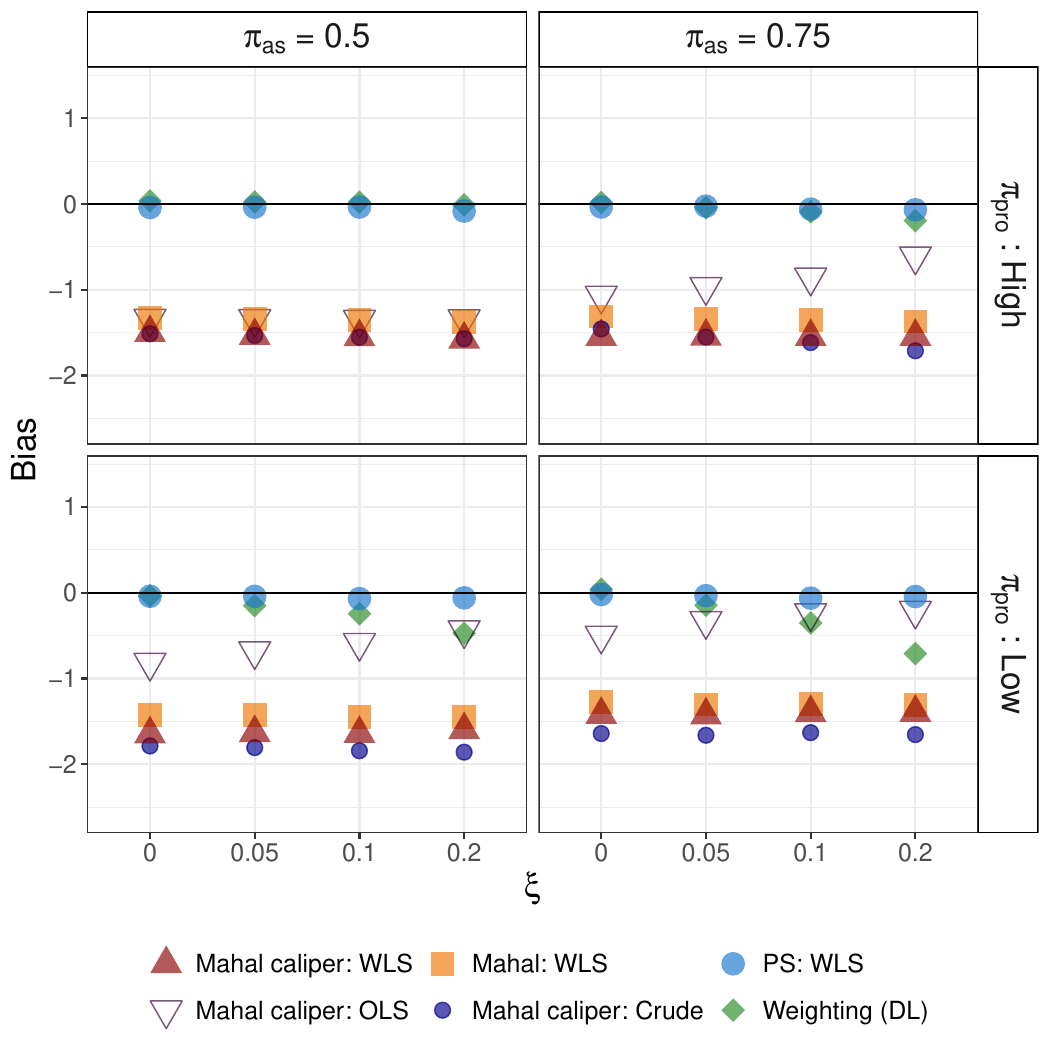}
% general correct wout inter
\end{minipage}%
\begin{minipage}{.475\textwidth}
\centering
\includegraphics[scale=0.415]{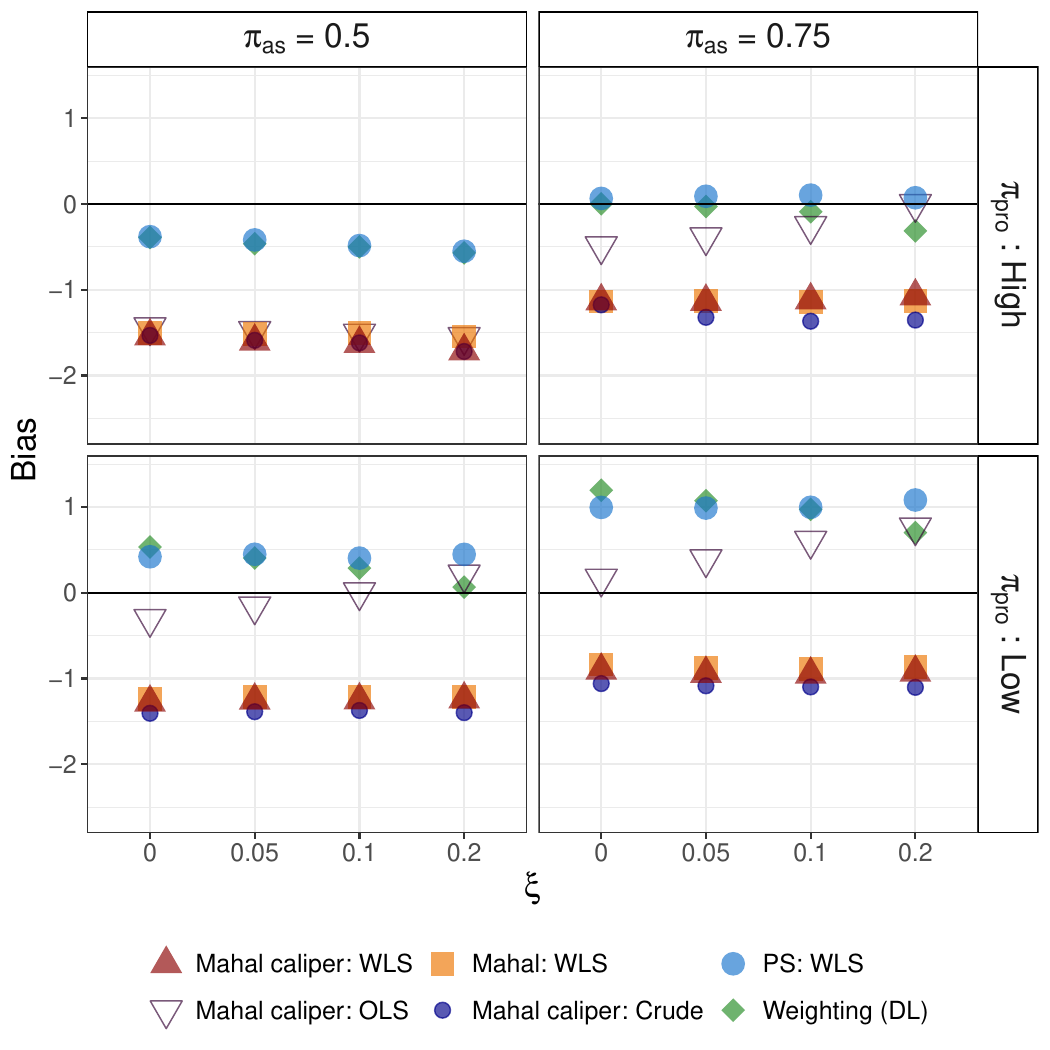}
\end{minipage}
\end{figure}
%%%%%%%%%%%%%%%%%%%%%%%%%%%%%%%%%%%%%%%%%%%%%%%%%%%%%%%%%%%%%
%%%%%%%%%%%%%%%%%%%%%%%%%%%%%%%%%%%%%%%%%%%%%%%%%%%%%%%%%%%%%
\begin{figure}
\centering
\caption{\footnotesize{Bias of different estimators, with $k=3, 5, 10$ covariates, when monotonicity holds ($\xi_{assm}=\xi=0$), under correctly specified models (top left), misspecified principal score model (top right), misspecified outcome model (bottom  left), and misspecified principal score model and outcome model (bottom  right). True outcome model outcome did not include $A$-$\bX_0$ interactions. 
Matching was on the Mahalanobis distance without (Mahal) or with a caliper (Mahal caliper), or on $\widehat{\widetilde{\pi}}^1_{as}(\bx_0)$ (PS).
WLS: weighted least squares with interactions;
OLS: ordinary least squares with interactions; 
Crude: Crude mean difference.
Weighting (DL): model-based weighting estimator of \cite{ding2017principal}.
% The WLS estimators and the Crude estimator followed matching with replacement. The OLS estimator followed matching without replacement.
The true SACE parameter is 2 for all $k$. \label{Fig:AppbiasS1woutInterDGMseq}}} \begin{minipage}{.5\textwidth}
\includegraphics[scale=0.415]{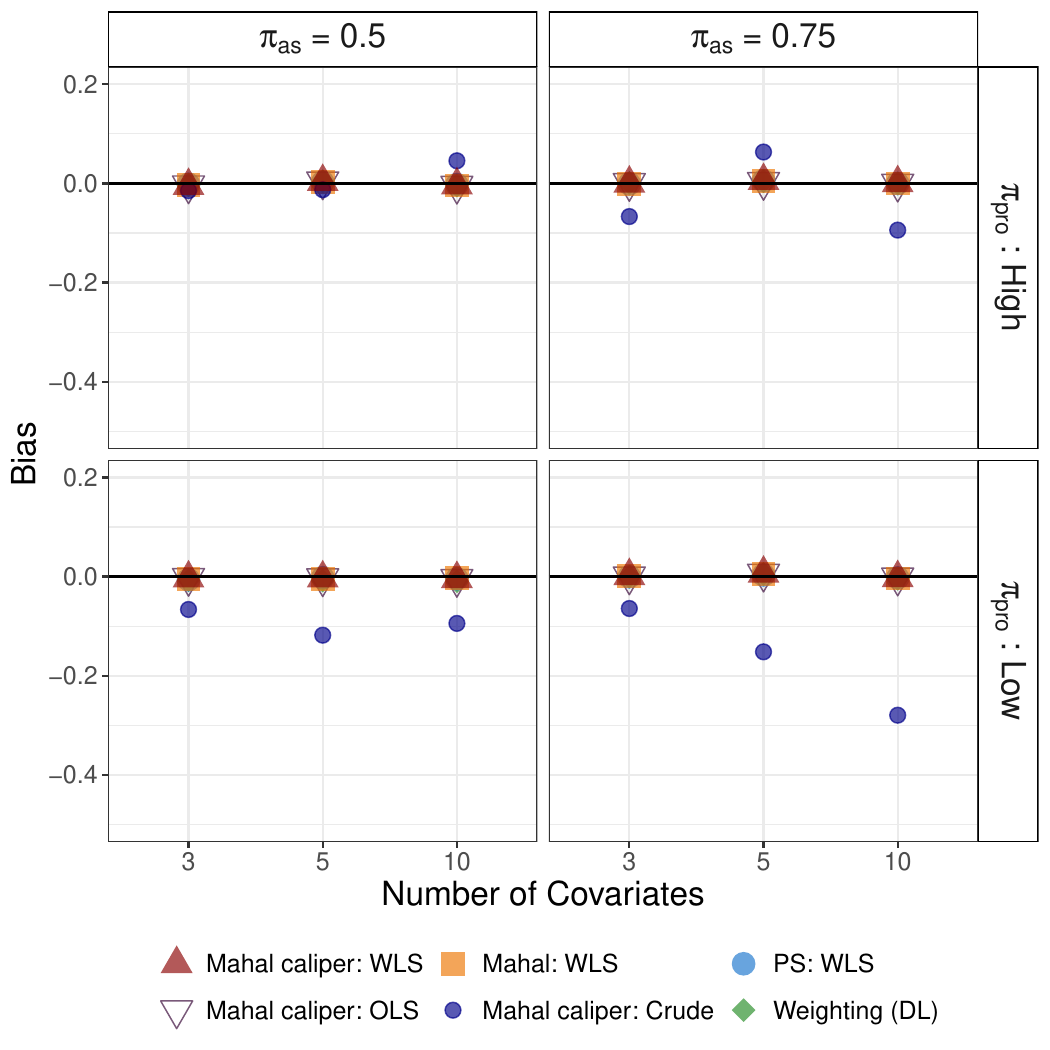}
% general correct wout inter
\end{minipage}%
\begin{minipage}{.475\textwidth}
\centering
\includegraphics[scale=0.415]{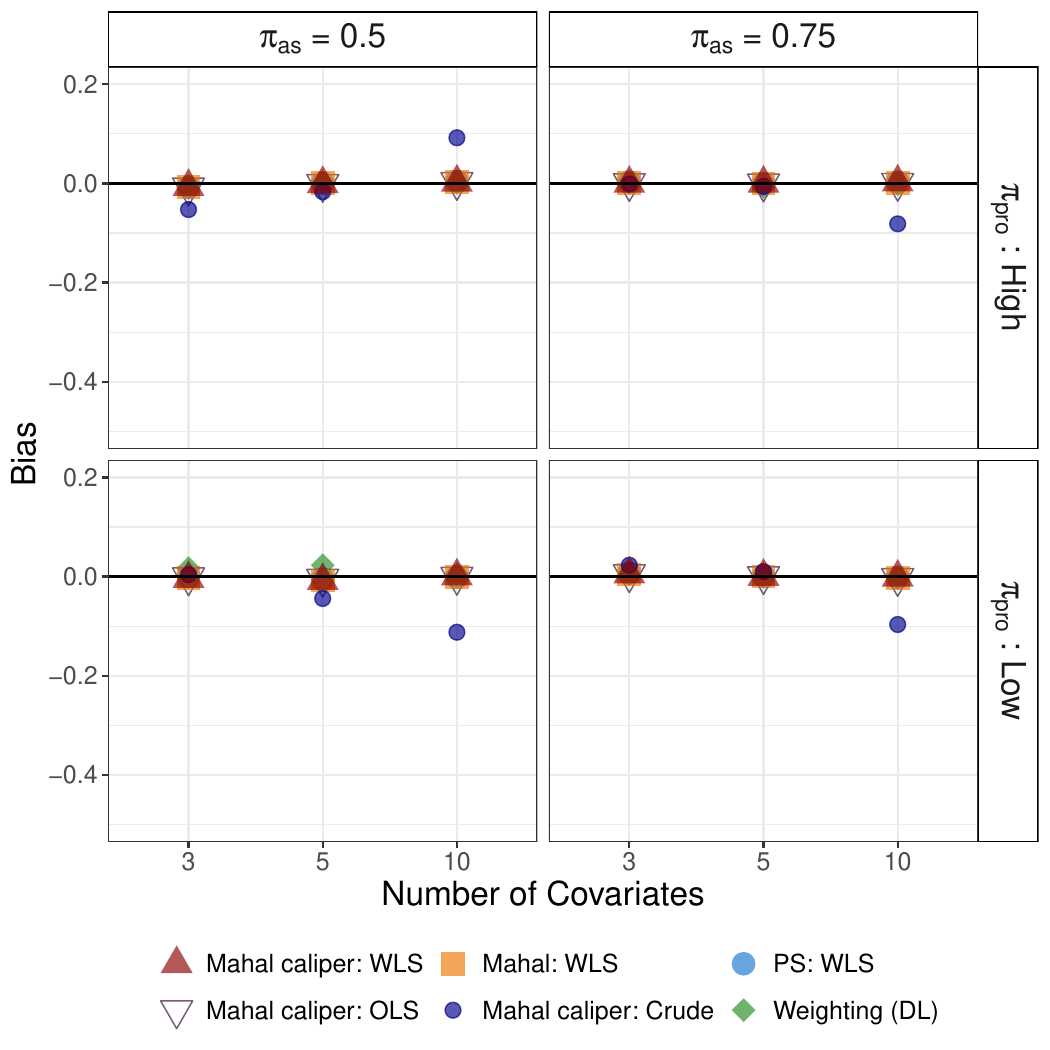}
\end{minipage}

\begin{minipage}{.5\textwidth}
\includegraphics[scale=0.415]{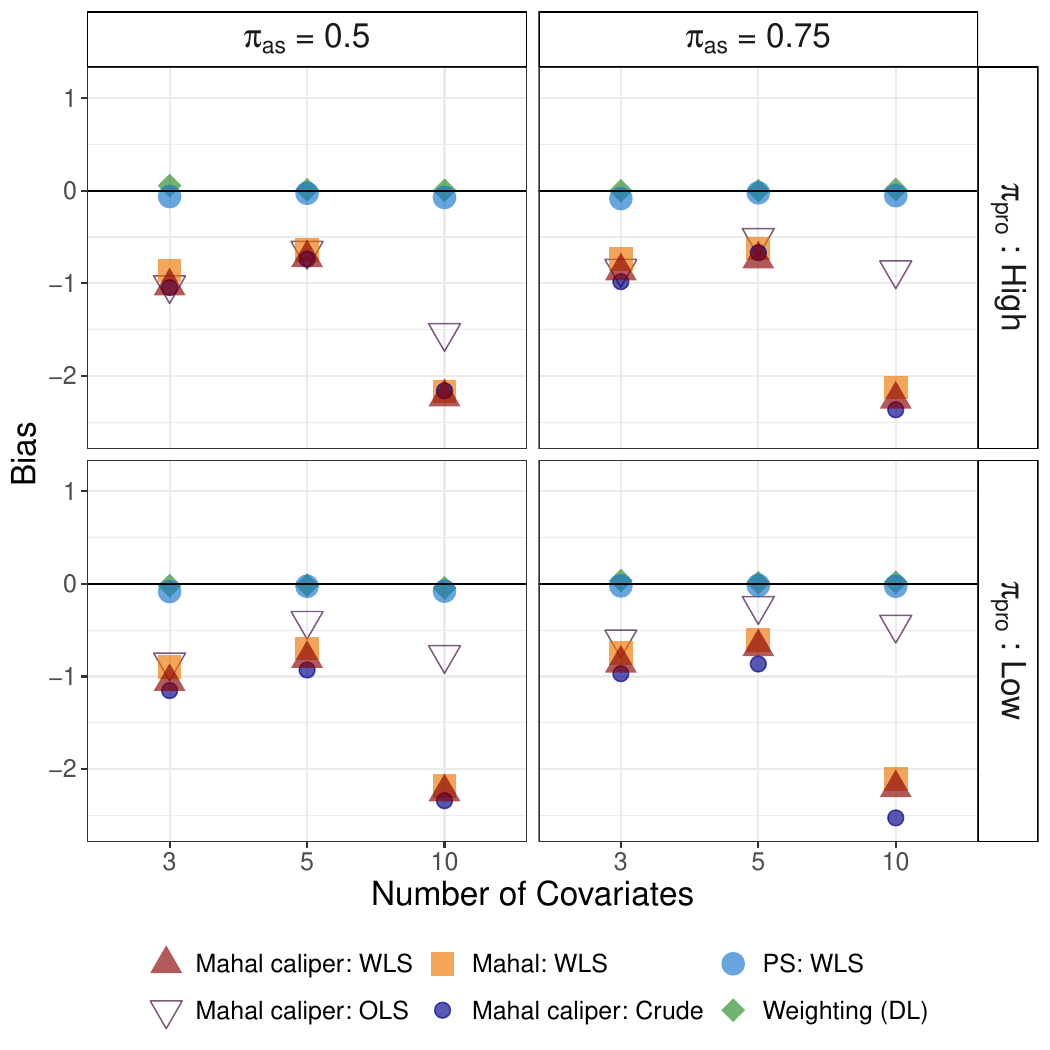}
% general correct wout inter
\end{minipage}%
\begin{minipage}{.475\textwidth}
\centering
\includegraphics[scale=0.415]{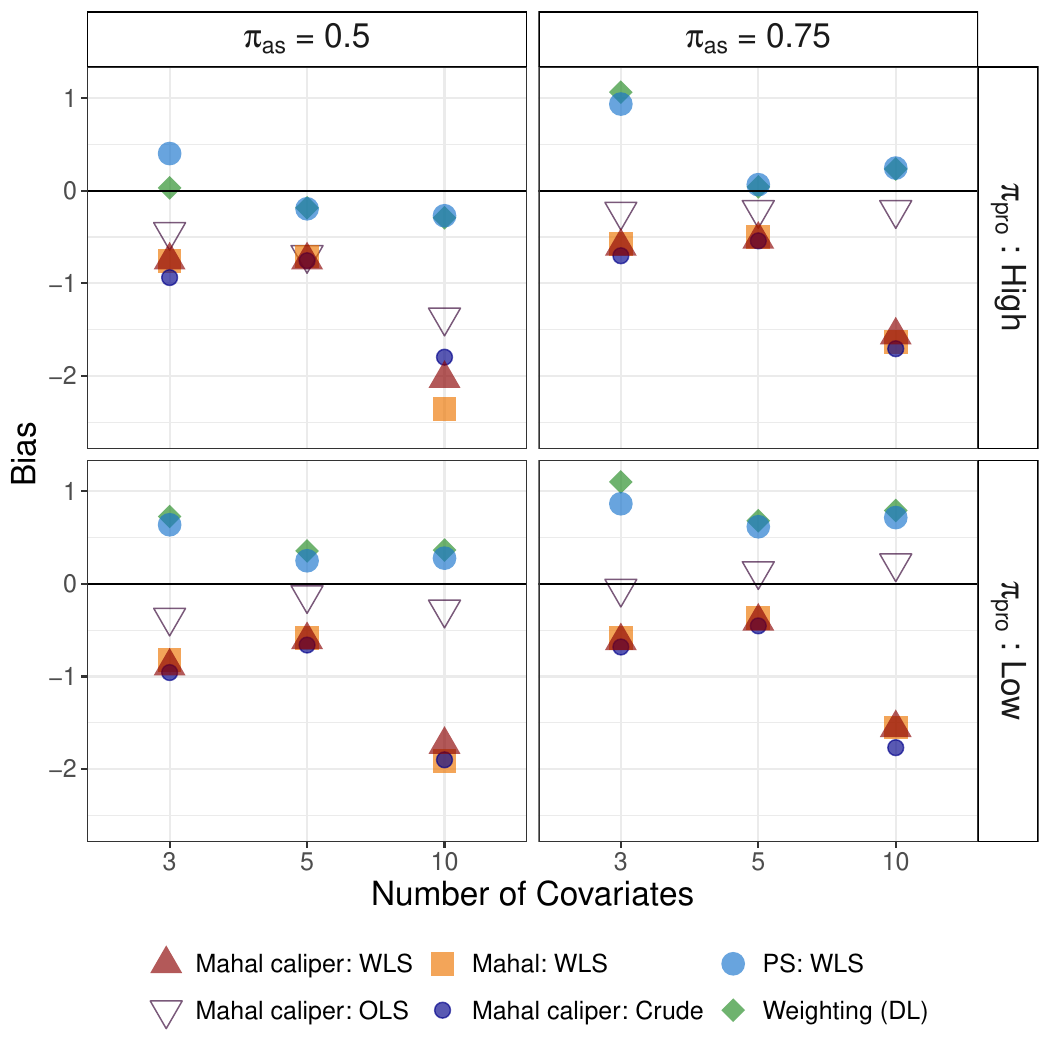}
\end{minipage}
\end{figure}
%%%%%%%%%%%%%%%%%%%%%%%%%%%%%%%%%%%%%%%%%%%%%%%%%%%%%%%%%%%%%

%%%%%%%%%%%%%%%%%%%%%%%%%%%%%%%%%%%%%%%%%%%%%%%%%%%%%%%%%%%%%
\begin{figure}
\centering
\caption{\footnotesize{Bias of different estimators, with $k=5$, for several $\xi$ values, and $\xi_{assm}=\xi$, under correctly specified models (top left), misspecified principal score model (top right), misspecified outcome model (bottom  left), and misspecified principal score model and outcome model (bottom  right). True outcome model did not include $A$-$\bX_0$ interactions. 
Matching was on the Mahalanobis distance without (Mahal) or with a caliper (Mahal caliper), or on $\widehat{\widetilde{\pi}}^1_{as}(\bx_0)$ (PS).
WLS: weighted least squares with interactions;
OLS: ordinary least squares with interactions; 
Crude: Crude mean difference.
Weighting (DL): model-based weighting estimator of \cite{ding2017principal}.
The true SACE parameter is 2 for all $\xi$ values. \label{Fig:AppbiasS2woutInterDGMseq}}} \begin{minipage}{.5\textwidth}
\includegraphics[scale=0.415]{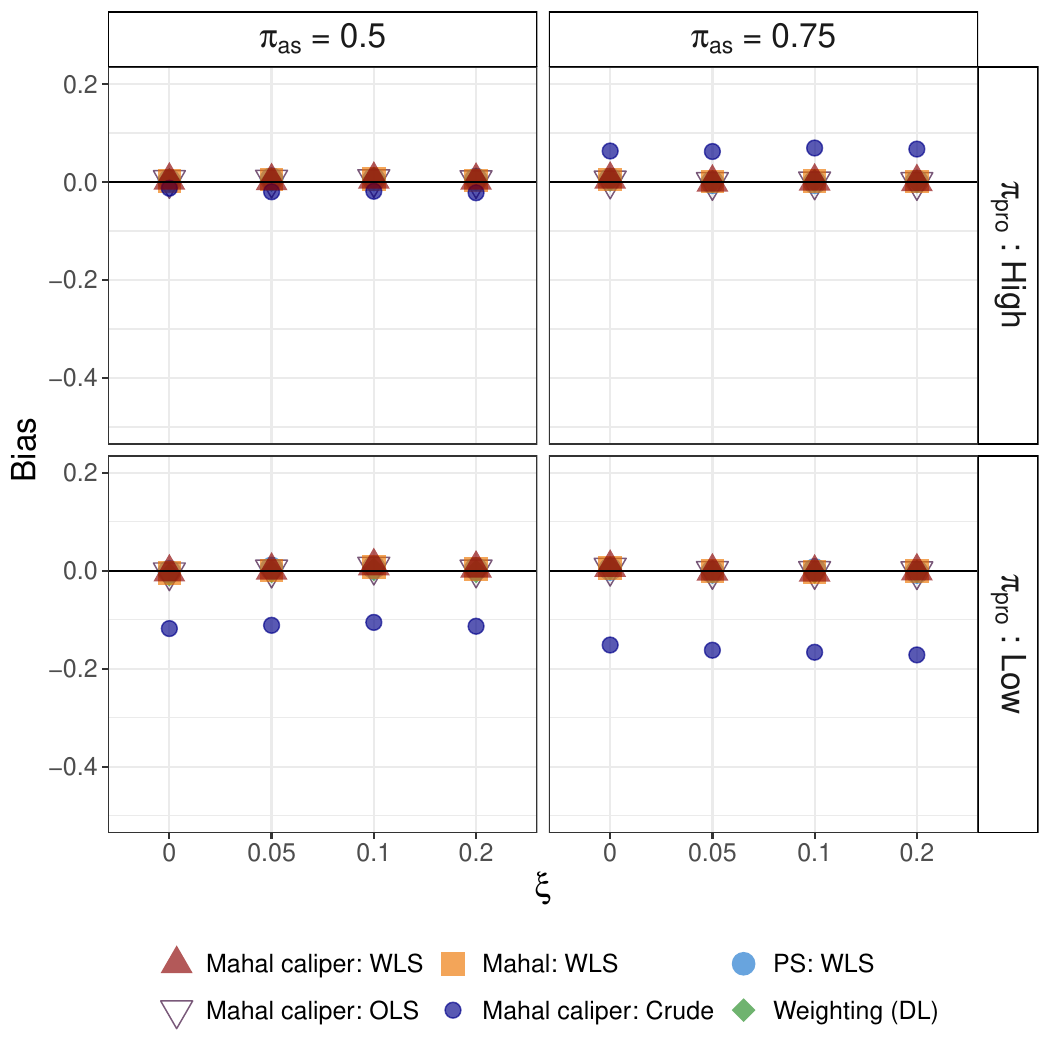}
% general correct wout inter
\end{minipage}%
\begin{minipage}{.475\textwidth}
\centering
\includegraphics[scale=0.415]{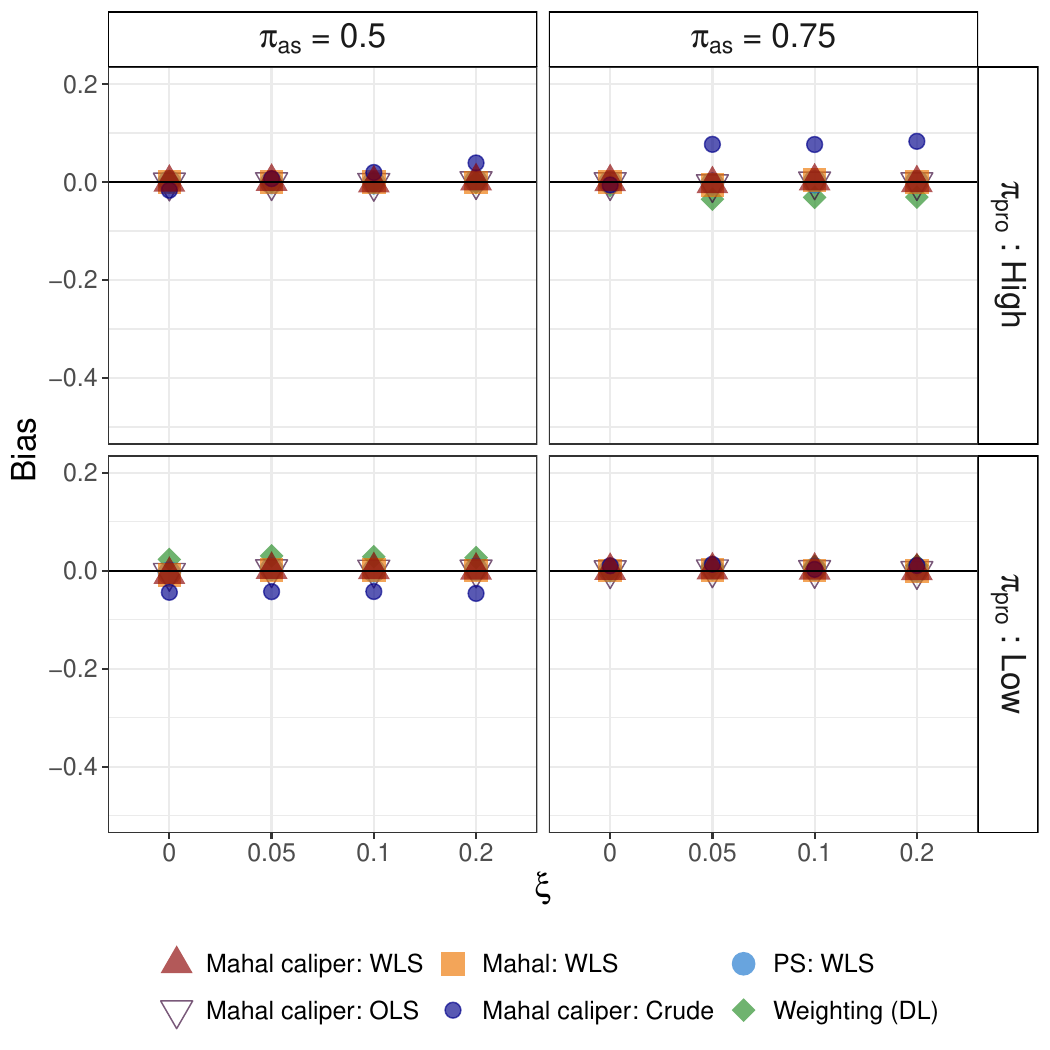}
\end{minipage}

\begin{minipage}{.5\textwidth}
\includegraphics[scale=0.415]{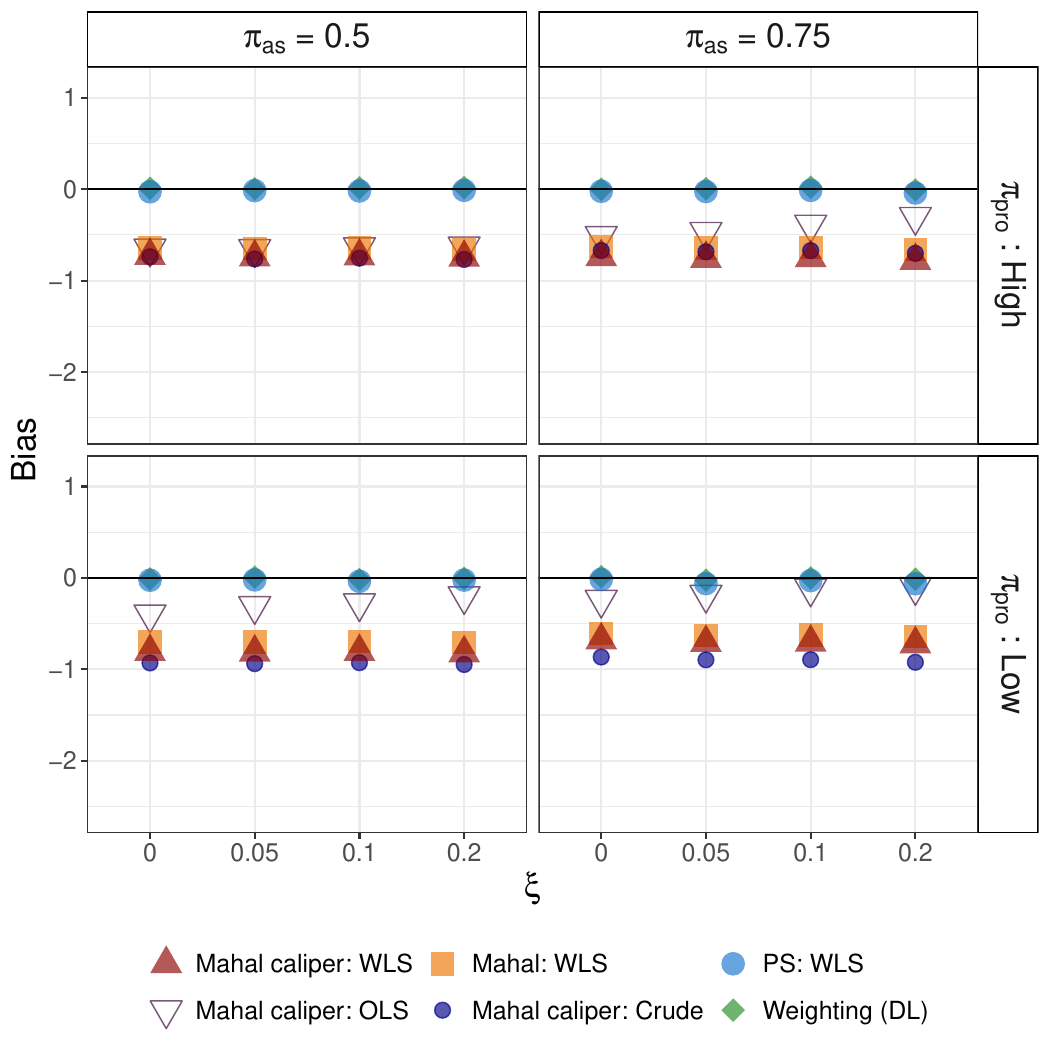}
% general correct wout inter
\end{minipage}%
\begin{minipage}{.475\textwidth}
\centering
\includegraphics[scale=0.415]{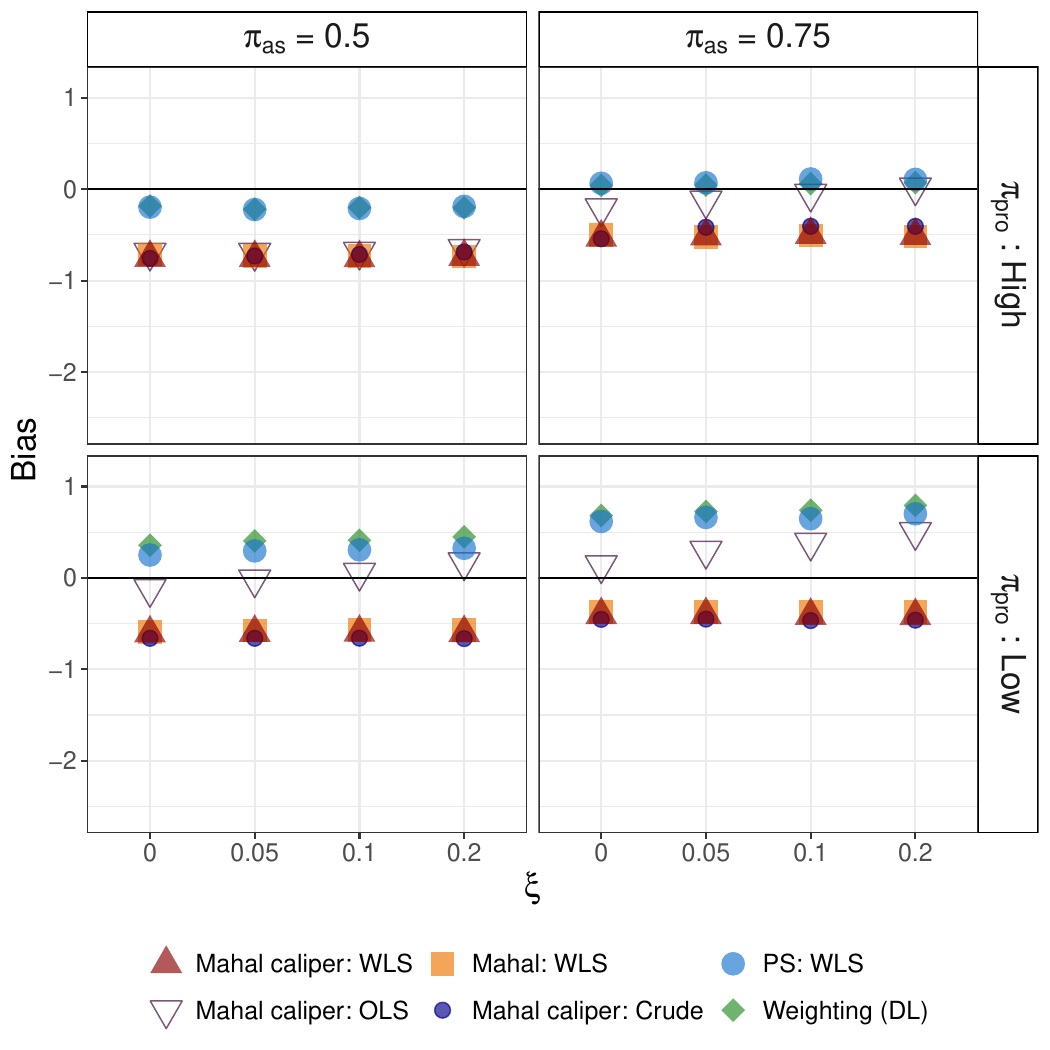}
\end{minipage}
\end{figure}
%%%%%%%%%%%%%%%%%%%%%%%%%%%%%%%%%%%%%%%%%%%%%%%%%%%%%%%%%%%%%

%%%%%%%%%%%%%%%%%%%%%%%%%%%%%%%%%%%%%%%%%%%%%%%%%%%%%%%%%%%%%
\begin{figure}
\centering
\caption{\footnotesize{Bias of different estimators, with $k=5$, for several $\xi$ values, and $\xi_{assm}=0$, under correctly specified models (top left), misspecified principal score model (top right), misspecified outcome model (bottom  left), and misspecified principal score model and outcome model (bottom  right). True outcome model outcome did not include $A$-$\bX_0$ interactions.
Matching was on the Mahalanobis distance without (Mahal) or with a caliper (Mahal caliper), or on $\widehat{\widetilde{\pi}}^1_{as}(\bx_0)$ (PS).
WLS: weighted least squares with interactions;
OLS: ordinary least squares with interactions; 
Crude: Crude mean difference.
Weighting (DL): model-based weighting estimator of \cite{ding2017principal}.
The true SACE parameter is 2 for all $\xi$ values.\label{Fig:AppbiasS3woutInterDGMseq}}} \begin{minipage}{.5\textwidth}
\includegraphics[scale=0.415]{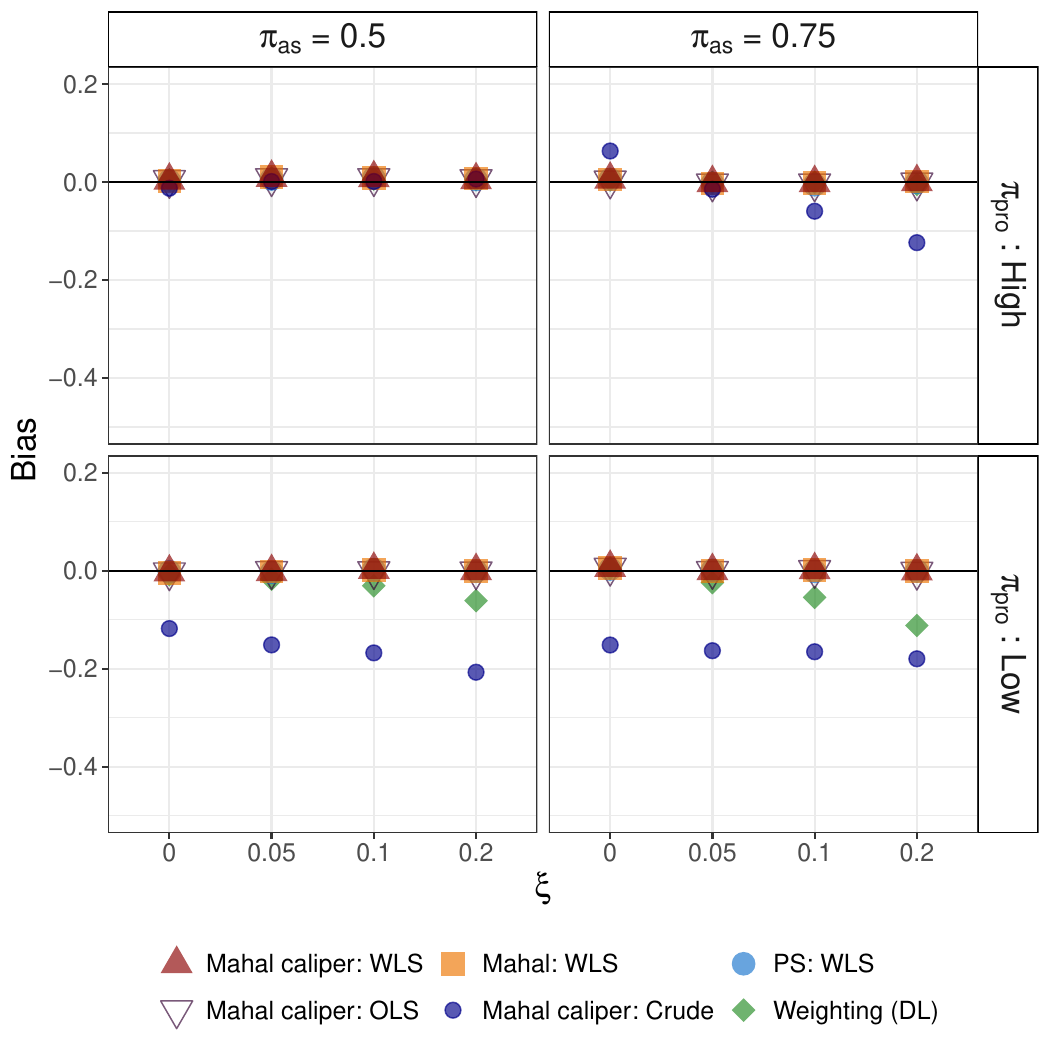}
% general correct wout inter
\end{minipage}%
\begin{minipage}{.475\textwidth}
\centering
\includegraphics[scale=0.415]{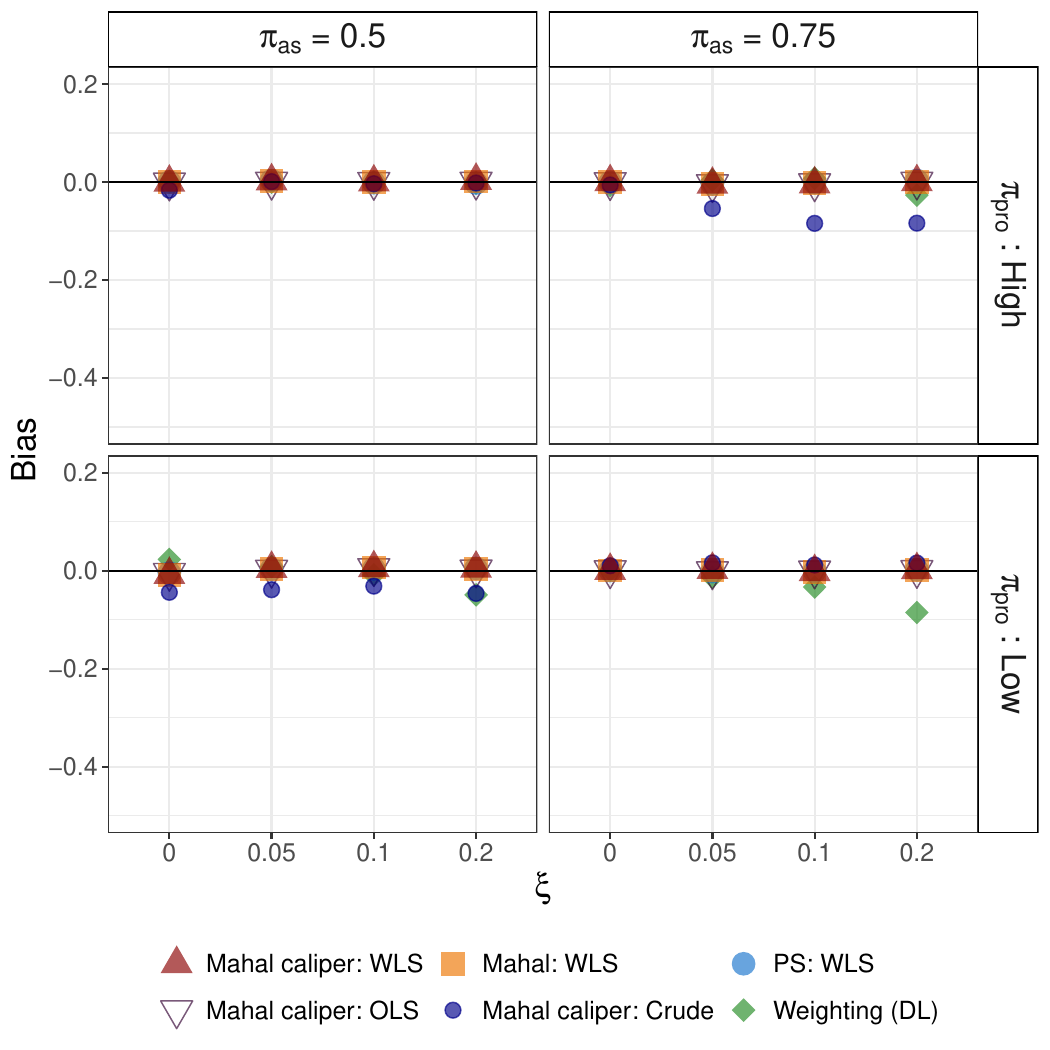}
\end{minipage}

\begin{minipage}{.5\textwidth}
\includegraphics[scale=0.415]{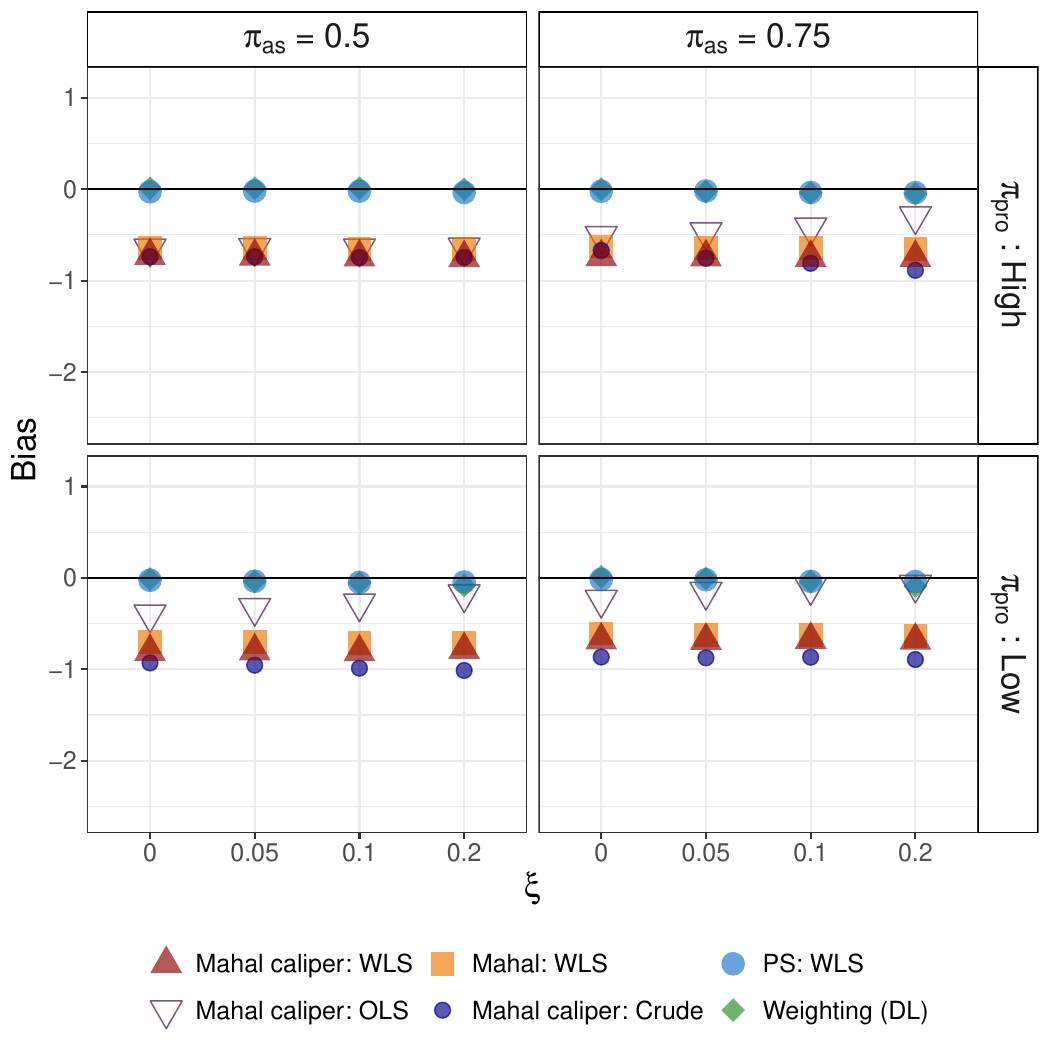}
% general correct wout inter
\end{minipage}%
\begin{minipage}{.475\textwidth}
\centering
\includegraphics[scale=0.415]{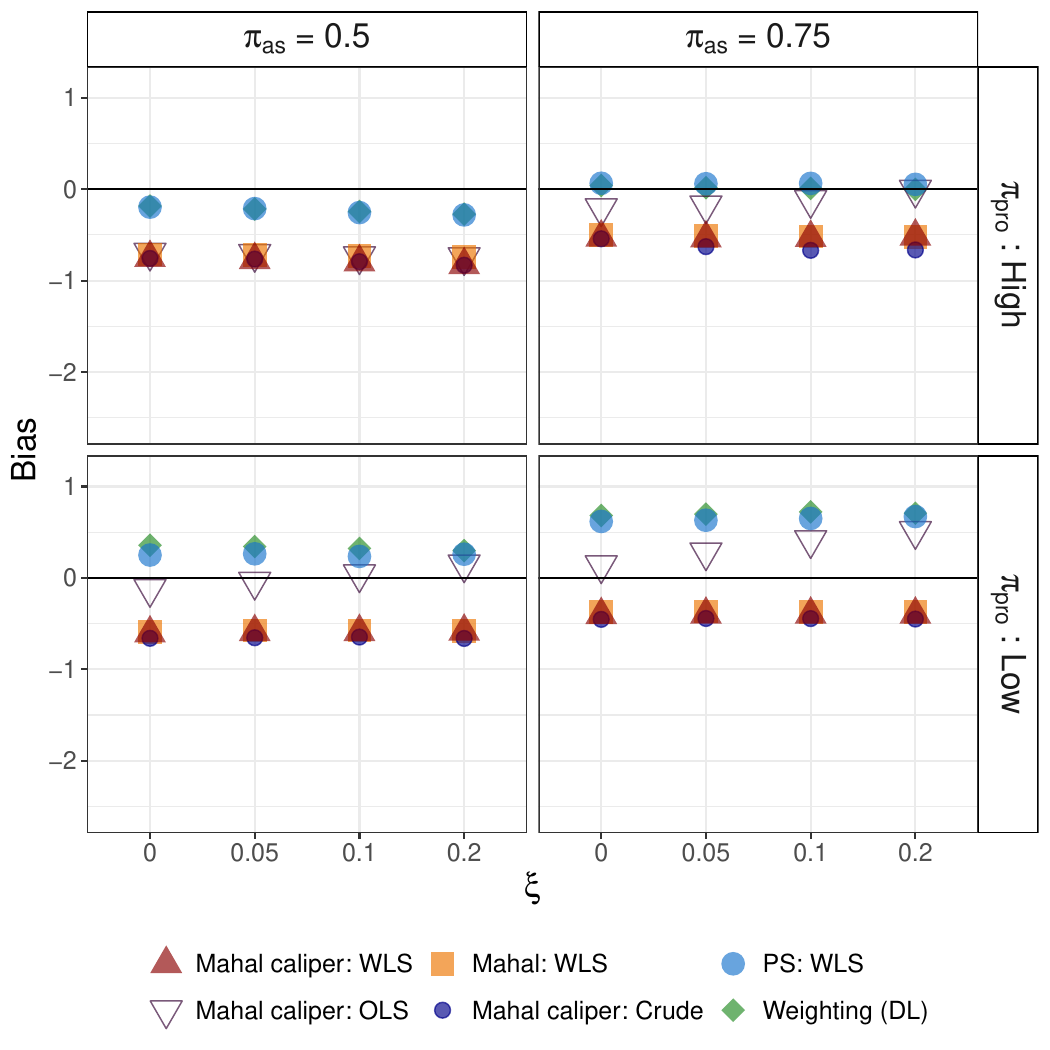}
\end{minipage}
\end{figure}
%%%%%%%%%%%%%%%%%%%%%%%%%%%%%%%%%%%%%%%%%%%%%%%%%%%%%%%%%%%%%
% end new scale

%%%%%%%%%%%%%%%%%%%%%%%%%%%%%%%%%%%%%%%%%%%%%%%%%%%%%%%%%%%%%%%%%%
\begin{figure}
\centering
\caption{\footnotesize{Bias of different estimators, with $k=3, 5, 10$ covariates, under correctly specified (left panel) and misspecified (right panel) principal score model, under a correctly specified  outcome model. The principal strata were generated according to the multinomial logistic regression model. True outcome model included $A$-$\bX_0$ interactions. The true SACE parameter ranged between 3.36 - 3.92 for $k=3$, 4.86 - 5.99 for $k=5$, and 8.29 - 9.48 for $k=10$.
Matching was on the Mahalanobis distance with a caliper.
Matching:Crude - Crude mean difference;
Matching:Regression - weighted least squares with interactions;
Matching:Bias-Corrected - bias-corrected matching estimator \citep{abadie2011bias}.
Weighting - model-based weighting estimator of \cite{ding2017principal}.
\label{Fig:AppbiasS1withInterDGMmulti}}} 
\begin{minipage}{.45\textwidth}
\includegraphics[scale=0.4]{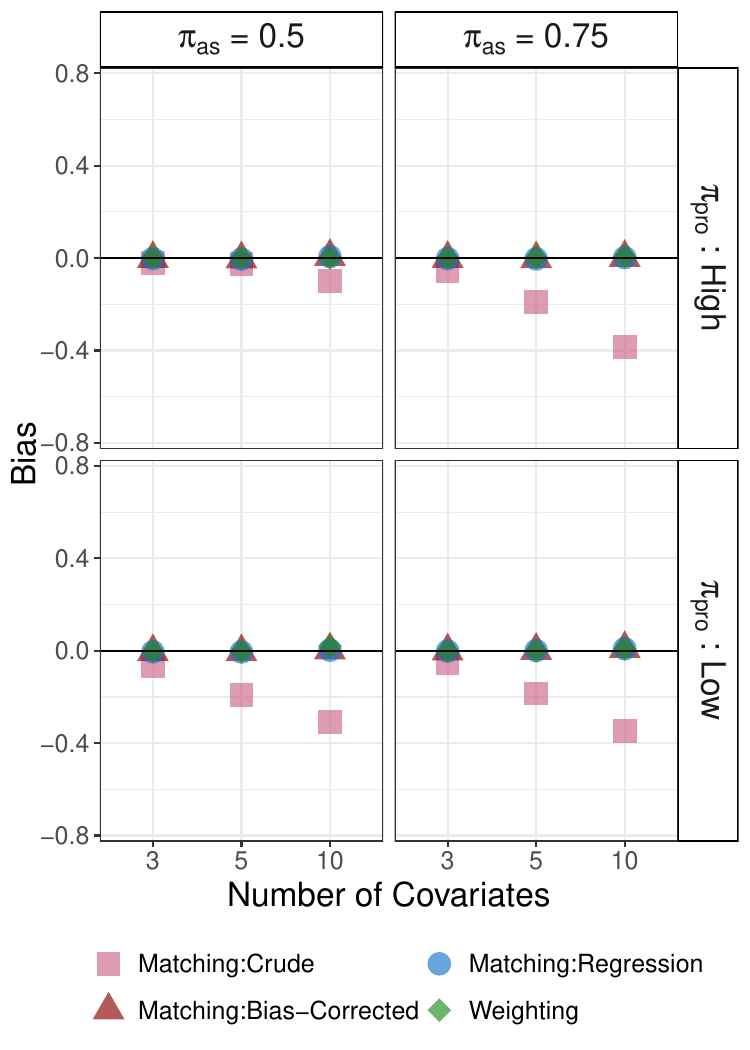}
\end{minipage}%
\begin{minipage}{.45\textwidth}
\centering
\includegraphics[scale=0.4]{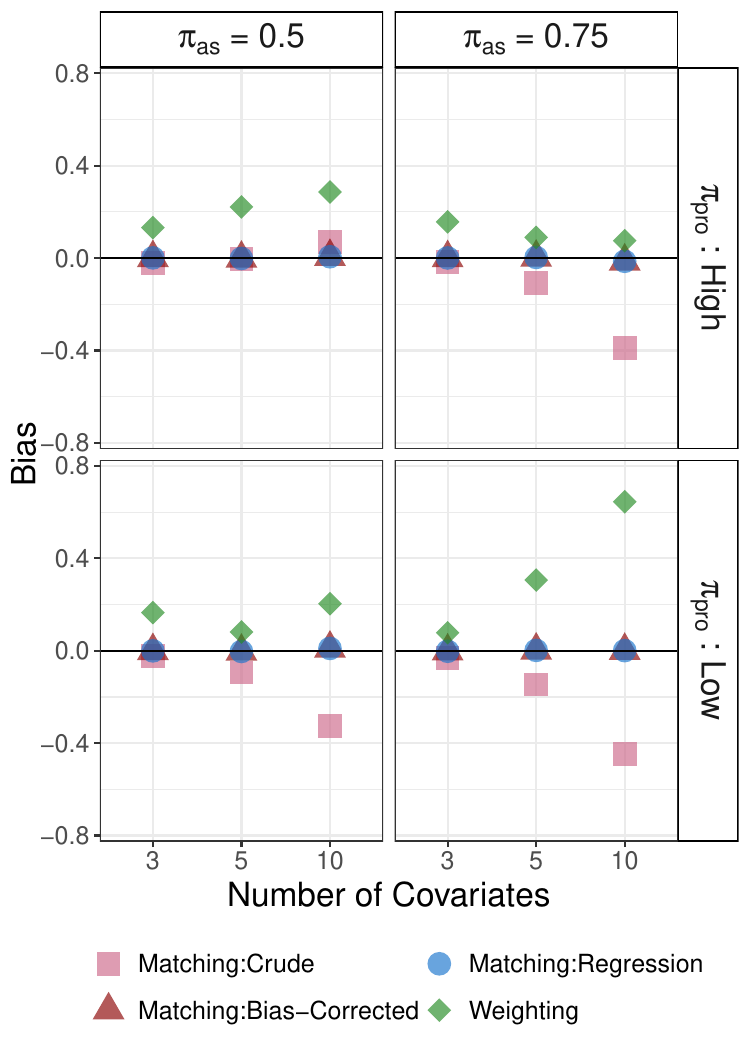}
\end{minipage}
\end{figure}

\clearpage
\newpage

%%%%%%%%%%%%%%%%%%%%%%%%%%%%%%%%%%%%%%%%%%%%%%%%%%%%%%%%%%%%%%%%%%
\section{Data analysis}
\label{Sec:AppData}

This section presents additional details and results for the data analysis presented in Section 8 of the main text. The section includes the following information:

\begin{itemize}

\item Table
\ref{Tab:ObsTrtSurv}
provides the number and proportion of participants within the four values of $\{A=a,S=s\}$.

\item Table \ref{Tab:variables description} provides a description of the variables available from the NSW dataset.

\item Table \ref{Tab:App EM PS coeffs LL} gives the estimated  sequential logistic regression coefficients obtained by the EM algorithm applied to the NSW dataset.

\item Table \ref{Tab:balance_several_caliper} presents balance results under several caliper values.

\item Tables \ref{Tab:AppbalanceSeveralMeasures} and \ref{Tab:AppResultsSeveralMeasures} present the covariates balance and the SACE estimates the NSW dataset, using several choices of distance measures.

\item Figures \ref{Fig:SAAppByEstimator}  and \ref{Fig:SAAppByMeasure} depict the sensitivity curves  for different  estimators and distance measures. The results resemble the results from the primary analysis presented in Figure 3 of the main text. 
\end{itemize}
%%%%%%%%%%%%%%%%%%%%%%%%%%%%%%%%%%%%%%%%%%%%%%%%%%%%%%%%%%%%%%%%%%
\clearpage
\newpage

%%%%%%%%%%%%%%%%%%%%%%%%%%%%%%%%%%%%%%%%%%%%%%%%%%%%%%%%%%%%%%%%%%
\begin{table}
\caption{\label{Tab:ObsTrtSurv}Observed treatment and employment status in the NSW data}
\centering
\fbox{
\begin{tabular}{*{4}{c}}
& $S=0$ & $S=1$ & \em Total \\[0.4em]
\hline
$A=0$ & 129 (18\%) & 296 (41\%) & 425 (59\%) \\
$A=1$ & 67 (9\%) & 230 (32\%) & 297 (41\%) \\
Total & 196 (27\%) & 526 (73\%) & 722 (100\%) %\\[0.4em] 
\end{tabular}}
\end{table}
%%%%%%%%%%%%%%%%%%%%%%%%%%%%%%%%%%%%%%%%%%%%%%%%%%%%%%%%%%%%%%%%%%

\begin{table}
\caption{\label{Tab:variables description} Description of the variables (and their abbreviations) in the NSW dataset}
\centering
\begin{tabular}{|l|l|}
\hline
Variable & Description  \\[0.4em] 
\hline 
Age & Age of participant (in years) \\[0.4em] 
Education (educ) & Number of school years  \\[0.4em] 
Black (blk) & Equals one if race was considered to be black and zero otherwise  \\[0.4em] 
Hispanic (hisp) & Equals one if race was considered to be hispanic and zero otherwise \\[0.4em]
Married (mar) & Marital status: one if married, zero otherwise \\[0.4em] 
Nodegree (nodeg) & Equals one if participant did not hold school degrees and zero otherwise \\[0.4em] 
Re74 & Real earnings in thousands in 1974 (in 1982 dollars)   \\[0.4em]
Re75 & Real earnings in thousands in 1975 (in 1982 dollars) \\[0.4em] 
Emp74 & Employment status in 1974: one if employed in 1974, zero otherwise  \\[0.4em]
Emp75 & Employment status in 1975: one if employed in 1975, zero otherwise \\[0.4em] 
Emp78 (S) & Employment status: one if employed in 1978, zero otherwise \\[0.4em] 
Re78 (Y) & Real earnings in 1978 (in 1982 dollars) \\[0.4em] 
\hline
\end{tabular}
\end{table}
%%%%%%%%%%%%%%%%%%%%%%%%%%%%%%%%%%%%%%%%%%%%%%%%%%%%%%%%%%%%%%%%%%

%%%%%%%%%%%%%%%%%%%%%%%%%%%%%%%%%%%%%%%%%%%%%%%%%%%%%%%%%%%%%%%%%%
\begin{table}
\caption{\label{Tab:App EM PS coeffs LL} Principal score coefficients obtained from the EM algorithm.
Re75 - real earnings in thousands in 1975 (in 1982 dollars); 
Emp75 - employment status in 1975: one if employed in 1975, zero otherwise.}
\centering
\fbox{\begin{tabular}{|rrrrrrrr|}
 \hline
Coefficient & Intercept & Age & Black & Hispanic & Married & Re75 & Emp75  \\[0.05cm] 
 \hline
 \\[0.1cm] 
$\widehat{\bgamma}_{S(0)}$ & 1.91 & -0.02 & -0.81 & 0.29 & -0.22 & 0.06 & 0.10 \\  \\[0.2cm] 
$\widehat{\bgamma}_{S(1)}$ & -4.40 & 0.15 & -2.07 & 1.38 & 3.10 & 0.29 & -4.00 \\ 
\end{tabular}}
%\endgroup
\end{table}
%%%%%%%%%%%%%%%%%%%%%%%%%%%%%%%%%%%%%%%%%%%%%%%%%%%%%%%%%%%%%%%%%%

%%%%%%%%%%%%%%%%%%%%%%%%%%%%%%%%%%%%%%%%%%%%%%%%%%%%%%%%%%%%%%%%%%%%%%%%%%%%%%%%%%%
\begin{table}
\caption{\label{Tab:balance_several_caliper} Covariates balance in the full, employed and the matched dataset. For each dataset, means (SDs) for continuous covariates and frequencies (proportions) for discrete covariates are presented according to treatment status, along with standardized mean differences (SMDs) and standard deviation ratio between the treated and the untretaed (SDRs).
Matching was carried out using the Mahalanobis distance with a caliper on $\widehat{\widetilde{\pi}}^1_{as}(\bx_0)$. 
Caliper size (c) varies between $0.05$ to $0.5 \cdot SD$ of $\widehat{\widetilde{\pi}}^1_{as}(\bx_0)$.
Re75 - real earnings in thousands in 1975 (in 1982 dollars); 
Emp75 - employment status in 1975: one if employed in 1975, zero otherwise.}
%Caliper size (c) varies between $0.05 \cdot SD$ and $0.5 \cdot SD$.}
% $c=\kappa \cdot SD$, where \kappa varies between 0.05 and 0.5.}. 
\centering
\fbox{
\scriptsize
\begin{tabular}{l|cccccccc}

& Untreated & Treated & SMD & SDR & Untreated & Treated & SMD & SDR \\ \\
& \multicolumn{3}{c}{\textbf{Full}} & \multicolumn{3}{c}{\textbf{Employed}}  \\ 
 \hline
 \textbf{Continuous} & & & & & & & & \\ 
Age & 24.4 (6.6) & 24.6 (6.7) & 0.03 & 1.01 & 24.1 (6.6) & 24.6 (6.7) & 0.09 & 1.02 \\ 
Education & 10.2 (1.6) & 10.4 (1.8) & 0.12 & 1.12 & 10.2 (1.6) & 10.4 (1.9) & 0.13 & 1.16 \\ 
Re75 & 3 (5.2) & 3.1 (4.9) & 0.01 & 0.94 & 3.4 (5.7) & 3.3 (5.1) & -0.02 & 0.9 \\ 
 \hline
  \textbf{Discrete} & & & & & & & & \\ 
Black & 340 (80\%) & 238 (80\%) & 0 & 1 & 224 (76\%) & 176 (77\%) & 0.02 & 0.99 \\ 
Hispanic & 48 (11\%) & 28 (9\%) & -0.06 & 0.92 & 41 (14\%) & 25 (11\%) & -0.09 & 0.9 \\ 
Married & 67 (16\%) & 50 (17\%) & 0.03 & 1.03 & 45 (15\%) & 44 (19\%) & 0.11 & 1.1 \\ 
Nodegree & 346 (81\%) & 217 (73\%) & -0.21 & 1.14 & 238 (80\%) & 167 (73\%) & -0.2 & 1.12 \\ 
Emp75 & 247 (58\%) & 186 (63\%) & 0.09 & 0.98 & 182 (61\%) & 150 (65\%) & 0.08 & 0.98 \\ \\ 

& \multicolumn{3}{c}{\textbf{Matched ($0.05 \cdot SD$)}} & \multicolumn{3}{c}{\textbf{Matched ($0.1 \cdot SD$)}} \\
 \hline
 \textbf{Continuous} & & & & & & & & \\ 
Age & 23.5 (6) & 23.5 (5.4) & 0 & 0.91 & 23.7 (6.1) & 23.8 (5.6) & 0.01 & 0.91 \\ 
Education & 10.2 (1.6) & 10.4 (1.4) & 0.11 & 0.87 & 10.2 (1.6) & 10.2 (1.6) & -0.01 & 0.96 \\ 
Re75 & 3.5 (5.7) & 2.4 (3.4) & -0.18 & 0.59 & 3.5 (5.8) & 2.6 (4.2) & -0.15 & 0.72 \\  
   \hline
  \textbf{Discrete} & & & & & & & & \\ 
 Black & 211 (75\%) & 212 (75\%) & 0.01 & 1 & 217 (75\%) & 218 (75\%) & 0.01 & 1 \\ 
  Hispanic & 39 (14\%) & 41 (15\%) & 0.02 & 1.02 & 41 (14\%) & 39 (13\%) & -0.02 & 0.98 \\ 
  Married & 34 (12\%) & 47 (17\%) & 0.14 & 1.14 & 40 (14\%) & 47 (16\%) & 0.07 & 1.07 \\ 
  Nodegree & 226 (80\%) & 226 (80\%) & 0 & 1 & 234 (81\%) & 237 (82\%) & 0.03 & 0.98 \\ 
  Emp75 & 179 (64\%) & 181 (64\%) & 0.01 & 1 & 181 (63\%) & 178 (62\%) & -0.02 & 1.01 \\ \\
  
  & \multicolumn{3}{c}{\textbf{Matched ($0.15 \cdot SD$)}} & \multicolumn{3}{c}{\textbf{Matched ($0.2 \cdot SD$)}}  \\
 \hline
 \textbf{Continuous} & & & & & & & & \\ 
Age & 23.9 (6.5) & 23.9 (5.8) & -0.01 & 0.89 & 24 (6.6) & 24 (6) & 0.01 & 0.91 \\ 
  Education & 10.2 (1.6) & 10.2 (1.5) & 0.03 & 0.95 & 10.2 (1.6) & 10.2 (1.5) & 0.03 & 0.92 \\ 
Re75 & 3.4 (5.7) & 3 (4.9) & -0.09 & 0.85 & 3.4 (5.7) & 3 (5.2) & -0.08 & 0.9 \\  \\ 
   \hline
  \textbf{Discrete} & & & & & & & & \\ 
  Black & 221 (75\%) & 214 (73\%) & -0.06 & 1.03 & 222 (76\%) & 218 (74\%) & -0.03 & 1.02 \\ 
  Hispanic & 41 (14\%) & 41 (14\%) & 0 & 1 & 41 (14\%) & 33 (11\%) & -0.08 & 0.91 \\ 
  Married & 42 (14\%) & 40 (14\%) & -0.02 & 0.98 & 43 (15\%) & 44 (15\%) & 0.01 & 1.01 \\ 
  Nodegree & 237 (81\%) & 241 (82\%) & 0.03 & 0.97 & 237 (81\%) & 240 (82\%) & 0.03 & 0.98 \\ 
 Emp75 & 182 (62\%) & 174 (59\%) & -0.06 & 1.01 & 182 (62\%) & 176 (60\%) & -0.04 & 1.01 \\ \\
  
  & \multicolumn{3}{c}{\textbf{Matched ($0.25 \cdot SD$)}} & \multicolumn{3}{c}{\textbf{Matched ($0.3 \cdot SD$)}}  \\
 \hline
 \textbf{Continuous} & & & & & & & & \\ 
Age & 24.1 (6.6) & 24.2 (6.4) & 0.02 & 0.98 & 24.1 (6.6) & 23.8 (6.1) & -0.04 & 0.93 \\ 
Education & 10.2 (1.6) & 10.2 (1.6) & 0 & 0.98 & 10.2 (1.6) & 10.2 (1.5) & 0.02 & 0.94 \\ 
Re75 & 3.4 (5.7) & 3 (5.1) & -0.08 & 0.89 & 3.4 (5.7) & 3 (5.1) & -0.06 & 0.88 \\ 
\hline
\textbf{Discrete} & & & & & & & & \\ 
Black & 224 (76\%) & 220 (74\%) & -0.03 & 1.02 & 224 (76\%) & 216 (73\%) & -0.06 & 1.04 \\ 
Hispanic & 41 (14\%) & 34 (11\%) & -0.07 & 0.92 & 41 (14\%) & 39 (13\%) & -0.02 & 0.98 \\ 
Married & 45 (15\%) & 46 (16\%) & 0.01 & 1.01 & 45 (15\%) & 53 (18\%) & 0.08 & 1.07 \\ 
Nodegree & 238 (80\%) & 241 (81\%) & 0.03 & 0.98 & 238 (80\%) & 239 (81\%) & 0.01 & 0.99 \\ 
Emp75 & 182 (61\%) & 178 (60\%) & -0.03 & 1.01 & 182 (61\%) & 179 (60\%) & -0.02 & 1 \\  \\
& \multicolumn{3}{c}{\textbf{Matched ($0.35 \cdot SD$)}} & \multicolumn{3}{c}{\textbf{Matched ($0.4 \cdot SD$)}}  \\ 
 \hline
 \textbf{Continuous} & & & & & & & & \\ 
Age & 24.1 (6.6) & 23.9 (6.1) & -0.03 & 0.93 & 24.1 (6.6) & 23.9 (6.1) & -0.02 & 0.93 \\ 
Education & 10.2 (1.6) & 10.2 (1.5) & 0.02 & 0.93 & 10.2 (1.6) & 10.2 (1.5) & 0 & 0.94 \\ 
Re75 & 3.4 (5.7) & 3.1 (5.1) & -0.06 & 0.88 & 3.4 (5.7) & 3.1 (5.2) & -0.05 & 0.91 \\ 
\hline
\textbf{Discrete} & & & & & & & & \\ 
Black & 224 (76\%) & 218 (74\%) & -0.05 & 1.03 & 224 (76\%) & 221 (75\%) & -0.02 & 1.01 \\ 
Hispanic & 41 (14\%) & 38 (13\%) & -0.03 & 0.97 & 41 (14\%) & 36 (12\%) & -0.05 & 0.95 \\ 
Married & 45 (15\%) & 55 (19\%) & 0.09 & 1.08 & 45 (15\%) & 54 (18\%) & 0.08 & 1.08 \\ 
Nodegree & 238 (80\%) & 240 (81\%) & 0.02 & 0.99 & 238 (80\%) & 241 (81\%) & 0.03 & 0.98 \\ 
Emp75 & 182 (61\%) & 180 (61\%) & -0.01 & 1 & 182 (61\%) & 182 (61\%) & 0 & 1 \\
& & & & & & & & \\ 
 & \multicolumn{3}{c}{\textbf{Matched ($0.45 \cdot SD$)}} & \multicolumn{3}{c}{\textbf{Matched ($0.5 \cdot SD$)}}  \\ 
 \hline
 \textbf{Continuous} & & & & & & & & \\ 
Age & 24.1 (6.6) & 23.9 (6.1) & -0.03 & 0.92 & 24.1 (6.6) & 23.8 (6.1) & -0.04 & 0.92 \\ 
Education & 10.2 (1.6) & 10.2 (1.6) & 0 & 0.96 & 10.2 (1.6) & 10.2 (1.6) & -0.02 & 0.96 \\ 
Re75 & 3.4 (5.7) & 3.1 (5.2) & -0.05 & 0.91 & 3.4 (5.7) & 3.1 (5.2) & -0.05 & 0.91 \\
   \hline
  \textbf{Discrete} & & & & & & & & \\
  Black & 224 (76\%) & 221 (75\%) & -0.02 & 1.01 & 224 (76\%) & 222 (75\%) & -0.02 & 1.01 \\ 
  Hispanic & 41 (14\%) & 37 (12\%) & -0.04 & 0.96 & 41 (14\%) & 36 (12\%) & -0.05 & 0.95 \\ 
  Married & 45 (15\%) & 55 (19\%) & 0.09 & 1.08 & 45 (15\%) & 55 (19\%) & 0.09 & 1.08 \\ 
  Nodegree & 238 (80\%) & 242 (82\%) & 0.03 & 0.97 & 238 (80\%) & 243 (82\%) & 0.04 & 0.97 \\ 
  Emp75 & 182 (61\%) & 181 (61\%) & -0.01 & 1 & 182 (61\%) & 181 (61\%) & -0.01 & 1 \\ 
\end{tabular}}
\end{table}
%%%%%%%%%%%%%%%%%%%%%%%%%%%%%%%%%%%%%%%%%%%%%%%%%%%%%%%%%%%%%%%%%%%%%%%%%%%%%%%%%%%

\begin{table}
\caption{\label{Tab:AppbalanceSeveralMeasures} Covariates balance in matched datasets, using three distance measures: \newline  $\widehat{\widetilde{\pi}}^1_{as}(\bx_0)$, Mahalanobis distance
and Mahalanobis distance with a caliper on $\widehat{\widetilde{\pi}}^1_{as}(\bx_0)$ ($c=0.4 \cdot SD$), \newline
with and without replacement.
For each dataset, means (SDs) for continuous covariates and frequencies (proportions) for discrete covariates are presented according to treatment status, along with standardized mean differences (SMDs) and standard deviation ratio between the treated and the untretaed (SDRs). Re75 - real earnings in thousands in 1975 (in 1982 dollars); 
Emp75 - employment status in 1975: one if employed in 1975, zero otherwise.} 
\centering
\fbox{
\scriptsize
\begin{tabular}{c|c@{\hspace{1\tabcolsep}}c@{\hspace{1\tabcolsep}}c@{\hspace{1\tabcolsep}}c@{\hspace{2\tabcolsep}}c@{\hspace{1\tabcolsep}}c@{\hspace{1\tabcolsep}}c@{\hspace{1\tabcolsep}}c@{\hspace{2\tabcolsep}}c@{\hspace{1\tabcolsep}}c@{\hspace{1\tabcolsep}}c@{\hspace{1\tabcolsep}}c}
\\ [0.025cm]
 & \multicolumn{4}{c}{$\widehat{\widetilde{\pi}}^1_{as}(\bx_0)$} & \multicolumn{4}{c}{Mahalanobis} 
 & \multicolumn{4}{c}{Mahalanobis with $\widehat{\widetilde{\pi}}^1_{as}(\bx_0)$ caliper} \\ [0.1cm]
& Untreated & Treated & SMD & SDR & Untreated & Treated & SMD & SDR & Untreated & Treated & SMD & SDR \\ 
 \hline \\ [0.025cm]
& \multicolumn{12}{c}{\textbf{With replacement}} \\ [0.15cm]
N & 296 & 296 & &  & 296 & 296 & & & 296 & 296 & & \\ [0.15cm] 
\textbf{Continuous} & & & & & & & & & & & & \\[0.1cm]
Age & 24.1 (6.6) & 23.7 (6) & -0.05 & 0.9 & 24.1 (6.6) & 24 (6.2) & -0.01 & 0.94 & 24.1 (6.6) & 23.9 (6.2) & -0.02 & 0.93 \\ 
Education & 10.2 (1.6) & 10.3 (1.8) & 0.1 & 1.08 & 10.2 (1.6) & 10.2 (1.5) & 0 & 0.95 & 10.2 (1.6) & 10.2 (1.5) & 0 & 0.94 \\ 
Re75 & 3.4 (5.7) & 2.5 (3.7) & -0.16 & 0.65 & 3.4 (5.7) & 3.2 (5.3) & -0.03 & 0.93 & 3.4 (5.7) & 3.1 (5.2) & -0.05 & 0.91 \\[0.1cm]
% \hline
\textbf{Discrete} & & & & & & & & & \\[0.1cm] 
Black & 224 (76\%) & 216 (73\%) & -0.06 & 1.04 & 224 (76\%) & 228 (77\%) & 0.03 & 0.98 & 224 (76\%) & 221 (75\%) & -0.02 & 1.01 \\ 
Hispanic & 41 (14\%) & 47 (16\%) & 0.06 & 1.06 & 41 (14\%) & 42 (14\%) & 0.01 & 1.01 & 41 (14\%) & 36 (12\%) & -0.05 & 0.95 \\ 
Married & 45 (15\%) & 46 (16\%) & 0.01 & 1.01 & 45 (15\%) & 52 (18\%) & 0.07 & 1.06 & 45 (15\%) & 54 (18\%) & 0.08 & 1.08 \\ 
Nodegree & 238 (80\%) & 223 (75\%) & -0.13 & 1.09 & 238 (80\%) & 241 (81\%) & 0.03 & 0.98 & 238 (80\%) & 241 (81\%) & 0.03 & 0.98 \\
Emp75 & 182 (61\%) & 174 (59\%) & -0.06 & 1.01 & 182 (61\%) & 178 (60\%) & -0.03 & 1.01 & 182 (61\%) & 182 (61\%) & 0 & 1 \\  &\\[0.1cm]
\hline \\ [0.025cm]

& \multicolumn{12}{c}{\textbf{Without replacement}} \\ [0.1cm]
\textbf{Continuous} & & & & & & & & & & & & \\[0.1cm]
N & 230 & 230 & &  & 230 & 230 & & & 223 & 223 & &\\ 
Age & 25.3 (6.9) & 24.6 (6.7) & -0.1 & 0.98 & 25.3 (6.9) & 24.6 (6.7) & -0.1 & 0.98 & 25 (6.6) & 24.1 (6.1) & -0.14 & 0.93 \\
Education & 10.5 (1.6) & 10.4 (1.9) & -0.06 & 1.2 & 10.5 (1.6) & 10.4 (1.9) & -0.06 & 1.2 & 10.5 (1.6) & 10.4 (1.8) & -0.07 & 1.13 \\ 
Re75 & 4.1 (6.3) & 3.3 (5.1) & -0.12 & 0.82 & 4.1 (6.3) & 3.3 (5.1) & -0.12 & 0.82 & 4.2 (6.3) & 3.3 (5.1) & -0.14 & 0.81  \\[0.1cm]
% \hline
\textbf{Discrete} & & & & & & & & & \\[0.1cm] 
Black & 170 (74\%) & 176 (77\%) & 0.06 & 0.97 & 170 (74\%) & 176 (77\%) & 0.06 & 0.97 & 163 (73\%) & 171 (77\%) & 0.08 & 0.95 \\ 
Hispanic & 29 (13\%) & 25 (11\%) & -0.05 & 0.94 & 29 (13\%) & 25 (11\%) & -0.05 & 0.94 & 29 (13\%) & 25 (11\%) & -0.05 & 0.94 \\ 
Married & 45 (20\%) & 44 (19\%) & -0.01 & 0.99 & 45 (20\%) & 44 (19\%) & -0.01 & 0.99 & 45 (20\%) & 40 (18\%) & -0.06 & 0.96 \\ 
Nodegree & 172 (75\%) & 167 (73\%) & -0.05 & 1.03 & 172 (75\%) & 167 (73\%) & -0.05 & 1.03 & 165 (74\%) & 163 (73\%) & -0.02 & 1.01 \\ 
Emp75 & 152 (66\%) & 150 (65\%) & -0.02 & 1.01 & 152 (66\%) & 150 (65\%) & -0.02 & 1.01 & 151 (68\%) & 148 (66\%) & -0.03 & 1.01 \\[0.1cm]
%\hline
\end{tabular}}
\end{table}

\begin{table}
\caption{\label{Tab:AppResultsSeveralMeasures} Estimated SACE (CI95\%) using matching under different distance measures and different estimators.} 
\centering
\fbox{
\scriptsize
\begin{tabular}{lcccc}
Distance measure/Estimator & Crude & Regression & Regression interactions & BC  \\[0.1cm] 
\hline
\\[0.05cm]
 Mahalanobis & 312 (-729, 1353) & 319 (-1026, 1663) & 232 (-1080, 1545) & 329 (-1146, 1804) \\ 
\\[0.25cm]
$\widehat{\widetilde{\pi}}^1_{as}(\bx_0)$ & 506 (-694, 1706) & 642 (-1257, 2542) & 467 (-1439, 2373) & 231 (-1048, 1510)  \\ [0.25cm] 
Mahalanobis with \\ $\widehat{\widetilde{\pi}}^1_{as}(\bx_0)$ caliper & 59 (-977, 1094) & 114 (-1226, 1453) & 55 (-1269, 1379) & 68 (-1376, 1512) \\ \\[0.1cm]
\end{tabular}}
\end{table}

\begin{figure}
\centering
\small\caption{Sensitivity analyses for the SACE estimation in the NSW data, using several estimators. The left panel presents the SACE estimate as a function of $\alpha_1$, without assuming PPI. The right panel presents SACE estimate as a function of $\xi$ and $\alpha_0$ without assuming  monotonicity.
\label{Fig:SAAppByEstimator}
}
\begin{minipage}{.5\textwidth}
\includegraphics[scale=0.32]{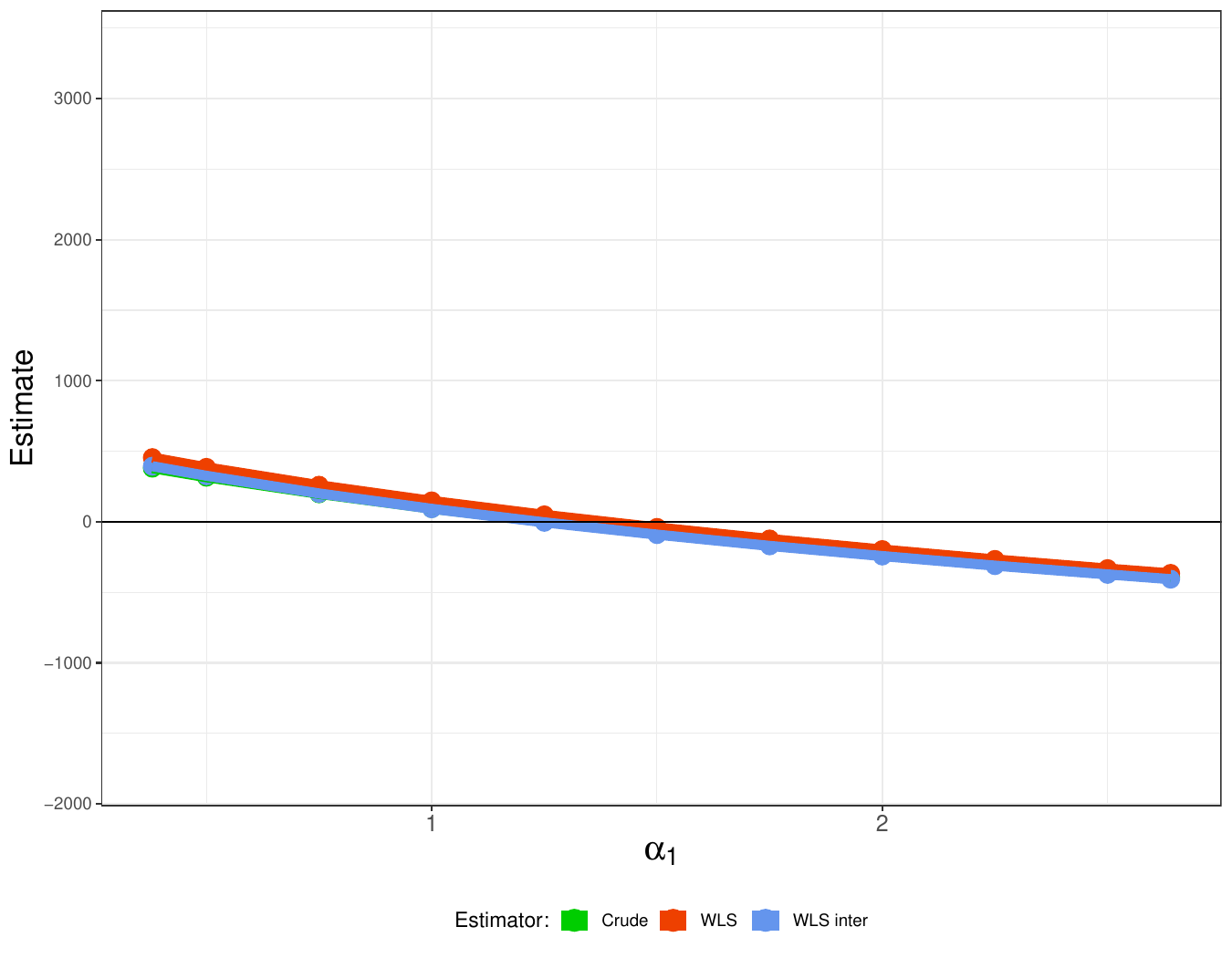}
\end{minipage}%
\begin{minipage}{.5\textwidth}
\centering
\includegraphics[scale=0.32]{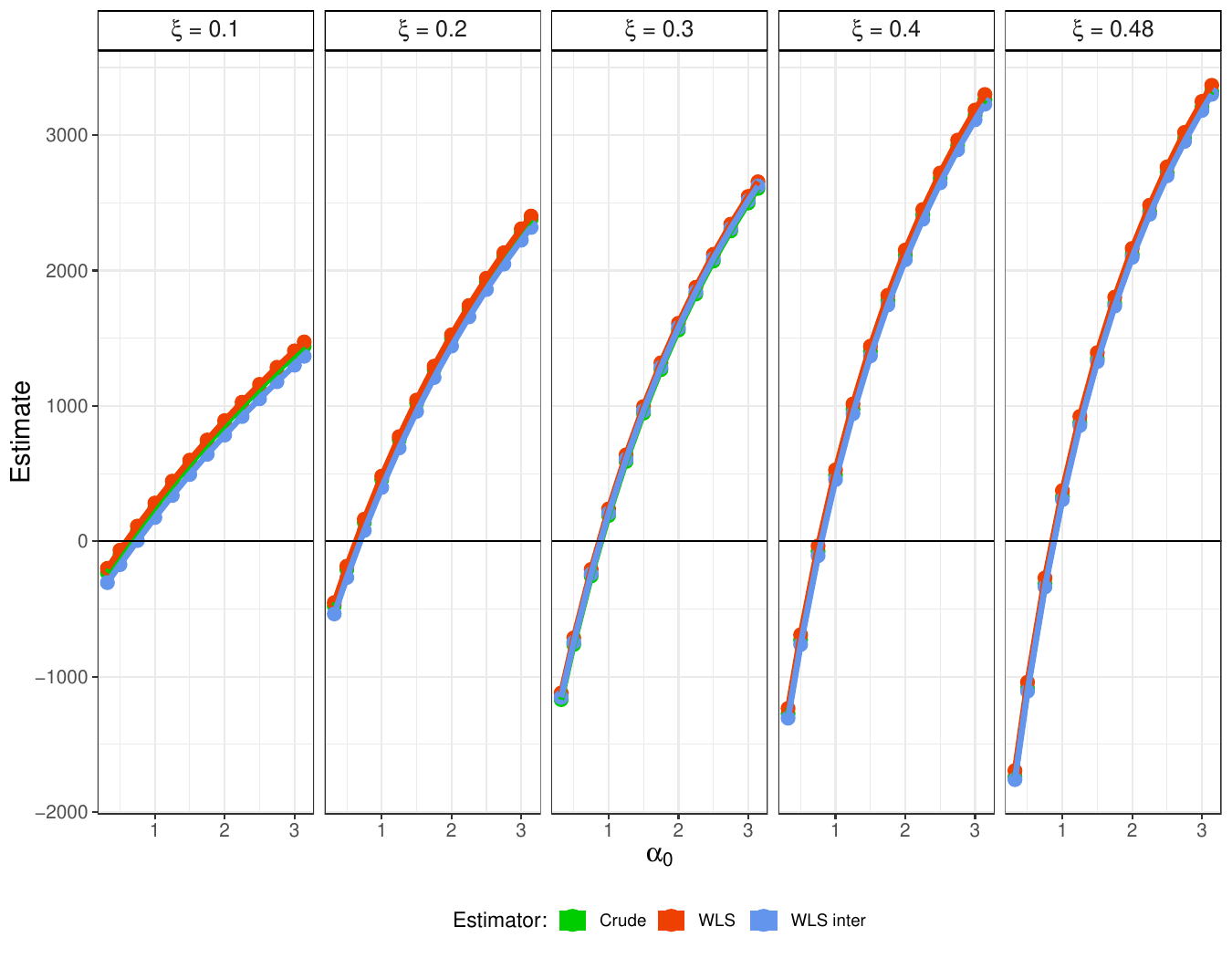}
\end{minipage}
\end{figure}

\begin{figure}
\centering
\small\caption{Sensitivity analyses for the SACE estimation in the NSW data, using several distance measures. The left panel presents the SACE estimate as a function of $\alpha_1$, without assuming PPI. The right panel presents SACE estimate as a function of $\xi$ and $\alpha_0$ without assuming  monotonicity.
\label{Fig:SAAppByMeasure}
}
\begin{minipage}{.5\textwidth}
\includegraphics[scale=0.32]{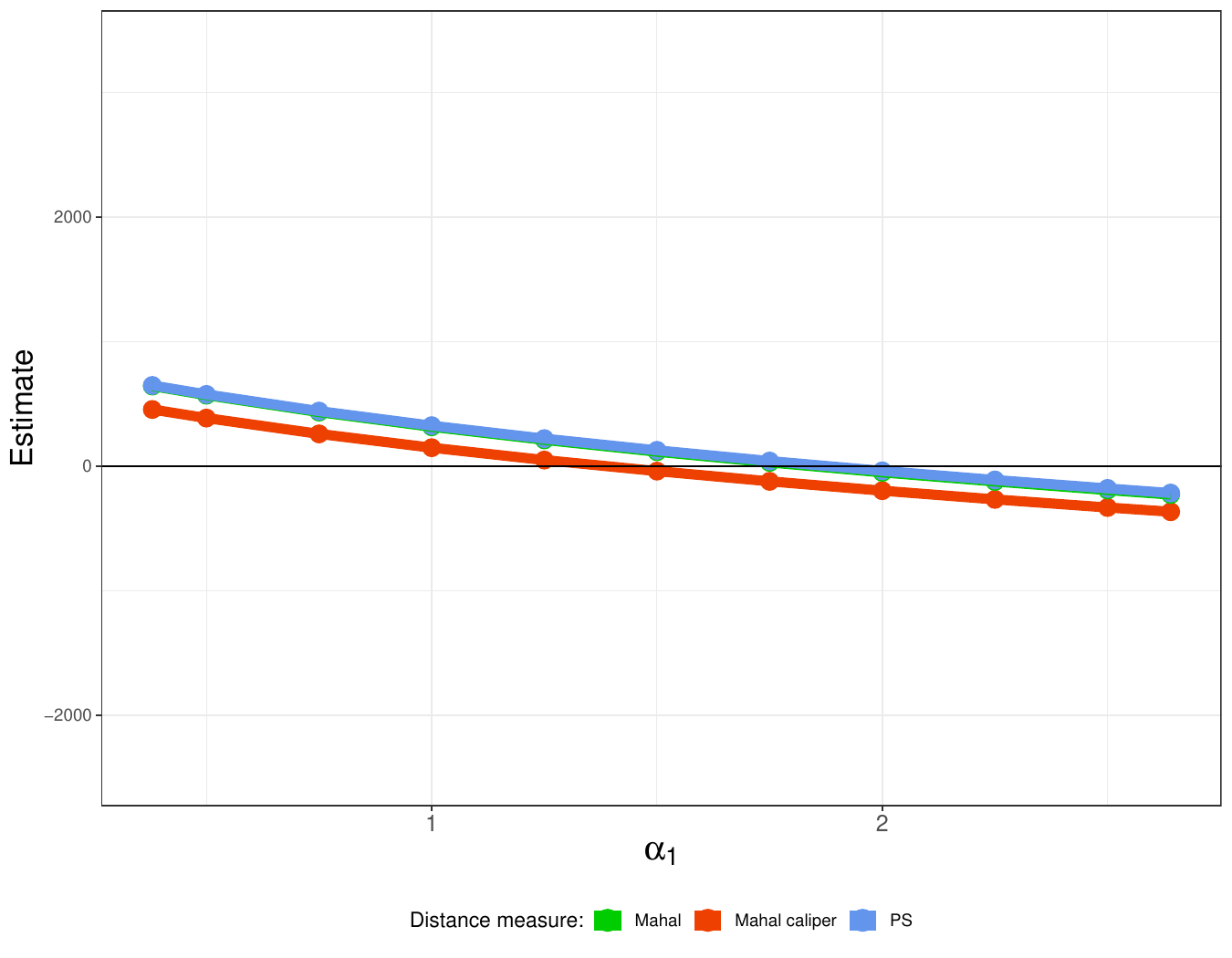}
\end{minipage}%
\begin{minipage}{.5\textwidth}
\centering
\includegraphics[scale=0.32]{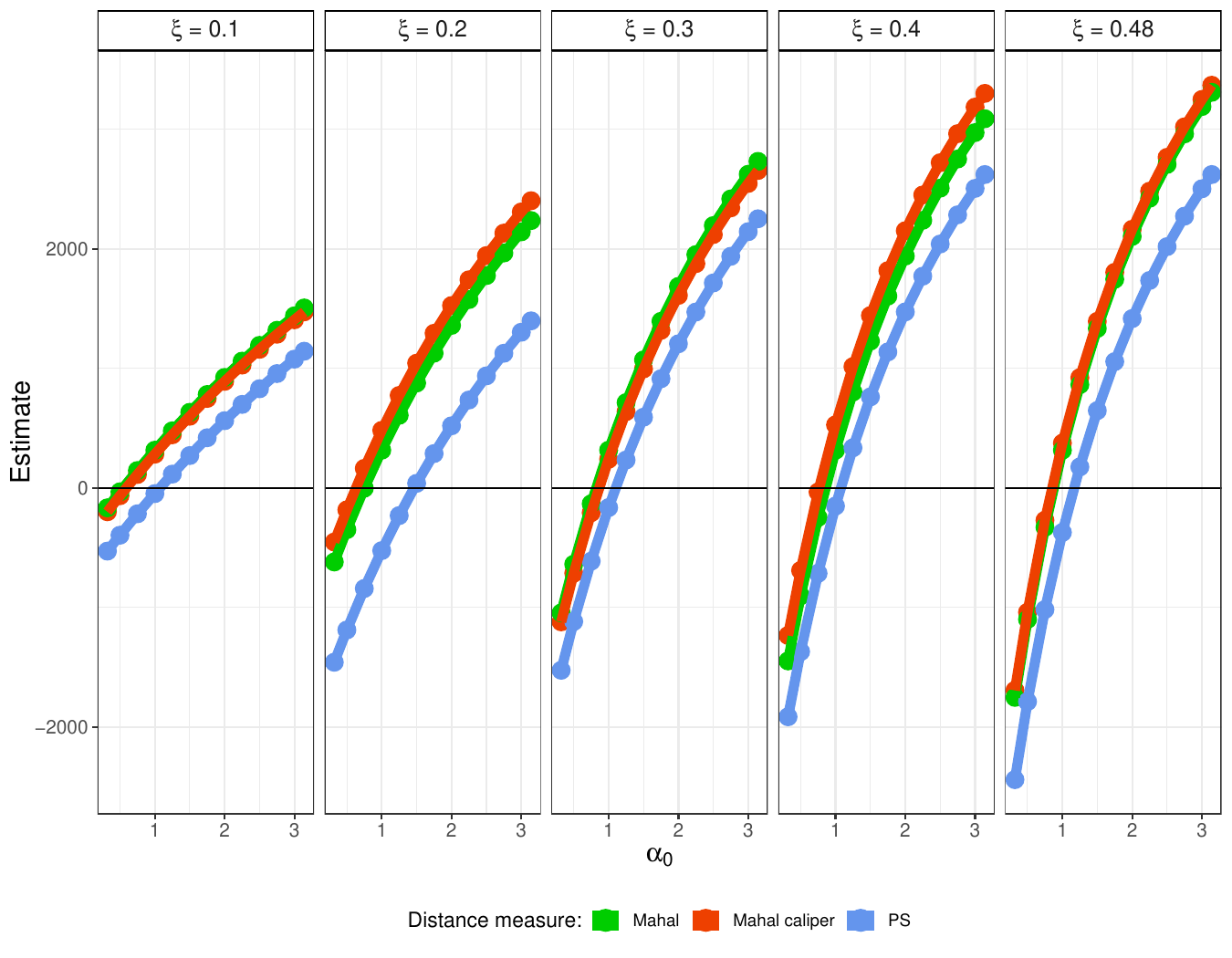}
\end{minipage}
\end{figure}

\end{document}